\documentclass[prx,aps,superscriptaddress,twocolumn,notitlepage,floatfix,10pt]{revtex4-2}
\usepackage{amsmath,mathtools,amsthm,amssymb,pifont}
\usepackage{mathrsfs}
\usepackage[utf8]{inputenc}
\usepackage[american]{babel}
\usepackage{graphicx,xcolor,bbold}
\usepackage{braket}
\usepackage{MnSymbol}
\usepackage{color}
\usepackage[colorlinks,
citecolor=red,
linkcolor=red,
urlcolor=red]{hyperref}
\usepackage{orcidlink}

\newcommand{\Tr}{\mathrm{Tr}}
\newcommand{\Fmu}{\mathcal{F}_{\boldsymbol\mu}}
\newcommand{\nuup}{\boldsymbol\nu^{\uparrow}}
\newcommand{\nudown}{\boldsymbol\nu^{\downarrow}}
\newcommand{\Ftwo}{\mathbb F_2}
\newcommand{\PW}{P_{W}}
\newcommand{\PiF}{\Pi_{\mathrm F}}
\newcommand{\cF}{c_{\mathrm F}}
\newcommand{\Fpm}{\mathcal F_{\mathrm P}}
\newcommand{\Afr}{\mathfrak A}
\newcommand{\Qa}{Q_{(1^4)}}
\newcommand{\Qsym}{Q_{(4)}}
\newcommand{\Qtt}{Q_{(2,2)}}
\newcommand{\Gm}{G_{\boldsymbol\mu}}

\usepackage{physics}
\usepackage{tikz}
\usetikzlibrary{calc,tikzmark,arrows.meta}

\newtheorem{theorem}{Theorem}
\newtheorem{proposition}{Proposition}
\newtheorem{lemma}{Lemma}
\newtheorem{corollary}{Corollary}
\newtheorem{conjecture}{Conjecture}
\theoremstyle{remark}

\begin{document}
\title{Spectral geometry of nonlocal stabilizer entropy}

\author{Piotr Sierant~\orcidlink{0000-0001-9219-7274}}
\email{piotr.sierant@bsc.es}
\affiliation{Barcelona Supercomputing Center Plaça Eusebi G\"uell, 1-3 08034, Barcelona, Spain}

\begin{abstract}
Magic, or nonstabilizerness, is the resource that promotes stabilizer operations to universal quantum computation.
For bipartite pure states, its component intrinsic to the correlations between the subsystems, the \emph{nonlocal magic}, is obtained by minimizing a magic measure over local unitaries.
For the stabilizer R\'enyi entropy (SRE), the minimum is conjectured to be attained by the computational-basis (CB) representative, the state obtained by assigning the Schmidt coefficients in decreasing order to matching computational-basis labels.
We prove this conjecture for two families of states at every system size and bipartition: states with dyadic-staircase Schmidt spectra and states of Schmidt rank at most six. 
For arbitrary bipartite pure states, we show that the CB value exceeds the nonlocal SRE by at most $4$, fixing the leading term of any divergent scaling of nonlocal SRE.
We further show that the nonlocal SRE grows at most logarithmically with the entanglement entropy.
Consequently, one-dimensional area-law states have bounded nonlocal SRE, while critical states with logarithmic entanglement entropy permit at most doubly logarithmic growth with system size.
We apply these results to the transverse-field Ising chain, where we tightly bound the nonlocal SRE in the gapped phases and identify its double-logarithmic growth at criticality.
\end{abstract}

\maketitle

\section{Introduction}
\label{sec:intro}
Entanglement and magic quantify complementary obstructions to classical descriptions of quantum systems.  
Entanglement enables tasks that local operations and classical communication alone cannot accomplish~\cite{Chitambar2019Resource, Horodecki2009Entanglement}, whereas magic, or nonstabilizerness, promotes stabilizer computation to universality~\cite{Bravyi2005magic,Veitch2014,Howard2017}. Neither resource implies the other: Clifford circuits can create highly entangled stabilizer states that remain classically tractable~\cite{Gottesman1998heis,Aaronson2004Simulation}, while a product of nonstabilizer states may contain extensive magic resources without any entanglement between qubits.

Entanglement has long been the organizing perspective on quantum many-body systems.  Scaling of the entanglement entropy with subsystem size separates the physically relevant corner of Hilbert space from the generic one: gapped ground states obey area laws~\cite{Srednicki1993,Hastings2007,Eisert2010,Ge2016AreaLaws,Anshu2022Spread} whose logarithmic violation at criticality measures the central charge of the underlying conformal field theory~\cite{Osterloh2002,Vidal2003,CalabreseCardy2004}. The structure of entanglement encodes phase information inaccessible to local order parameters: universal subleading constants and the entanglement spectrum diagnose topological order~\cite{KitaevPreskill2006,LevinWen2006,Li2008,Zhang2012MES}, and entanglement growth separates thermalizing from localized dynamics~\cite{Abanin2019, Sierant2025MBL}. 

Nonstabilizerness became accessible to investigations of many-body systems through the stabilizer R\'enyi entropy (SRE)~\cite{Leone2022SRE}, which is computed directly from expectation values of Pauli operators~\cite{Sierant2026Computing,Huang2025XOR,Xiao2026Sampling}.  
This direct formulation avoids the costly optimizations involved in other nonstabilizerness measures such as the relative entropy of magic, robustness of magic, and stabilizer extent~\cite{Veitch2014, Howard2017,Bravyi2019LowRank, Liu2022ManyBody}, whose evaluation is limited to a few qubits~\cite{Hamaguchi2024Handbook, Hamaguchi2025Faster}.
The SRE is a magic monotone for pure states~\cite{Haug2023Monotones, Leone2024Monotones} with operational meaning in stabilizer testing and fidelity estimation~\cite{Bittel2026Operational, Leone2023Fidelity, Bu2025Testing}, and can be efficiently estimated by tensor-network, Monte Carlo, and quantum algorithms~\cite{Lami23perfect, Tarabunga2024PauliMPS,
Haug23quanti, Haug2024Algorithms, Ding2025Sampling}.  
This tractability has opened magic to systematic many-body study, revealing structure in conformal ground states~\cite{White2021CFT,Hoshino2026CFT}, quantum critical points~\cite{Tarabunga2023ManyBody,Sarkar2020Critical,Tarabunga2024RokhsarKivelson}, and chaotic dynamics~\cite{Goto2022Chaos,Leone2021QuantumChaos,Turkeshi2025Spreading,Turkeshi2025Pauli} beyond that captured by entanglement alone.
Its dynamics have been studied after quenches and during thermalization~\cite{Sewell2022Thermalization, Rattacaso2023Quench, Tirrito25anti, Sierant2026magic}, as well as across monitored circuit transitions~\cite{Leone2024PhaseTransition, Niroula2024MagicTransition, Oliviero2021Measurements, Bejan2024Dynamical}, and recent work further delineated behavior of magic in many-body systems~\cite{Gu2025Separation,Zhang2024Advantage, Smith25pxp, Falcao25mbl,   Hallam2026Spectral, Sun2026TFD, Ebner2026Barrier, Zhang2026CoupledSYK, Li2026Stark, Li2026Fibonacci, Xiao2026Diffusive, Paviglianiti26true, Cao2026Sudden, Aditya26mbe, Jha2026su2, Salazar26frame, Bhakuni26free, Aditya26sprea, Sarma2026universal}.

Local basis changes can remove part of the magic of a bipartite state without changing its correlations across the bipartition. Magic measures evaluated in a fixed local basis do not distinguish this removable contribution from the magic intrinsic to those correlations.
\emph{Nonlocal magic} isolates the latter by minimizing a magic measure over all local unitaries~\cite{Cao2025magical, Qian2025Nonlocal, Sierant26exact, Viscardi26nonlocal}.   
  By construction, nonlocal magic is invariant under local unitaries: it is the magic that neither party can create or remove by acting on its own subsystem.
Nonlocal magic is the relevant quantity for magic-state distillation~\cite{Bao2022Distillation}, operator entanglement~\cite{Andreadakis2026Operator}, and holography, where it governs area operators, gravitational backreaction, and Schwinger pair production~\cite{Cao2023Area, Cao2025magical, Grieninger2026Holographic}.  It has been computed for high-energy scattering and decays~\cite{Robin2026Scattering, Low2026Wigner, Gargalionis2026Scattering, Banacki2026Higgs, Kadam26nuclear}, dense neutrino gases~\cite{Hite2026Neutrino}, resource harvesting from quantum fields~\cite{Cepollaro2025Harvesting}, and in works on gauge theories and gluon amplitudes~\cite{Chen2026SU2, Grieninger2026StringBreaking, Chu2026PhaseIndependent}. Nonlocal magic has also been measured on a superconducting quantum processor~\cite{Ahmad2025Experimental}.

The minimization defining the nonlocal SRE reintroduces a difficult optimization problem, now over a group of local unitaries whose dimension grows exponentially with subsystem size. Bounds based on entanglement-spectrum flatness, results for random states, and exact solutions for small systems constrain the nonlocal SRE without identifying a general optimizer~\cite{Cao2025magical,Tirrito2024Flatness,  Qian2025Nonlocal, Iannotti2025Random}.
A natural candidate for the optimizer is the sorted computational-basis Schmidt representative state~\cite{Torre2026Spectrum, Liu2026Schmidt, Franchini2026Schmidt, Huang2026Spectral}, obtained by placing the Schmidt coefficients in descending order on the binary labels of the computational basis.  The computational-basis (CB) optimality conjecture asserts that this representative maximizes stabilizer purity, equivalently minimizes the SRE, on the local-unitary orbit.  Analytic results cover a single qubit on one side, where the Schmidt rank is at most two~\cite{Qian2025Nonlocal, Franchini2026Schmidt}, and the conjecture is open in general.

In this work, we develop an analytic framework for the nonlocal SRE of bipartite pure states.
We reduce the optimization over two local unitaries exactly to one over a single unitary and express the stabilizer purity through an operator compression onto a stabilizer subspace.
Using this reduction, we prove the CB optimality conjecture for two families of spectra at every system size and every bipartition: dyadic-staircase spectra~\footnote{Here, dyadic refers to the powers of two. We group the computational-basis labels $x=0,1,\dots,d-1$ into the shells $\{0\}$, $\{1\}$, $\{2,3\}$, $\{4,\dots,7\}$, and so on, where the $k$-th shell consists of the $2^{k-1}$ labels with exactly $k$ binary digits. A dyadic-staircase spectrum is a Schmidt spectrum that is constant on each shell.} and spectra of Schmidt rank at most six.
For an arbitrary spectrum, we show that the nonlocal SRE is fixed, up to bounded constants, by the largest Schmidt probability of each dyadic shell, and that the CB value exceeds it by at most four, so that the conjecture holds up to an additive constant and all scaling laws derived from the CB value are unconditional.
The same construction bounds the nonlocal SRE by $3\log_{2}S_{1}+O(1)$ in terms of the entanglement entropy $S_{1}$ of the cut, which is exponentially stronger than the bounds available so far~\cite{Cao2025magical, Torre2026Spectrum} and confines the nonlocal SRE of area-law, critical, and volume-law states to $O(1)$, $O(\log\log N)$, and $O(\log N)$, respectively.
We apply these results to the transverse-field Ising chain, for which the nonlocal SRE diverges as $\log_{2}\ln N$ at the critical point and reaches finite plateaus in the gapped phases.

The remainder of this paper is organized as follows.
Section~\ref{sec:nlm} defines the nonlocal SRE and the CB optimality conjecture and gives an overview of our results.
In Sec.~\ref{sec:reduction}, we reduce the minimization over local unitaries to a single unitary.
Sections~\ref{sec:dyadic} and~\ref{sec:rank6} prove CB optimality for dyadic-staircase spectra and for spectra of rank at most six.
In Sec.~\ref{sec:witness}, we derive the bounds on the nonlocal SRE of an arbitrary state, and Sec.~\ref{sec:allrank} applies them to many-body systems.
Section~\ref{sec:conclusions} contains our conclusions and an outlook, and the details of proofs are given in the appendices.

\section{Overview}
\label{sec:nlm}

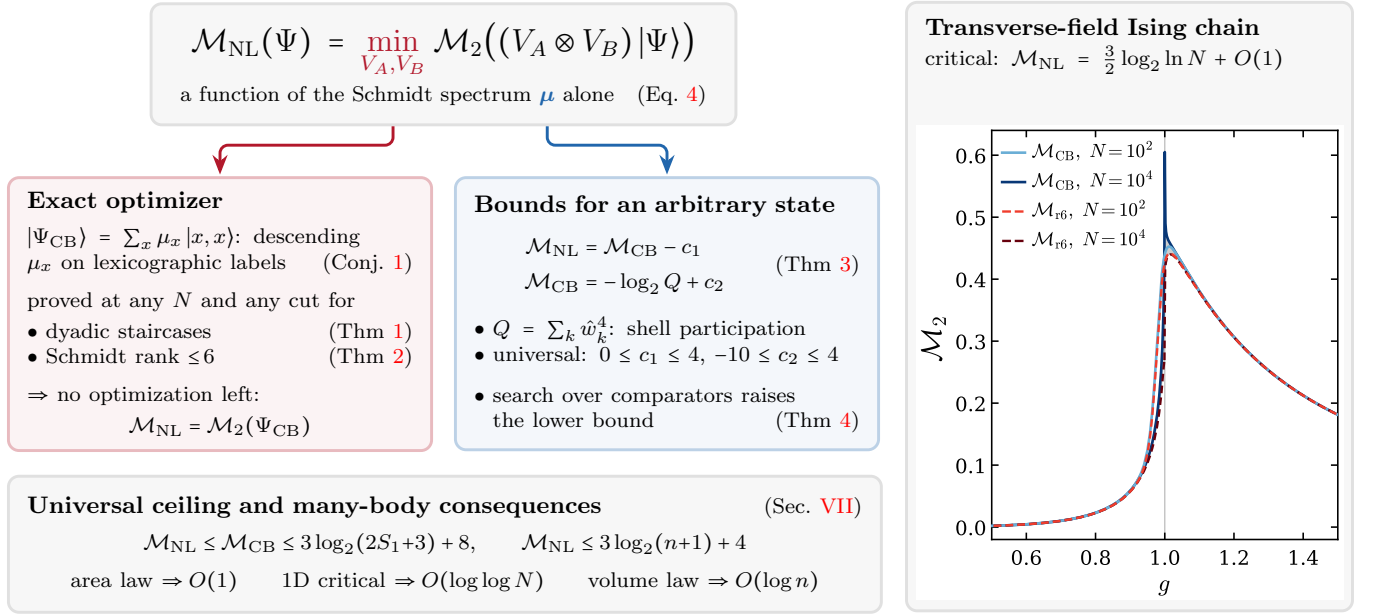
\begin{figure*}[t]
\centering
\begingroup
\definecolor{ovred}{HTML}{B2182B}
\definecolor{ovblue}{HTML}{2166AC}
\colorlet{ovredfg}{ovred!30}
\colorlet{ovredbg}{ovred!5}
\colorlet{ovbluefg}{ovblue!30}
\colorlet{ovbluebg}{ovblue!5}
\colorlet{ovgrayfg}{black!12}
\colorlet{ovgraybg}{black!3}
\def\ovlw{11.55}      
\def\ovgap{0.35}      
\def\ovrx{11.90}      
\def\ovtw{17.80}      
\def\ovh{8.05}        
\def\ovrowb{-2.30}    
\def\ovbranchh{3.60}  
\def\ovrowc{-6.25}    
\def\ovapph{1.80}     
\def\ovbtw{5.106}     
\def\ovatw{11.056}    
\tikzset{ovbox/.style={rounded corners=4pt, line width=1.1pt, align=left,
                       inner xsep=7pt, inner ysep=5pt}}
\begin{tikzpicture}[remember picture]
\node[ovbox, draw=ovgrayfg, fill=ovgraybg, inner xsep=11pt, inner ysep=7pt,
      anchor=north, align=center] (master) at (0.5*\ovlw,0) {\large
  $\mathcal M_{\mathrm{NL}}(\Psi)\;=\;
   \tikzmarknode{ovsnopt}{{\color{ovred}\displaystyle\min_{V_{A},V_{B}}}}\;
   \mathcal M_{2}\bigl((V_{A}\otimes V_{B})\ket{\Psi}\bigr)$\\[2pt]
  {\footnotesize a function of the Schmidt spectrum
   $\tikzmarknode{ovsnmu}{{\color{ovblue}\boldsymbol\mu}}$ alone
   \ \ (Eq.~\ref{eq:MNLsre})}};
\node[ovbox, draw=ovredfg, fill=ovredbg, anchor=north west,
      text width=\ovbtw cm, minimum height=\ovbranchh cm]
      (exact) at (0,\ovrowb) {%
  {\small\bfseries Exact optimizer}\\[4pt]
  {\footnotesize
   $\ket{\Psi_{\mathrm{CB}}}=\sum_x\mu_x\ket{x,x}$: descending\\
   $\mu_x$ on lexicographic labels\hfill(Conj.~\ref{conj:cb})\\[5pt]
   proved at any $N$ and any cut for\\[2pt]
   $\bullet$ dyadic staircases\hfill(Thm~\ref{thm:resDyadic})\\
   $\bullet$ Schmidt rank ${\le}\,6$\hfill(Thm~\ref{thm:resRank6})\\[5pt]
   $\Rightarrow$ no optimization left:\\[2pt]
   \mbox{}\hfill$\mathcal M_{\mathrm{NL}}
    =\mathcal M_{2}(\Psi_{\mathrm{CB}})$\hfill\mbox{}\par}};
\node[ovbox, draw=ovbluefg, fill=ovbluebg, anchor=north east,
      text width=\ovbtw cm, minimum height=\ovbranchh cm]
      (cert) at (\ovlw,\ovrowb) {%
  {\small\bfseries Bounds for an arbitrary state}\\[6pt]
  {\footnotesize
   \mbox{}\hfill$\begin{aligned}
     \mathcal M_{\mathrm{NL}}&=\mathcal M_{\mathrm{CB}}-c_{1}\\[1pt]
     \mathcal M_{\mathrm{CB}}&=-\log_{2}Q+c_{2}
   \end{aligned}$\hfill(Thm~\ref{thm:resValue})\\[7pt]
   $\bullet$ $Q=\sum_{k}\hat w_{k}^{4}$: shell participation\\
   $\bullet$ universal: $0\le c_{1}\le4$, $-10\le c_{2}\le4$\\[6pt]
   $\bullet$ search over comparators raises\\
   \phantom{$\bullet$ }the lower bound\hfill(Thm~\ref{thm:resComparator})\par}};
\coordinate (ovlevel) at ($(master.south)+(0,-0.25)$);
\tikzset{ovarrow/.style={-{Stealth[length=2.6mm,width=2.0mm]}, line width=1.1pt,
                         rounded corners=3pt, shorten >=0.6pt}}
\draw[ovarrow, ovred]
  (ovsnopt.south |- master.south) -- (ovsnopt.south |- ovlevel)
  -- (exact.north |- ovlevel) -- (exact.north);
\draw[ovarrow, ovblue]
  (ovsnmu.south |- master.south) -- (ovsnmu.south |- ovlevel)
  -- (cert.north |- ovlevel) -- (cert.north);
\node[ovbox, draw=ovgrayfg, fill=ovgraybg, anchor=north west,
      text width=\ovatw cm, minimum height=\ovapph cm]
      (ceiling) at (0,\ovrowc) {%
  {\small\bfseries Universal ceiling and many-body consequences}\hfill
   {\footnotesize(Sec.~\ref{sec:allrank})}\\[4pt]
  {\footnotesize
   \mbox{}\hfill$\mathcal M_{\mathrm{NL}}\le\mathcal M_{\mathrm{CB}}
    \le3\log_{2}(2S_{1}{+}3)+8$,\qquad
   $\mathcal M_{\mathrm{NL}}\le3\log_{2}(n{+}1)+4$\hfill\mbox{}\\[5pt]
   \mbox{}\hfill area law $\Rightarrow O(1)$\hfill
   1D critical $\Rightarrow O(\log\log N)$\hfill
   volume law $\Rightarrow O(\log n)$\hfill\mbox{}\par}};
\draw[draw=ovgrayfg, fill=ovgraybg, rounded corners=4pt, line width=1.1pt]
  (\ovrx,0) rectangle (\ovtw,-\ovh);
\node[anchor=north west, text width=5.45cm, align=left, inner sep=0pt]
  at ($(\ovrx,0)+(0.24,-0.20)$) {%
  {\small\bfseries Transverse-field Ising chain}\\[2pt]
  {\footnotesize critical: $\mathcal M_{\mathrm{NL}}
   =\tfrac32\log_{2}\ln N+O(1)$\par}};
\node[anchor=south, inner sep=0pt]
  at ($(0.5*\ovrx+0.5*\ovtw,-\ovh)+(0,0.14)$)
  {\includegraphics[width=5.65cm]{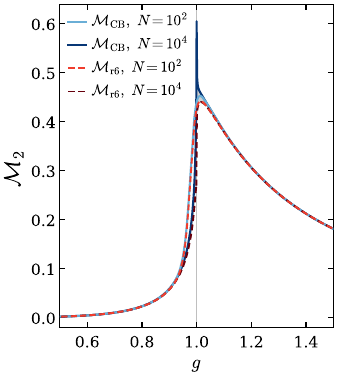}};
\end{tikzpicture}%
\endgroup
\caption{Overview of the results. The nonlocal SRE $\mathcal M_{\mathrm{NL}}$
minimizes the stabilizer R\'enyi-2 entropy over local unitaries,
Eq.~\eqref{eq:MNLsre}. \emph{Red branch}: 
we prove that the minimization is solved exactly by $\ket{\Psi_{\mathrm{CB}}}$ at
any system size and any cut for states $\ket{\Psi}$ with dyadic-staircase spectra and at Schmidt rank at most six (Theorems~\ref{thm:resDyadic} and~\ref{thm:resRank6}). \emph{Blue
branch}: for an arbitrary spectrum, and without assuming
Conjecture~\ref{conj:cb}, we prove that the nonlocal SRE is fixed up to universal constants
by the order-four participation $Q$ of the Schmidt weight across dyadic rank
shells (Theorem~\ref{thm:resValue}); the constant $c_{1}$ is the gap, bounded by $4$, between $\mathcal M_{\mathrm{CB}}$ and $\mathcal M_{\mathrm{NL}}$; a comparator search sharpens the lower bound for an individual spectrum
(Theorem~\ref{thm:resComparator}). \emph{Lower box}: nonlocal SRE is bounded from below by the logarithm of the entanglement
entropy. \emph{Right}: bounds for nonlocal SRE in the transverse-field
Ising chain. Solid curves show the CB value $\mathcal M_{\mathrm{CB}} \geq \mathcal{M}_{\mathrm{NL}}$ across
the transverse field $g$ for chains of $N=10^{2}$ and $10^{4}$ sites; dashed
curves show the rank-six comparator lower bound $\mathcal
M_{\mathrm{r6}}\le\mathcal M_{\mathrm{NL}}$ of
Theorem~\ref{thm:resComparator}, which is tight away from criticality; the shaded window between the two bounds
contains $\mathcal M_{\mathrm{NL}}$. At criticality
$\mathcal M_{\mathrm{NL}}=\frac32\log_{2}\ln N+O(1)$
(Sec.~\ref{ssec:critical}).}
\label{fig:overview}
\end{figure*}

\subsection{Nonlocal magic and the CB optimality}
We consider a pure state $\ket{\Psi}$ of $N=n_{A}+n_{B}$ qubits, divided into subsystems $A$ and $B$ of $n_{A}$ and $n_{B}$ qubits, with local Hilbert-space dimensions $d_{A}=2^{n_{A}}$ and $d_{B}=2^{n_{B}}$, respectively.
We label the computational basis (CB) of each factor by $\Ftwo^{n_{A}}$ and $\Ftwo^{n_{B}}$ and read the labels as binary integers.
Every pure state has a Schmidt decomposition~\cite{Nielsen2000Quantum}
  \begin{equation}
       \ket{\Psi}=\sum_{x=0}^{d-1}\mu_{x}\,
       \ket{\phi_{x}}_{A}\otimes\ket{\chi_{x}}_{B},
       \qquad \mu_{x}\ge0,\qquad \sum_{x}\mu_{x}^{2}=1,
       \label{eq:schmidt}
  \end{equation}
where $d:=\min(d_{A},d_{B})$ and the sets $\{\ket{\phi_{x}}\}\subset\mathbb C^{d_{A}}$ and $\{\ket{\chi_{x}}\}\subset\mathbb C^{d_{B}}$ are orthonormal.
We order the Schmidt coefficients as $\mu_{0}\ge\mu_{1}\ge\cdots\ge\mu_{d-1}$, and we denote the resulting Schmidt spectrum by $\boldsymbol\mu=(\mu_{0},\dots,\mu_{d-1})$.
Local basis changes act as $\ket{\Psi}\mapsto(V_A\otimes V_B)\ket{\Psi}$ with $V_A\in U(d_{A})$ and $V_B\in U(d_{B})$.

Magic, or nonstabilizerness, is the resource that promotes stabilizer protocols to universal quantum computation~\cite{Gottesman1998heis, Bravyi2005magic, Veitch2014, Howard2017}.  A magic monotone is a function $\mathfrak M$ of the state that vanishes precisely on stabilizer states, is invariant under Clifford unitaries, and does not increase under stabilizer protocols~\cite{Veitch2014, Liu2022ManyBody}.
Such a measure generally depends on the local bases: local unitaries preserve the correlations across the bipartition but can change $\mathfrak M$. The associated \emph{nonlocal magic} removes this dependence by minimizing over local unitaries,
\begin{equation}
     \mathfrak M_{\mathrm{NL}}(\Psi):=\min_{V_A\in U(d_A),\,V_B\in U(d_B)}
     \mathfrak M\bigl((V_A\otimes V_B)\ket{\Psi}\bigr).
     \label{eq:MNLmin}
\end{equation}
For pure states the Schmidt spectrum is a complete local-unitary invariant~\cite{Nielsen2000Quantum}: local unitaries leave the Schmidt values $\mu_{x}$ unchanged and can map any pair of Schmidt bases to any other. Hence $\mathfrak M_{\mathrm{NL}}$ is a function of $\{\mu_{x}\}$ alone.
The central difficulty is how to extract it from the spectrum without resorting to the formidable optimization over the local unitaries $V_A,V_B$.

In this work, we take $\mathfrak M$ to be the stabilizer $2$-R\'enyi entropy (SRE).  We denote by $\mathcal P_N$ the set of Hermitian Pauli strings on $N$ qubits, normalized by $\Tr(PQ)=2^N\delta_{P,Q}$.  For a pure $N$-qubit state the SRE is~\cite{Leone2022SRE}
\begin{equation}
       \mathcal M_{2}(\Psi)=-\log_{2}\sum_{P\in\mathcal P_{N}}
       \frac{\bigl(\bra{\Psi}P\ket{\Psi}\bigr)^{4}}{2^{N}} ,
       \label{eq:sre}
\end{equation}
and it is faithful, Clifford invariant, additive, and monotone under pure-state stabilizer protocols~\cite{Haug2023Monotones, Leone2024Monotones, Bittel2026Operational}.  
Specializing Eq.~\eqref{eq:MNLmin} to this choice gives the nonlocal SRE
\begin{equation}
       \mathcal M_{\mathrm{NL}}(\Psi)=\min_{V_{A}\in U(d_{A}),\,V_{B}\in U(d_{B})}
       \mathcal M_{2}\bigl((V_{A}\otimes V_{B})\ket{\Psi}\bigr),
       \label{eq:MNLsre}
\end{equation}
that we focus on in this work.

For the bipartite state $\ket{\Psi}$ of Eq.~\eqref{eq:schmidt}, we take as candidate minimizer in Eq.~\eqref{eq:MNLsre} the computational-basis representative
\begin{equation}
       \ket{\Psi_{\mathrm{CB}}}:=\sum_{x=0}^{d-1}\mu_{x}\,\ket{x}_{A}\otimes\ket{x}_{B},
       \label{eq:cbstate}
\end{equation}
where in each factor $\ket{x}$ denotes the computational basis vector labelled by the binary integer $x$. When $n_{A}\neq n_{B}$, the sum occupies only the first $d$ labels of the larger factor, that is, the labels whose leading $|n_{A}-n_{B}|$ bits vanish. 
Local unitaries $V_A, V_B$ can map the Schmidt vectors to the corresponding computational-basis vectors, so $\ket{\Psi_{\mathrm{CB}}}$ lies on the local-unitary orbit of $\ket{\Psi}$.
We study the following conjecture.

  \begin{conjecture}[CB optimality]
  \label{conj:cb}
  For every pure state $\ket{\Psi}$ and every bipartition $N=n_{A}+n_{B}$, the minimum in Eq.~\eqref{eq:MNLsre} is attained at the computational-basis representative of
  Eq.~\eqref{eq:cbstate},
  \begin{equation}
   \mathcal M_{\mathrm{NL}}(\Psi)=\mathcal M_{2}(\Psi_{\mathrm{CB}}).
   \label{eq:cbSRE}
  \end{equation}
  \end{conjecture}
This conjecture was proposed in Refs.~\cite{Liu2026Schmidt, Torre2026Spectrum, Franchini2026Schmidt} (see also Ref.~\cite{Busoni2026Qudits} for related conjectures in qudit systems).
Previous proofs cover $N=2$~\cite{Qian2025Nonlocal} and, more generally, every bipartition with $\min\{n_A,n_B\}=1$~\cite{Franchini2026Schmidt}.


Notably, for the nonlocal min-relative entropy of nonstabilizerness (NMRE), $D_{\mathrm{min}}^{\mathrm{NL}}$, the corresponding optimization has been solved in closed form at every system size and bipartition, with the sorted CB representative attaining the minimum~\cite{Sierant26exact,Viscardi26nonlocal}.
This quantity can be evaluated directly from the Schmidt spectrum and accessed experimentally through entanglement spectroscopy~\cite{Sierant26exact}. A related result was obtained for nonlocal trace distance magic~\cite{EbnerToAppear}.
The minimization for the nonlocal SRE is harder, because the SRE weighs the entire Pauli spectrum state, and Conjecture~\ref{conj:cb} remains open in general.

\subsection{Overview of the results}
\label{ssec:results}

In this section, we give an overview of our results, summarized in Fig.~\ref{fig:overview}. We begin with the two regimes in which CB optimality is established globally, and then turn to the two-sided bounds for the nonlocal SRE that apply to an arbitrary Schmidt spectrum.

\subsubsection{States with proven CB optimality}
We prove Conjecture~\ref{conj:cb} globally in two complementary regimes, in both cases at an arbitrary bipartition $N=n_{A}+n_{B}$.

 \begin{theorem}[Global CB optimality for dyadic-staircase spectra]
  \label{thm:resDyadic}
  Let $\ket{\Psi}$ be a pure state of $N$ qubits, and let it have a \emph{dyadic-staircase} spectrum, that is, let its Schmidt values be
  constant on each dyadic shell $D_k$,
  \begin{equation}
   \mu_{x}=\theta_{k}\ \ \text{for}\ \ x\in D_{k},
   \quad
   D_{0}=\{0\},
   \,\,\,
   D_{k}=\{2^{k-1},\dots,2^{k}-1\},
   \label{eq:dyadicspec}
  \end{equation}
  with $k=1,\dots,\log_{2}d$.
  Then Conjecture~\ref{conj:cb} holds for $\ket{\Psi}$.
  \end{theorem}

\begin{theorem}[Global CB optimality at rank at most six]
  \label{thm:resRank6}
  Let $\ket{\Psi}$ be a pure state of $N$ qubits, and let its Schmidt rank be at most six, $\mu_{x}=0$ for $x\ge6$.
  Then Conjecture~\ref{conj:cb} holds for $\ket{\Psi}$.
\end{theorem}

For the spectra covered by Theorems~\ref{thm:resDyadic} and~\ref{thm:resRank6}, the minimization of Eq.~\eqref{eq:MNLsre} is solved \textit{exactly}.
The nonlocal SRE $\mathcal M_{\mathrm{NL}}(\Psi)$ is obtained from Eq.~\eqref{eq:cbSRE} by evaluating $\mathcal M_{2}$ on the single state $\ket{\Psi_{\mathrm{CB}}}$, cf. Eq.~\eqref{eq:cbstate}, and no optimization over local unitaries is left to perform.

The two theorems apply under complementary restrictions on the Schmidt spectrum.
The dyadic-staircase condition constrains the shape of the spectrum, while allowing an arbitrary dyadic (i.e., power-of-two) rank and an arbitrary system size. The rank-six theorem, in contrast, allows for arbitrary Schmidt values at any system size, but restricts the rank of the spectrum.

The two proofs rely on different structures.
In Sec.~\ref{sec:dyadic}, we reduce CB optimality to a trace inequality for projectors of dyadic ranks.
We derive twenty upper bounds that depend only on these ranks and show that the tightest is saturated by projectors onto the leading computational-basis labels.
In Sec.~\ref{sec:rank6}, we reduce CB optimality to two inequalities involving the antisymmetric stabilizer projector.
We establish these using the structure of the antisymmetric stabilizer subspace and the matching polytope of the complete graph $K_{6}$.

\subsubsection{Bounds on nonlocal SRE}
Spectra of rank seven or larger that are not dyadic staircases are covered by neither theorem, and Conjecture~\ref{conj:cb} remains open for them.
These spectra are nevertheless not beyond reach.
For states with an arbitrary Schmidt spectrum we obtain two-sided bounds that determine the nonlocal SRE up to a uniformly bounded additive error.

We group the labels into the dyadic shells $D_{k}$ of Eq.~\eqref{eq:dyadicspec}, write $\lambda_{x}=\mu_{x}^{2}$ for the Schmidt probabilities, and define
\begin{equation}
     \hat w_0:=\lambda_0,
     \qquad
     \hat w_k:=2^{k-1}\lambda_{2^{k-1}},
     \qquad k=1,\dots,\log_2 d.
     \label{eq:resShellMasses}
\end{equation}
Each $\hat w_k$ is the size of the shell $D_k$ times its largest probability and therefore bounds the total probability in that shell from above.  We write $Q:=\sum_{k=0}^{\log_2 d}\hat w_k^4$ for the sum of their fourth powers.

\begin{theorem}[Universal nonlocal SRE approximation]
  \label{thm:resValue}
  For every pure state $\ket{\Psi}$ with Schmidt spectrum $\boldsymbol\mu$,
 \begin{equation}
   \mathcal M_{\mathrm{NL}}=\mathcal M_{\mathrm{CB}}-c_{1},
   \qquad
   \mathcal M_{\mathrm{CB}}=-\log_{2}Q+c_{2},
   \label{eq:resCombined}
  \end{equation}
  where $0\le c_{1}\le4$ and $-\log_{2}C_{\mathrm{sh}}\le c_{2}\le4$, with $C_{\mathrm{sh}}<2^{10}$ a universal constant.
\end{theorem}
The constants $c_1, c_2$ depend on the state, but their ranges do not: they are the same for every spectrum, every system size $N$, and every bipartition.
Theorem~\ref{thm:resValue} is \textit{the central result} of this work: it determines the nonlocal SRE of an arbitrary pure state up to bounded constants, without assuming Conjecture~\ref{conj:cb}.
The conjecture asserts precisely that $c_{1}=0$, so a full global proof of the CB optimality conjecture would only fix that constant to zero.
In other words, Theorem~\ref{thm:resValue} demonstrates that Conjecture~\ref{conj:cb} holds up to an additive error of at most four, uniformly in the spectrum, the system size, and the bipartition.
In particular, for any family of states in which $\mathcal M_{\mathrm{CB}}$, or equivalently $-\log_{2}Q$, diverges with the system size, the leading term of that divergence is also the leading term of the nonlocal SRE $\mathcal M_{\mathrm{NL}}$.

Both statements in Theorem~\ref{thm:resValue} follow from the \emph{shell partners} of $\boldsymbol\mu$: the generally unnormalized dyadic staircases $\nudown$ and $\nuup$ obtained by replacing every entry in each dyadic shell with that shell's smallest or largest entry, respectively.
They satisfy $\nudown\le\boldsymbol\mu\le\nuup$ entrywise.
Homogeneity extends the CB optimality of Theorem~\ref{thm:resDyadic} to these unnormalized vectors. Monotonicity under entrywise increases then yields bounds on the target's optimal stabilizer purity. Section~\ref{sec:witness} states and proves the two bounds separately.

For general spectra, the constant $c_{1}$ is not narrowed further by the arguments of this work.
For individual spectra, however, a search over suitable \emph{comparator spectra} can yield sharper lower bounds on the nonlocal SRE.
We choose these comparators from the families covered by Theorems~\ref{thm:resDyadic} and~\ref{thm:resRank6}, whose optimal stabilizer purities are known exactly.
The mismatch between the target and comparator is controlled by a continuity bound on the fourth root of stabilizer purity that holds uniformly over all local unitaries.

\begin{theorem}[Comparator bound]
  \label{thm:resComparator}
  Let $\ket{\Psi}$ have Schmidt spectrum $\boldsymbol\mu$, and let $\boldsymbol\nu\ge0$ be a nonincreasing, possibly unnormalized spectrum of the same length.
  Write $c=\sum_{x}\nu_{x}^{2}$ and $\hat{\boldsymbol\nu}=\boldsymbol\nu/\sqrt c$.
  If $\hat{\boldsymbol\nu}$ is a dyadic staircase or has rank at most six, then
  \begin{equation}
   \mathcal M_{\mathrm{NL}}(\Psi)\ \ge\ -\log_{2}U(\boldsymbol\nu;\boldsymbol\mu),
   \label{eq:resLower}
\end{equation}
where
\begin{equation}
   U(\boldsymbol\nu;\boldsymbol\mu)=\min\Bigl\{1,
   \bigl[c\,2^{-\mathcal M_{\mathrm{CB}}(\hat{\boldsymbol\nu})/4}
   +(B^{\flat})^{1/4}\bigr]^{4}\Bigr\},
   \label{eq:resU}
\end{equation}
and
\begin{equation}
   B^{\flat}=\bigl(1+c^{2}-2z\bigr)\bigl((1+c)^{2}-4z\bigr),
   \qquad z=\Bigl(\sum_{x}\mu_{x}\nu_{x}\Bigr)^{2}.
   \label{eq:resBflat}
  \end{equation}
\end{theorem}
Each admissible $\boldsymbol\nu$ defines a possibly unnormalized \emph{comparator state} $\ket{\phi_{\boldsymbol\nu}}=\sum_x\nu_x\ket{x}_A\otimes\ket{x}_B$.
Theorems~\ref{thm:resDyadic} and~\ref{thm:resRank6} determine the nonlocal SRE of its normalized counterpart $\ket{\phi_{\boldsymbol\nu}}/\sqrt c$ exactly.
The operator $\Delta:=\ketbra{\Psi_{\mathrm{CB}}}{\Psi_{\mathrm{CB}}}  -\ketbra{\phi_{\boldsymbol\nu}}{\phi_{\boldsymbol\nu}}$  has rank at most two.
Bounding its contribution to the Pauli fourth moment gives the correction $(B^\flat)^{1/4}$, which controls the difference between the fourth roots of the target and comparator stabilizer purities.
This correction depends only on the comparator's squared norm $c$ and squared overlap $z$ with the target, so the same bound holds at every pair of local unitaries.

Every admissible comparator yields a valid lower bound through Eq.~\eqref{eq:resLower}. We evaluate this bound for a set of comparators and retain the largest value. The estimate is accurate when $(B^\flat)^{1/4}$ is small compared with $c\,2^{-\mathcal M_{\mathrm{CB}}(\hat{\boldsymbol\nu})/4}$, and is exact when the target itself is an admissible comparator.
A simple choice retains the six leading Schmidt values and sets the rest to zero; the resulting bound becomes tight as the discarded weight $\epsilon_6:=\sum_{x=6}^{d-1}\mu_x^2$ tends to zero.

\subsubsection{Application to many-body systems}
We close this overview with the many-body consequences of Theorems~\ref{thm:resDyadic}--\ref{thm:resValue}.
By Eq.~\eqref{eq:resCombined}, a lower bound on $Q$ yields an upper bound on the nonlocal SRE. The derivation in Sec.~\ref{ssec:value} uses both the total number $n+1$ of dyadic shells and the fact that $O(1+S_1)$ leading shells suffice to contain a fixed fraction of the total
Schmidt probability.
We obtain
\begin{align}
     \mathcal M_{\mathrm{NL}}\ \le\ \mathcal M_{\mathrm{CB}}
     &\ \le\ 3\log_{2}(2S_{1}+3)+8,
     \label{eq:resEntropy}\\
     \mathcal M_{\mathrm{NL}}\ \le\ \mathcal M_{\mathrm{CB}}
     &\ \le\ 3\log_{2}(n+1)+4,
     \label{eq:resCeiling}
\end{align}
where $S_{1}=-\sum_{x}\lambda_{x}\log_{2}\lambda_{x}$ is the entanglement entropy for  the given cut.

The nonlocal SRE therefore grows at most logarithmically with the entanglement entropy and, since $S_1\le \mathrm{min}\{n_A,n_B\}$, also with the size of the smaller subsystem.
In one dimension, area-law entanglement, $S_1=O(1)$, implies $\mathcal M_{\mathrm{NL}}=O(1)$, while critical states with logarithmic entanglement, $S_1=O(\log N)$, have $\mathcal M_{\mathrm{NL}}=O(\log\log N)$.
Even volume-law entanglement permits at most $O(\log N)$ growth of nonlocal SRE.
The entanglement entropy enters Eq.~\eqref{eq:resEntropy} through an upper bound; the finer spectral quantity $-\log_2 Q$ determines the nonlocal SRE up to a uniformly bounded additive error, as stated in Eq.~\eqref{eq:resCombined}.

We apply the above theorems to determine the behavior of nonlocal SRE in the ground state of the transverse-field Ising chain~\cite{Mbeng24ising, Surace22ferms}.
The chain is ordered for transverse field $0\le g<1$ and disordered for $g>1$, with a quantum critical point at $g=1$.
In the gapped phases, the area law keeps the nonlocal SRE bounded as the system size grows.
In the gapped regime shown in Fig.~\ref{fig:overview}, the six leading Schmidt probabilities carry nearly all the spectral weight and the corresponding rank-six comparator gives, through Theorem~\ref{thm:resComparator}, a lower bound close to $\mathcal M_{\mathrm{CB}}$, tightly enclosing the nonlocal SRE.

At the critical point the Schmidt weight instead spreads over many dyadic shells, and the nonlocal SRE grows without bound.
Our numerical evaluation of the bounds of Theorems~\ref{thm:resValue} and~\ref{thm:resComparator}, exact for chains up to $N=3\times10^{6}$ and extended to $N\simeq10^{28}$, supports the doubly-logarithmic growth
\begin{equation}
   \mathcal M_{\mathrm{NL}}(N)=\tfrac{3}{2}\log_{2}\ln N+O(1),
   \label{eq:resIsing}
\end{equation}
without assuming Conjecture~\ref{conj:cb}.
The same behavior of $\mathcal{M}_{\mathrm{NL}}$ follows from Theorem~\ref{thm:resValue} for the Calabrese--Lefevre form of the entanglement spectrum~\cite{Calabrese2008spectrum}, whose Gaussian shell profile of width $\Theta(\sqrt{\ln N})$ gives $Q=\Theta((\ln N)^{-3/2})$, and it agrees with the double-logarithmic scaling scenario put forward for the CB value $\mathcal{M}_{\mathrm{CB}}$ in Ref.~\cite{Torre2026Spectrum}.
Section~\ref{ssec:critical} presents the numerical results and the bounds of Theorems~\ref{thm:resValue} and~\ref{thm:resComparator} for the Ising chain.

\section{Reduction of the stabilizer purity functional}
  \label{sec:reduction}

In this section, we reduce the optimization~\eqref{eq:MNLsre} that defines the nonlocal SRE to a single-unitary problem on an explicit subspace.
We begin by showing that it suffices to treat the balanced bipartitions ($n_A=n_B$), and then we obtain the stabilizer purity functional $\Fmu$, in terms of which every later statement is made.

\subsection{Balanced bipartitions}
\label{subsec:bal}
Throughout the analysis we take the two subsystems to be of equal size, $n_{A}=n_{B}=n$, so that $d_{A}=d_{B}=d=2^{n}$.
An unbalanced bipartition becomes a balanced one once the smaller subsystem is enlarged by ancillary qubits in a fixed product state, so every bipartition is covered.
To see this, let $n_{A}<n_{B}$ and write $\ket{\Psi'}\equiv\ket{\Psi}\otimes\ket{0}^{\otimes(n_{B}-n_{A})}$ for the padded state.
Adjoining the ancillas in the state $\ket{0}$ only fixes the leading bits of the enlarged factor to zero, so the labels $x<d_{A}$ of Eq.~\eqref{eq:cbstate} are read unchanged as labels of the larger space, and the computational-basis representative of $\ket{\Psi'}$ at the balanced bipartition is 
$\ket{\Psi'_{\mathrm{CB}}}=\ket{\Psi_{\mathrm{CB}}}\otimes\ket{0}^{\otimes(n_{B}-n_{A})}$.
If CB optimality holds at the balanced bipartition, then
\begin{equation}
  \begin{split}
   \mathcal M_{\mathrm{NL}}(\Psi)
   \ \ge\ &\mathcal M_{\mathrm{NL}}(\Psi')
   \ =\ \mathcal M_{2}(\Psi'_{\mathrm{CB}})\\
   &\ =\ \mathcal M_{2}(\Psi_{\mathrm{CB}})
   \ \ge\ \mathcal M_{\mathrm{NL}}(\Psi),
  \end{split}
  \label{eq:padchain}
\end{equation}
where the first step restricts the padded minimization to local unitaries acting trivially on the ancillas, the second is CB optimality at the balanced bipartition, the third is additivity of $\mathcal M_{2}$ together with $\mathcal M_{2}(\ket{0})=0$, and the last holds because $\ket{\Psi_{\mathrm{CB}}}$ lies on the local-unitary orbit of $\ket{\Psi}$.
Hence CB optimality at $n_{A}=n_{B}$ implies it at every bipartition. The two classes of Theorems~\ref{thm:resDyadic} and~\ref{thm:resRank6} are themselves stable under the padding, since appending zeros to a spectrum leaves its rank unchanged and extends
  a dyadic staircase by shells of vanishing weight.

\subsection{The stabilizer purity functional}
\label{subsec:stabfuntional}
We first absorb the Schmidt bases of Eq.~\eqref{eq:schmidt} into the local unitaries. Let $U_{A}$ and $U_{B}$ be unitaries with $U_{A}\ket{x}=\ket{\phi_{x}}$ and $U_{B}\ket{x}=\ket{\chi_{x}}$ for $x=0,\dots,d-1$, so that
\begin{equation}
   \ket{\Psi}=(U_{A}\otimes U_{B})\ket{\Psi_{\mathrm{CB}}}\;.
   \label{eq:absorb}
\end{equation}
As $V_{A}$ and $V_{B}$ range over $U(d)$, so do $V_{A}U_{A}$ and $V_{B}U_{B}$.
Equation~\eqref{eq:MNLsre} is therefore unchanged if we replace $\ket{\Psi}$ by $\ket{\Psi_{\mathrm{CB}}}$,
\begin{equation}
   \mathcal M_{\mathrm{NL}}(\Psi)=\min_{V_{A},V_{B}\in U(d)}
   \mathcal M_{2}\bigl((V_{A}\otimes V_{B})\ket{\Psi_{\mathrm{CB}}}\bigr)\;.
   \label{eq:MNLcb}
\end{equation}
Every Pauli string on the $2n$ qubits factorizes as $P\otimes Q$ with $P,Q\in\mathcal P_{n}$, and evaluating its expectation value on $\ket{\Psi_{\mathrm{CB}}}$ gives
\begin{equation}
\begin{split}
   \bra{\Psi_{\mathrm{CB}}}&(V_{A}^{\dagger}PV_{A})\otimes(V_{B}^{\dagger}QV_{B})\ket{\Psi_{\mathrm{CB}}}\\
   &=\Tr\bigl[V_{A}^{\dagger}PV_{A}\,M\,(V_{B}^{\dagger}QV_{B})^{\mathsf T}M\bigr]\;,
\end{split}
\label{eq:trace}
\end{equation}
where $M\equiv\mathrm{diag}(\mu_{0},\dots,\mu_{d-1})$ is a diagonal matrix.
The transpose in Eq.~\eqref{eq:trace} may be omitted.
Every $Q\in\mathcal P_{n}$ obeys $Q^{\mathsf T}=\pm Q'$ for some $Q'\in\mathcal P_{n}$, the exponent in Eq.~\eqref{eq:sre} is even, and $V_{B}\mapsto\overline{V_{B}}$ is a bijection of $U(d)$.
Relabeling the summation over $Q$ and the optimization variable $V_{B}$ therefore removes the transpose, and substituting Eq.~\eqref{eq:trace} into Eq.~\eqref{eq:sre} expresses the \emph{stabilizer purity} $2^{-\mathcal M_{2}}$ as
\begin{equation}
   \Fmu(V_A,V_B)\equiv
   \frac{1}{d^{2}}\sum_{P,Q\in\mathcal P_n}
   \Bigl(\Tr\bigl[V_A^{\dagger}PV_A\,M\,V_B^{\dagger}QV_B\,M\bigr]\Bigr)^{4},
   \label{eq:Fstar}
\end{equation}
so that the nonlocal SRE follows as
\begin{equation}
   \mathcal M_{\mathrm{NL}}=-\log_{2}\max_{V_{A},V_{B}}\Fmu(V_{A},V_{B})\;.
   \label{eq:MNLF}
\end{equation}
The relabeling makes $\Fmu(V_{A},V_{B})$ the stabilizer purity of $(V_{A}\otimes\overline{V_{B}})\ket{\Psi_{\mathrm{CB}}}$, which is immaterial in Eq.~\eqref{eq:MNLF} because conjugation is a bijection of $U(d)$.
Conjecture~\ref{conj:cb} states that the maximum in Eq.~\eqref{eq:MNLF} is attained at $V_{A}=V_{B}=\mathbb 1$. All results below are statements about the functional $\Fmu$.

\subsection{Diagonal reduction}
\label{subsec:diag}
The first simplification removes optimization over one of the two local unitaries.
\begin{lemma}[Diagonal reduction]
  \label{lem:diag}
  For every spectrum $\boldsymbol\mu$,
  \begin{equation}
   \max_{V_{A},V_{B}}\Fmu(V_{A},V_{B})=\max_{V}\Fmu(V,V),
  \end{equation}
  and, with $N=VMV^{\dagger}$,
  \begin{equation}
   \Fmu(V,V)=\Gm(V),\;\;
   \Gm(V)\equiv\frac1{d^{2}}\!\sum_{P,Q\in\mathcal P_{n}}\!
   \bigl[\Tr(PNQN)\bigr]^{4}.
   \label{eq:Ghat}
  \end{equation}
\end{lemma}

  \begin{proof}
The two unitaries enter Eq.~\eqref{eq:Fstar} only through the traces $\Tr[V_{A}^{\dagger}PV_{A}\,M\,V_{B}^{\dagger}QV_{B}\,M]$, and each such trace separates into a factor built from $V_{A}$ and a factor built from $V_{B}$.
With the Hermitian matrices
\begin{equation}
   a_{P}\equiv M^{1/2}V_{A}^{\dagger}PV_{A}M^{1/2},
   \qquad
   b_{Q}\equiv M^{1/2}V_{B}^{\dagger}QV_{B}M^{1/2},
\end{equation}
the trace is the real Hilbert--Schmidt inner product $\langle a_{P},b_{Q}\rangle\equiv\Tr(a_{P}b_{Q})$.
A fourth power of an inner product is again an inner product, this time of fourth tensor powers, $\langle a_{P},b_{Q}\rangle^{4}=\langle a_{P}^{\otimes4},b_{Q}^{\otimes4}\rangle$.
Summing over $P$ and $Q$ therefore collects the whole functional into a single inner product,
\begin{equation}
   d^{2}\,\Fmu(V_{A},V_{B})
   =\Bigl\langle\,\sum_{P}a_{P}^{\otimes4},\ \sum_{Q}b_{Q}^{\otimes4}\Bigr\rangle,
\end{equation}
between one vector built from $V_{A}$ alone and one built from $V_{B}$ alone.
The Cauchy--Schwarz inequality bounds this inner product by the product of the two norms.
The square of each norm is $d^{2}\Fmu$ evaluated with the same unitary in both arguments, and therefore
\begin{equation}
   \Fmu(V_{A},V_{B})^{2}\;\le\;\Fmu(V_{A},V_{A})\,\Fmu(V_{B},V_{B}).
\end{equation}
No pair $(V_{A},V_{B})$ therefore exceeds the larger of $\Fmu(V_{A},V_{A})$ and $\Fmu(V_{B},V_{B})$, both of which are themselves admissible in the maximization. The maximum is thus attained at some $V_{A}=V_{B}=V$.
\end{proof}

\begin{corollary}[Single-unitary form of the CB conjecture]
  \label{cor:cbG}
  For a Schmidt spectrum $\boldsymbol\mu$, Conjecture~\ref{conj:cb} holds if and only if
  \begin{equation}
   \Gm(V)\;\le\;\Gm(\mathbb 1)
   \qquad\text{for every }V\in U(d).
   \label{eq:cbGhat}
  \end{equation}
\end{corollary}
We attack the CB conjecture in this form throughout the paper. 
All global optimality results below are statements about the single-unitary functional $\Gm(V)$. The two-unitary form then follows from Lemma~\ref{lem:diag}, and unbalanced bipartitions from Sec.~\ref{subsec:bal}.

\subsection{Stabilizer subspace formulation}
The functional $\Gm$ admits an exact four-replica representation.
For $u,v\in\Ftwo^{n}$ we define the \emph{stabilizer vectors}
  \begin{equation}
     \ket{\Phi_{u,v}}\equiv\frac1{\sqrt d}\sum_{x\in\Ftwo^{n}}
     \ket{x,\,x{+}u,\,x{+}v,\,x{+}u{+}v}
     \label{eq:Phibasis}
  \end{equation}
in $(\mathbb{C}^{d})^{\otimes4}$, where each replica carries the computational basis $\{\ket{x}\}_{x\in\Ftwo^{n}}$ used in Eq.~\eqref{eq:cbstate}.
The four replicas carry the labels $x$, $x{+}u$, $x{+}v$ and $x{+}u{+}v$, whose pairwise differences are only $u$, $v$ and $u{+}v$.
  The vector $\ket{\Phi_{u,v}}$ is the uniform superposition of these quadruples over $x\in\Ftwo^{n}$.
A basis term of Eq.~\eqref{eq:Phibasis} determines $x$, $u$ and $v$, so distinct labels give disjoint terms and the $d^{2}$ vectors $\ket{\Phi_{u,v}}$ are orthonormal.
Their span $W$ is the \emph{stabilizer subspace}, the subspace fixed by every fourth Pauli power $P^{\otimes4}$, and its orthogonal projector is the Pauli average~\cite{Zhu16fails, Roth18recovering, Leone2021QuantumChaos, Turkeshi23measuring, Gross21comm, Bittel26complete} (see also Lemma~\ref{lem:Phi})
\begin{equation}
   \PW=\sum_{u,v\in\Ftwo^{n}}\ketbra{\Phi_{u,v}}{\Phi_{u,v}}
   =\frac1{d^{2}}\sum_{P\in\mathcal P_n}P^{\otimes4}.
   \label{eq:PWdef}
\end{equation}
Expanding the two Pauli sums of Eq.~\eqref{eq:Ghat} against Eq.~\eqref{eq:PWdef} gives the exact representation
\begin{equation}
   \Gm(V)=d^{2}\,\Tr\!\left[(\PW N^{\otimes4}\PW)^{2}\right]
   =d^{2}\,\bigl\|\PW N^{\otimes4}\PW\bigr\|_{\mathrm{HS}}^{2},
   \label{eq:Gcompressed}
\end{equation}
where $\|X\|_{\mathrm{HS}}\equiv[\Tr(X^{\dagger}X)]^{1/2}$ is the Hilbert--Schmidt norm and $N=VMV^{\dagger}$. 
Maximizing $\Gm$ is thus equivalent to maximizing $\|\PW N^{\otimes4}\PW\|_{\mathrm{HS}}$ over the unitary orbit of $M$. This compression is a large reduction, since $N^{\otimes4}$ acts on $d^{4}$ dimensional space while $\PW N^{\otimes4}\PW$ is supported on the $d^{2}$-dimensional stabilizer space $W$.

\section{Global CB optimality for dyadic-staircase spectra}
\label{sec:dyadic}
This section proves Theorem~\ref{thm:resDyadic}: at every $d=2^{n}$, $V=\mathbb 1 $, corresponding to the CB representative $\ket{\Psi_{\mathrm{CB}}}$, is a global maximizer of $\Gm(V)$ on the $(n{+}1)$-parameter family of dyadic-staircase spectra.
A dyadic staircase is a nonnegative combination of projectors onto the leading computational-basis labels, so the trace in Eq.~\eqref{eq:Gcompressed}, quadratic in the fourth tensor power, expands into a nonnegative combination of terms, each built from eight such projectors.

Each term is then bounded by its value in the computational basis, by an inequality that holds for arbitrary projectors of the given dyadic ranks and therefore for every $V$ at once.
Summing those bounds with nonnegative weights reassembles $\Gm(\mathbb 1)$ and proves Theorem~\ref{thm:resDyadic}.

  \subsection{Dyadic-staircase spectra}
  \label{ssec:staircase}

We first record a second description of the family of spectra in Theorem~\ref{thm:resDyadic}, which is the one the proof uses.
Let $\mathcal K_{n}$ denote the set of diagonal matrices $M=\mathrm{diag}(\mu_{0},\dots,\mu_{d-1})$ whose diagonal is a sorted dyadic staircase in the label order, i.e., obeys Eq.~\eqref{eq:dyadicspec}, and let $E_{r}\equiv\sum_{j=0}^{r-1}\ketbra{j}{j}$ be the \emph{truncation projectors}, which retain the $r$ leading computational-basis labels.
Writing $\theta_{k}$ for the common value of the spectrum on the dyadic shell $D_{k}$, every $M\in\mathcal K_{n}$ satisfies
\begin{equation}
   M=\sum_{k=0}^{n}c_{k}\,E_{2^{k}},
   \qquad
   c_{k}\equiv\theta_{k}-\theta_{k+1},
   \qquad
   \theta_{n+1}\equiv0 .
   \label{eq:cone}
\end{equation}
Since $\boldsymbol\mu$ is sorted in descending order, the shell values obey $\theta_{0}\ge\theta_{1}\ge\dots\ge\theta_{n}\ge0$, and every $c_{k}$ is nonnegative.
A dyadic staircase is therefore a nonnegative combination of truncation projectors at dyadic ranks.

Such a staircase has steps only at $1,2,4,8\dots$, and is parametrized by the $n{+}1$ shell values $\theta_{0},\dots,\theta_{n}$ rather than by all $2^{n}$ Schmidt values.
The family contains every flat spectrum of dyadic rank, and in particular the Schmidt spectrum of every bipartite stabilizer state.
It also contains single exceptions above a flat floor and block-geometric staircases, which are the natural dyadic caricatures of slowly varying physical spectra.

\subsection{Expansion into truncation projectors}
\label{ssec:expansion}

We now expand $\Gm$ using Eq.~\eqref{eq:cone}.
Let $M\in\mathcal K_{n}$ and $N=VMV^{\dagger}$, so that $N=\sum_{k}c_{k}VE_{2^{k}}V^{\dagger}$.
Expanding the fourth tensor power gives
\begin{equation}
   N^{\otimes4}=\sum_{\alpha\in\{0,\dots,n\}^{4}}c_{\alpha}\,P_{\alpha}(V),
   \qquad c_{\alpha}\equiv\prod_{j=1}^{4}c_{\alpha_{j}}\ \ge0,
\end{equation}
with $P_{\alpha}(V)\equiv\bigotimes_{j}VE_{2^{\alpha_{j}}}V^{\dagger}$.
Inserting this into Eq.~\eqref{eq:Gcompressed} and writing $X_{\alpha}(V)\equiv\PW P_{\alpha}(V)\PW$, we obtain
\begin{equation}
   \Gm(V)=d^{2}\sum_{\alpha,\beta}c_{\alpha}c_{\beta}\,
   \Tr\bigl[X_{\alpha}(V)X_{\beta}(V)\bigr].
   \label{eq:conicexp}
\end{equation}
All of the dependence on the local unitary $V$ is now contained in the traces $\Tr[X_{\alpha}(V)X_{\beta}(V)]$.
Each trace involves eight projectors of dyadic rank, four from $P_{\alpha}(V)$ and four from $P_{\beta}(V)$, with each group of four tensored and then compressed to $W$ by $\PW$.
Every weight $c_{\alpha}c_{\beta}$ is nonnegative, so it suffices to bound each trace separately by its value at $V=\mathbb 1$.
We prove such a bound in the next subsection.

\subsection{The rank-level inequality}
\label{ssec:rank}
Let $\boldsymbol r=(2^{\alpha_{1}},2^{\alpha_{2}},2^{\alpha_{3}},2^{\alpha_{4}})$ and $\boldsymbol s=(2^{\beta_{1}},2^{\beta_{2}},2^{\beta_{3}},2^{\beta_{4}})$ be rank vectors, and let $P_{j}$ and $Q_{j}$ be arbitrary orthogonal projectors of ranks $r_{j}$ and $s_{j}$.
We write $A_{W}\equiv\PW(P_{1}\otimes P_{2}\otimes P_{3}\otimes P_{4})\PW$ and $B_{W}\equiv\PW(Q_{1}\otimes Q_{2}\otimes Q_{3}\otimes Q_{4})\PW$.
Replacing every $P_{j}$ and $Q_{j}$ by the \textit{truncation projector} of the same rank gives the computational-basis reference operators $A_{W}^{\mathrm{CB}}\equiv\PW\bigl(\bigotimes_{j}E_{2^{\alpha_{j}}}\bigr)\PW$ and $B_{W}^{\mathrm{CB}}\equiv\PW\bigl(\bigotimes_{j}E_{2^{\beta_{j}}}\bigr)\PW$.
These are the operators $X_{\alpha}(\mathbb 1)$ and $X_{\beta}(\mathbb 1)$ of Eq.~\eqref{eq:conicexp}

\begin{theorem}
  \label{thm:rank}
For all dyadic rank vectors $\boldsymbol r$ and $\boldsymbol s$, and for arbitrary orthogonal projectors $P_{j}$ and $Q_{j}$ of these ranks,
\begin{equation}
   \Tr\bigl(A_{W}B_{W}\bigr)\ \le\
   \Tr\bigl(A_{W}^{\mathrm{CB}}B_{W}^{\mathrm{CB}}\bigr).
   \label{eq:rankineq}
  \end{equation}
\end{theorem}
The complete proof of Theorem~\ref{thm:rank} is given in Appendix~\ref{app:rank}, and here we describe the mechanism behind it. 
Our strategy is to bound $\Tr(A_{W}B_{W})$ from above by quantities that depend only on the eight ranks $r_j, s_j$, and then to show that the computational-basis operators $A_{W}^{\mathrm{CB}}, B_{W}^{\mathrm{CB}}$ attain the tightest of these bounds, so that no other choice of the projectors can exceed them.
To obtain such bounds, we replace some of the projectors by the identity.
Since $P\le\mathbb 1$ for every projector and $\Tr(A_{W}B_{W})$ is monotone in each of its two positive-semidefinite arguments, this replacement can only increase the trace, and what remains can be bounded in terms of the ranks alone. Carrying this out in all admissible ways, we obtain twenty upper bounds that hold at every rank, dyadic or not.

The role of dyadicity is to make the tightest of these bounds sharp.
At dyadic ranks, the tightest of the twenty bounds is exactly the right-hand side of Eq.~\eqref{eq:rankineq}, which by Eq.~\eqref{eq:Scount} counts the pairs of label quadruples compatible with the eight ranks.
The two coincide because both are determined by the dimension of the same linear space of label configurations, so that the minimum over twenty combinatorial bounds reduces to a single dimension count.
Dyadicity is what makes this space linear, because for $r=2^{\alpha}$ the support $\{0,\dots,r-1\}$ of $E_{r}$ is a linear subspace of $\Ftwo^{n}$, and it is essential, since at non-dyadic ranks the twenty bounds are no longer guaranteed to be sharp.

Theorem~\ref{thm:rank} leads immediately to the CB optimality for the dyadic staircase spectra.
 \begin{theorem}[Global maximality of the sorted diagonal]
  \label{thm:cone}
  Let $M\in\mathcal K_{n}$. Then $\Gm(V)\le\Gm(\mathbb 1)$ for every $V\in U(d)$.
  \end{theorem}
\begin{proof}
Each term in the sum of Eq.~\eqref{eq:conicexp} is the left-hand side of Eq.~\eqref{eq:rankineq} with the eight projectors $VE_{2^{\alpha_{j}}}V^{\dagger}$ and
$VE_{2^{\beta_{j}}}V^{\dagger}$, and is therefore bounded by its value at $V=\mathbb 1$.
The weights $c_{\alpha}c_{\beta}$ are nonnegative, so the bounds sum without cancellation, and the summed right-hand sides are Eq.~\eqref{eq:conicexp} at $V=\mathbb 1$, that is, $\Gm(\mathbb 1)$.
\end{proof}

By Corollary~\ref{cor:cbG}, this proves Theorem~\ref{thm:resDyadic}, and it isolates the mechanism: the sorted diagonal maximizes every term of Eq.~\eqref{eq:conicexp} separately, and  because all weights are nonnegative, it also maximizes their sum.
The inequality is saturated on the orbit of the sorted diagonal under Clifford unitaries and under the symmetries of $M$, which leave $\Gm$ unchanged.
In Appendix~\ref{app:rank}, we determine the gap $\Gm(\mathbb 1)-\Gm(V)$ exactly, not only its sign.
\begin{proposition}
  \label{prop:gapformula}
  For $M\in\mathcal K_{n}$ and every $V$,
  \begin{equation}
   \Gm(\mathbb 1)-\Gm(V)
   =d^{2}\sum_{\alpha,\beta}c_{\alpha}c_{\beta}\,\Delta_{\alpha\beta}(V),
   \label{eq:gapformula}
  \end{equation}
where $\Delta_{\alpha\beta}(V)\equiv\Tr[X_{\alpha}(\mathbb 1)X_{\beta}(\mathbb 1)]-\Tr[X_{\alpha}(V)X_{\beta}(V)]$ is the loss of the term $(\alpha,\beta)$ at $V$, nonnegative by Theorem~\ref{thm:rank}.
  \end{proposition}
The gap at an arbitrary unitary is thus a nonnegative combination of termwise losses, and this is the structure that the covering arguments of Sec.~\ref{sec:witness} extend beyond dyadic staircases.

\section{Global CB optimality at Schmidt rank at most six}
\label{sec:rank6}
In this section, we prove Theorem~\ref{thm:resRank6}: 
for every bipartite state whose Schmidt spectrum has at most six nonzero values, the CB state is a global maximizer of the stabilizer purity at every $d=2^n$.
This is the regime complementary to Sec.~\ref{sec:dyadic}, where the spectrum was constrained in shape but not in rank.

We begin by splitting the compressed functional of Eq.~\eqref{eq:Gcompressed} with the help of Schur--Weyl duality into three sectors, labeled by the irreducible representations of the symmetric group $S_{4}$ that permutes the four replicas: the symmetric sector $(4)$, the mixed sector $(2,2)$, and the antisymmetric sector $(1^{4})$.
We then show, by means of two sum-of-squares certificates that hold at every system size and for every spectrum, that the symmetric and the mixed sector are bounded by data of the  antisymmetric sector alone.
In this way we arrive at an exact reduced form of the stabilizer purity in which all dependence on the local unitary resides in two quantities built from a single projector, the \emph{antisymmetric stabilizer projector}.
Only at this point do we use the rank hypothesis, in two inequalities for this projector.
The first is a compression inequality, valid up to rank seven, which bounds the antisymmetric-sector part of the stabilizer purity by the overlap of the projector with $\Lambda^{4}(N^{2})$.
The second bounds that overlap, at rank at most six, by its computational-basis value, and we prove it with the matching polytope of $K_{6}$.
Substituting the two inequalities into the reduced form, we obtain Theorem~\ref{thm:resRank6}.

The Schur--Weyl analysis is developed in Appendix~\ref{app:SW}, and the proofs of the two sector inequalities in Appendix~\ref{app:flatcode}.
Here we state the results and describe the mechanisms behind them.

\subsection{Three sectors of the compressed functional}
\label{sec:sectors}

The symmetric group $S_{4}$ permutes the four replicas of $(\mathbb C^{d})^{\otimes4}$, and both $\PW$ and $N^{\otimes4}$ commute with this action.
The compressed operator $\PW N^{\otimes4}\PW$ of Eq.~\eqref{eq:Gcompressed} is therefore block diagonal, with one block for each of the three irreducible representation of $S_{4}$ that occur in the stabilizer subspace $W$.
In the $\Phi$-basis of Eq.~\eqref{eq:Phibasis}, exchanging the replicas in two disjoint pairs only relabels the summation variable $x$ and fixes every $\ket{\Phi_{u,v}}$.
The only irreducible representations of $S_{4}$ in which such double exchanges act trivially are the trivial, the sign, and the two-dimensional one, $[4]$, $[1^{4}]$, and $[2,2]$, so the two three-dimensional representations are absent.
The squared Hilbert--Schmidt norm of the block in the sector $\lambda$ is $\dim[\lambda]\,Q_{\lambda}(N)$, where $Q_{\lambda}(N)$ is the contribution of a single copy of the representation (Appendix~\ref{app:sectors}).
Since $[2,2]$ is two dimensional while $[4]$ and $[1^{4}]$ are one dimensional, we obtain for $N=VMV^{\dagger}$ the exact decomposition
\begin{equation}
   \Gm(V)=d^{2}\bigl[\Qsym(N)+2\,\Qtt(N)+\Qa(N)\bigr].
   \label{eq:split}
  \end{equation}
The antisymmetric sector lives in the totally antisymmetric subspace $\Lambda^{4}\mathbb C^{d}\subset(\mathbb C^{d})^{\otimes4}$, which is spanned by the wedge products $v_{1}\wedge v_{2}\wedge v_{3}\wedge v_{4}$ of four vectors and is left invariant by $N^{\otimes4}$.
On it, $N^{\otimes4}$ acts as the operator
\begin{equation}
   \Lambda^{4}(N)\,(v_{1}\wedge v_{2}\wedge v_{3}\wedge v_{4})
   =Nv_{1}\wedge Nv_{2}\wedge Nv_{3}\wedge Nv_{4}
   \label{eq:wedgeaction}
\end{equation} 
with the action extended by linearity to sums of such wedge products, which exhaust $\Lambda^{4}\mathbb C^{d}$. The operator $\PW$ acts as the projector
\begin{equation}
   \PiF\equiv\frac1{d^{2}}\sum_{P\in\mathcal P_{n}}\Lambda^{4}(P),
   \label{eq:flatprojector}
\end{equation}
which we refer to as the \emph{antisymmetric stabilizer projector}, since its range is the antisymmetric part $W_{\mathrm F}\equiv W\cap\Lambda^{4}\mathbb C^{d}$ of the stabilizer subspace, the \emph{antisymmetric stabilizer subspace}.

We work in a specific basis of the antisymmetric stabilizer subspace.
For a four-element subset $S=\{s_{1}<s_{2}<s_{3}<s_{4}\}\subset\Ftwo^{n}$ let $e_{S}\equiv\ket{s_{1}}\wedge\ket{s_{2}}\wedge\ket{s_{3}}\wedge\ket{s_{4}}$ be the wedge of the four computational-basis vectors labeled by $S$, with the labels in increasing order. These $e_{S}$ form an orthonormal basis of $\Lambda^{4}\mathbb C^{d}$.
A four-element set of labels of the form $F=\{x,\,x{+}u,\,x{+}v,\,x{+}u{+}v\}$, equivalently a set of four labels whose sum vanishes, is a \emph{plane} in $\Ftwo^{n}$, and $L=\{0,u,v,u{+}v\}$ is its \emph{direction}.
Planes with the same direction are \emph{parallel}, and $\Ftwo^{n}$ splits into $d/4$ parallel planes of each direction.
The $Z$ part of the Pauli average in Eq.~\eqref{eq:flatprojector} annihilates every $e_{S}$ whose labels do not form a plane, and the $X$ part averages $e_{F}$ over all planes parallel to $F$, with all signs positive (Appendix~\ref{app:flatcode}).
Hence
\begin{equation}
 \PiF=\sum_{L}\ketbra{\psi_{L}}{\psi_{L}},
 \qquad
 \psi_{L}\equiv\sqrt q\sum_{F\parallel L}e_{F},
 \qquad
 q\equiv\frac4d,
 \label{eq:codeform}
\end{equation}
where the sum runs over the $d/4$ parallel planes of direction $L$, so that $\psi_{L}$ is a unit vector, and the $\psi_{L}$ form an orthonormal basis of $W_{\mathrm F}$, one vector per direction.
In particular, $\langle e_{F},\PiF e_{F}\rangle=q$ for every plane $F$, and this constant appears in every statement below.

In terms of the antisymmetric stabilizer projector, the functional in the $W_{\mathrm F}$ subspace is
\begin{equation}
   \Qa(N)=\bigl\|\PiF\,\Lambda^{4}(N)\,\PiF\bigr\|_{\mathrm{HS}}^{2},
   \label{eq:Qadef}
\end{equation}
and alongside it we use the \emph{antisymmetric stabilizer overlap}
\begin{equation}
   \cF(C)\equiv\Tr\bigl[\PiF\,\Lambda^{4}(C)\bigr]
   \label{eq:flatoverlap}
\end{equation}
of an operator $C$ on $\mathbb C^{d}$.
In the following, we show that these two quantities control the functionals $\Qsym(N)$ and $\Qtt(N)$ of the other two sectors.

\subsection{Reduction to the antisymmetric sector}
\label{ssec:reduction6}

For the symmetric sector, a Pauli--Fourier identity expresses $\Qsym(N)-\Qa(N)$ as an explicit Pauli sum that is bounded by the Pauli fourth moment $\sum_{P}[\Tr(PN^{2})]^{4}$. The fourth moment is in turn bounded by $\cF(N^{2})$ and two spectral invariants of matrix $M$, the difference being an explicit sum of squares over pairs of anticommuting Pauli operators (Appendix~\ref{app:symSOS}).
For the mixed sector, the antisymmetric Pauli operators form an orthonormal frame of $\Lambda^{2}\mathbb C^{d}$, and two elementary sums of squares over that frame bound $\Qtt(N)$ by $\Qa(N)$, $\cF(N^{2})$, and a spectral invariant (Appendix~\ref{app:mixed}).
Both inequalities hold at every $d=2^{n}$ and for every Hermitian matrix $N$, and both are saturated at $V=\mathbb 1$, where $N=M$ is diagonal.
Inserting them into Eq.~\eqref{eq:split} gives the following exact reduced form of the $\Gm(V)$ functional (Appendix~\ref{app:reduced}).

\begin{proposition}
\label{prop:master}
For every $V\in U(d)$, with $N=VMV^{\dagger}$,
\begin{equation}
 \Gm(V)=\mathcal I(\boldsymbol\mu)
 +6d^{2}\,\Qa(N)+18d\,\cF(N^{2})-\mathcal D(N),
 \label{eq:master}
\end{equation}
where
\begin{equation}
 \mathcal I(\boldsymbol\mu)\equiv7\Bigl[\sum_{x}\mu_{x}^{4}\Bigr]^{2}-6\sum_{x}\mu_{x}^{8}
 \label{eq:specinv}
\end{equation}
depends on the spectrum of $M$ only, and $\mathcal D(N)\ge0$ is a sum of three nonnegative terms, one for each certificate, given explicitly in Eq.~\eqref{eq:defectdef}, which vanishes at $V=\mathbb 1$.
\end{proposition}

It is important to stress that Eq.~\eqref{eq:master} holds for every local unitary $V$ and for every diagonal matrix $M$, that is, for every Schmidt spectrum of the state $\ket{\Psi}$, and that no assumption on the rank of $M$ has been made so far.
All dependence on $V$ resides in the two nonnegative quantities $\Qa(N)$ and $\cF(N^{2})$ and in the nonnegative $\mathcal D(N)$.
Since $\mathcal D$ can only lower $\Gm$ and vanishes at $V=\mathbb 1$, CB optimality in the form of Eq.~\eqref{eq:cbGhat} follows for every spectrum for which the combination $6d^{2}\Qa(N)+18d\,\cF(N^{2})$ is maximized at $V=\mathbb 1$.
In the rest of this section we assume that the rank of $M$ does not exceed six, which is the hypothesis of Theorem~\ref{thm:resRank6}, and we prove the latter statement under this assumption.

\subsection{Two inequalities for the antisymmetric stabilizer projector}
\label{ssec:flatineq}

The rank hypothesis is required only for the following two statements, in which $q=4/d$ is the constant of Eq.~\eqref{eq:codeform}.

\begin{theorem}[Antisymmetric-sector bound at rank at most seven]
\label{thm:rsevencap}
Let $N\succeq0$ on $\mathbb C^{d}$, $d=2^{n}$, $n\ge3$, with $\mathrm{rank}\,N\le7$. Then
\begin{equation}
 \Qa(N)\ \le\ q\,\cF(N^{2}),
 \label{eq:r7cap}
\end{equation}
with equality whenever $N$ is diagonal with support in the first seven computational-basis labels.
\end{theorem}

\begin{theorem}[CB optimality of the antisymmetric stabilizer overlap at rank at most six]
\label{thm:cfsix}
Let $N\succeq0$ on $\mathbb C^{d}$, $d=2^{n}$, $n\ge3$, with $\mathrm{rank}\,N\le6$ and eigenvalues $t_{1}\ge\cdots\ge t_{6}\ge0$, padded with zeros. Then
\begin{equation}
 \cF(N)\ \le\ q\,\bigl(t_{1}t_{2}t_{3}t_{4}+t_{1}t_{2}t_{5}t_{6}+t_{3}t_{4}t_{5}t_{6}\bigr),
 \label{eq:cf6}
\end{equation}
and the right-hand side is the CB value: it equals $\cF(T)$ for the sorted diagonal $T=\mathrm{diag}(t_{1},\dots,t_{6},0,\dots,0)$.
\end{theorem}

We prove the two theorems in Appendix~\ref{app:flatcode}. Here we show how they imply Theorem~\ref{thm:resRank6}.

\begin{proof}[Proof of Theorem~\ref{thm:resRank6}]
By Corollary~\ref{cor:cbG} it suffices to prove $\Gm(V)\le\Gm(\mathbb 1)$ for every $V\in U(d)$. Let first $n\ge3$, and let $N=VMV^{\dagger}$, of rank at most six. Proposition~\ref{prop:master} with $\mathcal D\ge0$, then Theorem~\ref{thm:rsevencap} applied to $N$ (its rank is at most $6\le7$, and $6d^{2}q=24d$), then Theorem~\ref{thm:cfsix} applied to $N^{2}=VM^{2}V^{\dagger}$ (positive semidefinite, rank at most six, sorted spectrum $\mu_{x}^{2}$) give
\begin{align}
 \Gm(V)&\le\mathcal I(\boldsymbol\mu)+6d^{2}\Qa(N)+18d\,\cF(N^{2})\notag\\
 &\le\mathcal I(\boldsymbol\mu)+42d\,\cF(N^{2})\notag\\
 &\le\mathcal I(\boldsymbol\mu)+42d\,\cF(M^{2})
 \ =\ \Gm(\mathbb 1).
 \label{eq:chain6}
\end{align}
The final equality holds because at $V=\mathbb 1$ both bounds are saturated: the defect $\mathcal D(M)$ vanishes at the computational diagonal, and the equality case of Theorem~\ref{thm:rsevencap} at $M$ (diagonal, support in the first six labels) gives $6d^{2}\Qa(M)=24d\,\cF(M^{2})$. Hence, by Proposition~\ref{prop:master} at $V=\mathbb 1$,
$\Gm(\mathbb 1)=\mathcal I(\boldsymbol\mu)+6d^{2}\Qa(M)+18d\,\cF(M^{2})=\mathcal I(\boldsymbol\mu)+42d\,\cF(M^{2})$. 
Finally, the thesis is valid also for $n <3$. At $n=1$ global CB optimality holds for every spectrum~\cite{Qian2025Nonlocal}, while for $n=2$ the space $\Lambda^{4}\mathbb C^{4}$ is one-dimensional with $\PiF$ the identity on it, so $\Qa(N)=\cF(N^{2})=(\det M)^{2}$ are constant on the orbit, the second and third steps of \eqref{eq:chain6} become equalities, and the chain reduces to the defect inequality $\Gm(V)\le\Gm(\mathbb 1)$ of Proposition~\ref{prop:master}.
\end{proof}

\subsection{The antisymmetric stabilizer subspace and the matching polytope of $K_{6}$}
\label{ssec:mechanism}

We now describe the mechanisms behind Theorems~\ref{thm:rsevencap} and~\ref{thm:cfsix}, whose proofs are detailed in Appendix~\ref{app:flatcode}.
Both rest on the basis of Eq.~\eqref{eq:codeform}.

\textit{Rank six and the matching polytope.---}
We first explain how Theorem~\ref{thm:cfsix} arises. Let $N\succeq0$ have rank at most six, and let $u_{1},\dots,u_{6}$ be orthonormal eigenvectors of $N$ with eigenvalues $t_{1}\ge\dots\ge t_{6}\ge0$.
The operator $\Lambda^{4}(N)$ is diagonal in the wedges $u_{A}\equiv u_{a_{1}}\wedge\cdots\wedge u_{a_{4}}$ of four eigenvectors, $A=\{a_{1}<\dots<a_{4}\}\subset\{1,\dots,6\}$, with eigenvalues $t_{A}\equiv\prod_{a\in A}t_{a}$.
Hence the antisymmetric stabilizer overlap is
\begin{equation}
   \cF(N)=\sum_{A}t_{A}\,p_{A},
   \qquad p_{A}\equiv\|\PiF u_{A}\|^{2}\ge0 .
   \label{eq:cFexp}
\end{equation}
The eigenvalues are fixed by the spectrum, while the eigenvectors carry the dependence on the local unitary, which therefore influences $\cF(N)$ only through the fifteen numbers $p_{A}$.
To organize them, we identify each four-subset $A$ with the complementary edge $e=\{1,\dots,6\}\setminus A$ of the complete graph $K_{6}$ on the six eigenvector labels, and we write $x_{e}\equiv p_{A}/q$.
We show in Appendix~\ref{app:flatcode} that three properties of $W_{\mathrm F}$ constrain these variables.
(i) Contracting one vector out of a vector of $W_{\mathrm F}$ is $\sqrt q$ times a signed permutation between bases, because every three labels complete uniquely to a plane. Applied to the wedge of three eigenvectors, this gives $x_{ij}+x_{ik}+x_{jk}\le1$ for every triangle $\{i,j,k\}$.
(ii) For every five-dimensional subspace $Y$, the squared norms of $\PiF$ on the five wedges of four vectors of an orthonormal basis of $Y$ sum to at most $q$. Applied to the span of five eigenvectors, this gives $\sum_{j\ne i}x_{ij}\le1$ for every vertex $i$.
(iii) The map $\alpha\otimes\psi\mapsto\alpha\wedge a_{x}\psi$, which contracts the vector $x$ out of $\psi\in W_{\mathrm F}$ and wedges the result with a two-form $\alpha$, has operator norm at most $\sqrt{2q}$.
Applied to $x=u_{r}$, this gives $\sum_{e\not\ni r}x_{e}\le2$ for every vertex $r$.

These three families are exactly the triangle, degree, and five-vertex odd-set constraints of the matching polytope of $K_{6}$, so by Edmonds' theorem~\cite{Edmonds1965matching} the vector $(x_{e})$ is a convex combination of indicator vectors of matchings. Consequently
  \begin{equation}
   \cF(N)\le q\max_{\mathcal M}\sum_{e\in\mathcal M}t_{\hat e},
  \end{equation}
where the maximum runs over the perfect matchings $\mathcal M$ of $K_{6}$ and $t_{\hat e}$ is the weight of the four-set complementary to $e$.
For sorted eigenvalues the maximum is attained by the matching $\{12,34,56\}$, whose three weights are the three products in Eq.~\eqref{eq:cf6}. This bound is the CB value.
Among the six labels $\{0,\dots,5\}$, the four-sets with vanishing sum are exactly $\{0,1,2,3\}$, $\{0,1,4,5\}$, and $\{2,3,4,5\}$, each a plane with $p_{A}=q$, so the sorted diagonal saturates Eq.~\eqref{eq:cf6}.

\textit{Rank seven and the compression bound.---}
We now turn to Theorem~\ref{thm:rsevencap}.
Its content is the operator inequality
\begin{equation}
   P_{\Lambda^{4}\mathcal T}\,\PiF\,P_{\Lambda^{4}\mathcal T}\ \preceq\ q\,P_{\Lambda^{4}\mathcal T}
   \label{eq:ocseven}
\end{equation}
for every subspace $\mathcal T\subset\mathbb C^{d}$ of dimension at most seven, where $P_{\Lambda^{4}\mathcal T}$ is the projector onto $\Lambda^{4}\mathcal T$.
The threshold has a simple geometric origin. On two parallel planes the antisymmetric stabilizer projector has a two-by-two block with all four entries equal to $q$, whose largest eigenvalue is $2q$, and two parallel planes span a three-dimensional affine subspace of eight points, which does not fit into seven dimensions.
We prove Eq.~\eqref{eq:ocseven} in Appendix~\ref{app:selfsim} by an induction over affine subspaces of $\Ftwo^{n}$, using that the antisymmetric stabilizer subspace built inside an affine subspace resolves exactly into those built inside its hyperplanes.
To derive Theorem~\ref{thm:rsevencap} from Eq.~\eqref{eq:ocseven}, let $\mathcal T=\mathrm{ran}\,N$ and consider the operators $\mathsf A\equiv\Lambda^{4}(N|_{\mathcal T})$ and $\mathsf B\equiv P_{\Lambda^{4}\mathcal T}\PiF P_{\Lambda^{4}\mathcal T}$ on $\Lambda^{4}\mathcal T$, in terms of which $\Qa(N)=\Tr(\mathsf B\mathsf A\mathsf B\mathsf A)$ and $\cF(N^{2})=\Tr(\mathsf B\mathsf A^{2})$. Since $q\,\mathbb 1-\mathsf B\succeq0$ by Eq.~\eqref{eq:ocseven}, we obtain
\begin{equation}
   q\,\cF(N^{2})-\Qa(N)
   =\bigl\|(q\,\mathbb 1-\mathsf B)^{1/2}\mathsf A\,\mathsf B^{1/2}\bigr\|_{\mathrm{HS}}^{2}\ \ge\ 0 ,
\end{equation}
which is Eq.~\eqref{eq:r7cap}.
Equality holds for a diagonal $N$ supported on seven labels.
In this case no two parallel planes fit into the support, so $\mathsf B$ is $q$ times a projector that commutes with $\mathsf A$, and the two traces coincide.

It is important to stress that each rank hypothesis is used exactly once.
Rank at most seven is needed for Eq.~\eqref{eq:ocseven}, which fails at rank eight for the reason given above. Rank at most six is needed for the matching polytope of $K_{6}$, whose facets are supplied by the three properties of $W_{\mathrm F}$.
We also note that Theorem~\ref{thm:resRank6} turns the low-rank comparator family of Sec.~\ref{sec:witness} into an unconditional certificate at every $d=2^{n}$, on the same footing as the dyadic family.

\section{Bounds on the nonlocal SRE of an arbitrary state}
\label{sec:witness}

In this section, we prove Theorems~\ref{thm:resValue} and~\ref{thm:resComparator}, which bound the nonlocal SRE of an arbitrary state.
The lower and upper bounds on $\mathcal M_{\mathrm{NL}}$ are of different nature.
Since $\mathcal M_{\mathrm{NL}}$ is a minimum over local unitaries, every pair of local unitaries at which the SRE is evaluated gives an upper bound for $\mathcal{M}_{\mathrm{NL}}$, and in particular the CB representative state yields $\mathcal M_{\mathrm{NL}}\le\mathcal M_{\mathrm{CB}}$.
A lower bound on $\mathcal M_{\mathrm{NL}}$, by contrast, is a statement about all local unitaries at once.
We obtain such statements in two ways.
We first enclose the Schmidt spectrum between two dyadic staircases, whose nonlocal SRE is known exactly by Theorem~\ref{thm:resDyadic}, and we arrive at Theorem~\ref{thm:resValue}, which bounds $\mathcal M_{\mathrm{NL}}$ from both sides.
We then approximate the spectrum of $\ket{\Psi}$ by a nearby spectrum with exactly known nonlocal SRE, which leads to Theorem~\ref{thm:resComparator}.

Throughout, $\mathcal F_{*}(\boldsymbol\mu)\equiv\max_{V_{A},V_{B}}\Fmu(V_{A},V_{B})$ denotes the optimum of the stabilizer purity, so that $\mathcal M_{\mathrm{NL}}=-\log_{2}\mathcal F_{*}$, and $\mathcal F_{\mathrm{CB}}(\boldsymbol\mu)\equiv\Fmu(\mathbb 1,\mathbb 1)=\Gm(\mathbb 1)$ is its computational-basis value, with $\mathcal M_{\mathrm{CB}}\equiv-\log_{2}\mathcal F_{\mathrm{CB}}$.
Both stabilizer purities are homogeneous of degree eight in the Schmidt amplitudes.
We therefore evaluate the same expressions also for unnormalized nonnegative vectors $\boldsymbol\nu$, for which $\mathcal F(\boldsymbol\nu)=c^{4}\mathcal F(\boldsymbol\nu/\sqrt c)$ with $c=\sum_{x}\nu_{x}^{2}$.
This extension is needed below since the shell partners and the comparator spectra are, in general, not normalized.

\subsection{Shell partners}
\label{ssec:partners}

We group the labels $x=0,\dots,d-1$ into the dyadic shells $D_{k}$ of Eq.~\eqref{eq:dyadicspec}, so that $D_{0}=\{0\}$ and $D_{k}=\{2^{k-1},\dots,2^{k}-1\}$ contains $2^{k-1}$ labels for $k\ge1$.
For a sorted nonnegative spectrum $\boldsymbol\mu$, possibly unnormalized, we define the two \emph{shell partners} by rounding every shell to one of its endpoints,
\begin{equation}
 (\nuup)_{x}\equiv\mu_{2^{k-1}},
 \qquad
 (\nudown)_{x}\equiv\mu_{2^{k}-1},
 \qquad x\in D_{k},\ k\ge1,
 \label{eq:shellpartners}
\end{equation}
with $(\nuup)_{0}=(\nudown)_{0}\equiv\mu_{0}$.
Both are sorted dyadic staircases, and they enclose the target entrywise, $\nudown\le\boldsymbol\mu\le\nuup$, as the largest staircase below it and the smallest staircase above it [Fig.~\ref{fig:witnessgeom}(a)].

\begin{figure}[t]
\centering
\includegraphics[width=0.96\columnwidth]{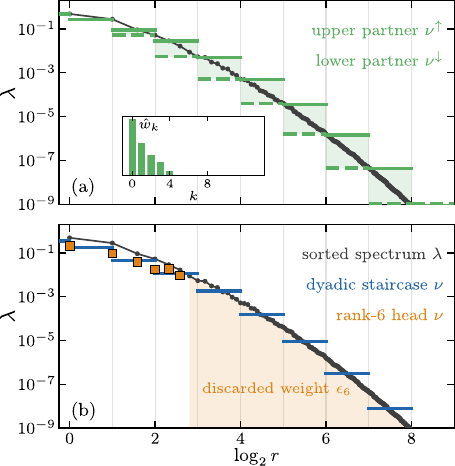}
\caption{Shell partners and comparators for the critical Ising (see Sec.~\ref{ssec:ising}) half-chain  spectrum at $N=3\times10^{6}$ (gray), on a $\log_{2}$ axis of the rank $r$, on which each dyadic shell is a unit column. (a)~The upper and lower partners $\nuup$ (solid) and $\nudown$ (dashed) of Eq.~\eqref{eq:shellpartners} enclose the target. Inset: the shell masses $\hat w_{k}$ of $\nuup$, which enter $Q$. (b)~The two comparator families of Sec.~\ref{ssec:comparator}: the optimized dyadic staircase (blue) follows the target across all scales, whereas the rank-six head (orange) keeps the six leading amplitudes and discards the shaded weight $\epsilon_{6}$.}
\label{fig:witnessgeom}
\end{figure}

To turn the shell partners $\nuup,\nudown$ into bounds, we use that $\mathcal F_{\mathrm{CB}}$ is nondecreasing under entrywise increases of a nonnegative spectrum.
Indeed, by Eq.~\eqref{eq:Phimatel} the operator $\PW M^{\otimes4}\PW$ is diagonal in the $\Phi$-basis, with entries that are polynomials with nonnegative coefficients in the amplitudes $\mu_{x}$, and $\mathcal F_{\mathrm{CB}}=\Gm(\mathbb 1)$ is $d^{2}$ times the sum of their squares.
Since the partners are dyadic staircases, we further make use of the fact that the computational-basis value of a staircase is controlled from both sides by its shell masses, as stated in the following theorem.

\begin{theorem}[Stabilizer purity of dyadic staircases]
\label{thm:shellpurity}
For every dyadic-staircase spectrum $\boldsymbol\nu$, possibly unnormalized, with shell masses $w_{k}\equiv\sum_{x\in D_{k}}\nu_{x}^{2}$,
\begin{equation}
 \sum_{k}w_{k}^{4}\ \le\ \mathcal F_{\mathrm{CB}}(\boldsymbol\nu)\ \le\ C_{\mathrm{sh}}\sum_{k}w_{k}^{4},
 \label{eq:shellpurity}
\end{equation}
where $C_{\mathrm{sh}}\equiv\tfrac{790}3+120\sqrt2+60\sqrt3+48\sqrt6<2^{10}$ is an absolute constant.
\end{theorem}

The proof is given in Appendix~\ref{app:shell}.
It evaluates $\mathcal F_{\mathrm{CB}}$ of a staircase through the channel decomposition of Eq.~\eqref{eq:channels} and bounds each channel by the shell masses.
We apply it below to the two partners.
On the shell $D_{k}$, every amplitude of $\nuup$ equals $\mu_{2^{k-1}}$ and every amplitude of $\nudown$ equals $\mu_{2^{k}-1}$, so that in terms of the probabilities  $\lambda_{x}=\mu_{x}^{2}$ their shell masses in Eq.~\eqref{eq:shellpurity} are
\begin{equation}
 \hat w_{k}\equiv2^{k-1}\lambda_{2^{k-1}},
 \qquad
 \check w_{k}\equiv2^{k-1}\lambda_{2^{k}-1},
 \qquad
 \hat w_{0}=\check w_{0}\equiv\lambda_{0}.
 \label{eq:shellmasses}
\end{equation}
The participation $Q\equiv\sum_{k}\hat w_{k}^{4}$ of Theorem~\ref{thm:resValue} is built from the masses of the upper partner.

\subsection{The nonlocal SRE up to bounded constants}
\label{ssec:value}

Theorem~\ref{thm:resValue} makes two statements.
The nonlocal SRE differs from the computational-basis value by a constant $c_{1}\in[0,4]$, and the computational-basis value differs from $-\log_{2}Q$ by a constant $c_{2}\in[-\log_{2}C_{\mathrm{sh}},4]$.
We prove both by enclosing the target between its partners.
The upper partner controls the optimum $\mathcal F_{*}(\boldsymbol\mu)$ from above, the lower partner controls the computational-basis value from below, and Theorem~\ref{thm:shellpurity} converts both into shell masses.
We begin by comparing the optimum of the target with the computational-basis value of the upper shell partner,
\begin{equation}
   \mathcal F_{\mathrm{CB}}(\boldsymbol\mu)
   \ \le\ \mathcal F_{*}(\boldsymbol\mu)
   \ \le\ \mathcal F_{*}(\nuup)
   \ =\ \mathcal F_{\mathrm{CB}}(\nuup).
   \label{eq:shellsandwich}
\end{equation}
The first inequality is the evaluation at $V_{A}=V_{B}=\mathbb 1$.
For the second, we use Lemma~\ref{lem:diag}, which reduces both optima to maxima of the single-unitary functional, and the entrywise domination $\boldsymbol\mu\le\nuup$.
With $M=\mathrm{diag}(\boldsymbol\mu)$ and $M^{\uparrow}=\mathrm{diag}(\nuup)$, we have $0\le VMV^{\dagger}\le VM^{\uparrow}V^{\dagger}$ at every $V$, hence the same order for the fourth tensor powers, for their compressions by $\PW$, and for the squared Hilbert--Schmidt norms of the compressions, so that $\Gm(V)\le G_{\nuup}(V)$ at every $V$ and therefore $\mathcal F_{*}(\boldsymbol\mu)\le\mathcal F_{*}(\nuup)$.
The final equality is Theorem~\ref{thm:resDyadic}, since $\nuup$ is a dyadic staircase.

Next we show that the upper partner exceeds the target by at most a factor of sixteen,
\begin{equation}
   \mathcal F_{\mathrm{CB}}(\nuup)\ \le\ 16\,\mathcal F_{\mathrm{CB}}(\boldsymbol\mu).
   \label{eq:doubling}
\end{equation}
Let $\boldsymbol\eta$ be the restriction of $\nudown$ to the labels $x<d/2$, a sorted spectrum on $n-1$ qubits per party.
Doubling every entry of $\boldsymbol\eta$ gives a sorted spectrum on $n$ qubits per party, which coincides with $\boldsymbol\eta\otimes(1,1)$ up to a permutation of the qubits, a Clifford operation that leaves $\mathcal F_{\mathrm{CB}}$ unchanged.
The doubled spectrum dominates $\nuup$ entrywise.
On the labels $0$ and $1$ it equals $\mu_{0}$, while $\nuup$ equals $\mu_{0}$ and $\mu_{1}$.
On the shell $D_{k}$ with $k\ge2$ it equals the smallest amplitude $\mu_{2^{k-1}-1}$ of the preceding shell, while $\nuup$ equals $\mu_{2^{k-1}}\le\mu_{2^{k-1}-1}$.
Since $\mathcal F_{\mathrm{CB}}$ is entrywise nondecreasing, multiplicative over tensor factors, and of degree eight, we obtain $\mathcal F_{\mathrm{CB}}(\nuup)\le\mathcal F_{\mathrm{CB}}(\boldsymbol\eta\otimes(1,1))=16\,\mathcal F_{\mathrm{CB}}(\boldsymbol\eta)$, where $16=\mathcal F_{\mathrm{CB}}((1,1))$ is the stabilizer purity of the unnormalized Bell pair.
In the other direction, $(1,0)\otimes\boldsymbol\eta$ is the copy of $\boldsymbol\eta$ padded with zeros to dimension $d$, which lies entrywise below $\nudown\le\boldsymbol\mu$, so $\mathcal F_{\mathrm{CB}}(\boldsymbol\eta)\le\mathcal F_{\mathrm{CB}}(\boldsymbol\mu)$, and Eq.~\eqref{eq:doubling} follows.
Combining Eqs.~\eqref{eq:shellsandwich} and~\eqref{eq:doubling} gives $\mathcal F_{\mathrm{CB}}\le\mathcal F_{*}\le16\,\mathcal F_{\mathrm{CB}}$, that is,
\begin{equation}
   \mathcal M_{\mathrm{CB}}-4\ \le\ \mathcal M_{\mathrm{NL}}\ \le\ \mathcal M_{\mathrm{CB}},
   \label{eq:cbgapM}
\end{equation}
which is the first statement of Theorem~\ref{thm:resValue}, with $0\le c_{1}\le4$. The factor of sixteen in~\eqref{eq:doubling} is sharp for this comparison, since for the flat spectrum of rank $2^{m}+1$ the ratio $\mathcal F_{\mathrm{CB}}(\nuup)/\mathcal F_{\mathrm{CB}}(\boldsymbol\mu)$ tends to $16$ as $m\to\infty$.
It bounds the rounding to the upper partner, not the optimization itself, and Conjecture~\ref{conj:cb} asserts that $c_{1}=0$.

Finally, we express $\mathcal F_{\mathrm{CB}}(\boldsymbol\mu)$ through $Q$.
Since the lower partner lies entrywise below the target, the fact that $\mathcal F_{\mathrm{CB}}$ is nondecreasing and the lower inequality in Eq.~\eqref{eq:shellpurity}, applied to $\nudown$, give
\begin{equation}
   \frac1{16}\,Q
   \ \le\ \sum_{k}\check w_{k}^{4}
   \ \le\ \mathcal F_{\mathrm{CB}}(\nudown)
   \ \le\ \mathcal F_{\mathrm{CB}}(\boldsymbol\mu).
   \label{eq:shellfloor}
\end{equation}
The first inequality relates the masses of the two partners.
Because the spectrum is sorted, the largest probability of a shell does not exceed the smallest probability of the preceding shell, so that $\hat w_{1}=\lambda_{1}\le\lambda_{0}=\check w_{0}$ and $\hat w_{k}\le2\check w_{k-1}$ for $k\ge2$, where the factor two is the ratio of the sizes of consecutive shells.
Raising these inequalities to the fourth power and summing over $k$ gives $\sum_{k}\hat w_{k}^{4}\le16\sum_{k}\check w_{k}^{4}$.
In the other direction, Eq.~\eqref{eq:shellsandwich} and the upper inequality in Eq.~\eqref{eq:shellpurity}, applied to $\nuup$, give $\mathcal F_{\mathrm{CB}}(\boldsymbol\mu)\le\mathcal F_{*}(\boldsymbol\mu)\le C_{\mathrm{sh}}Q$.
Writing both bounds in terms of $\mathcal M_{\mathrm{NL}}=-\log_{2}\mathcal F_{*}$ and $\mathcal M_{\mathrm{CB}}=-\log_{2}\mathcal F_{\mathrm{CB}}$, we obtain
\begin{equation}
   -\log_{2}Q-\log_{2}C_{\mathrm{sh}}
   \ \le\ \mathcal M_{\mathrm{NL}}
   \ \le\ \mathcal M_{\mathrm{CB}}
   \ \le\ -\log_{2}Q+4,
   \label{eq:shellparticipation}
\end{equation}
which is the second statement of Theorem~\ref{thm:resValue}, with $- \log_{2}C_{\mathrm{sh}}\le c_{2}\le4$.
This completes the proof.

Three consequences of Eq.~\eqref{eq:shellparticipation} will be used in Sec.~\ref{sec:allrank}.
The constants in Eq.~\eqref{eq:shellparticipation} do not depend on the spectrum, so for any family of states we have
\begin{equation}
   \mathcal M_{\mathrm{NL}}
   =\mathcal M_{\mathrm{CB}}+O(1)
   =-\log_{2}Q+O(1),
   \label{eq:cbgapscaling}
\end{equation}
and the growth of the nonlocal SRE with the system size, both its functional form and its leading coefficient, is independent of whether Conjecture~\ref{conj:cb} is true.
Moreover, $Q$ involves one probability per shell, namely the largest one, so the $n{+}1$ probabilities $\lambda_{0},\lambda_{1},\lambda_{2},\lambda_{4},\dots,\lambda_{2^{n-1}}$ determine the nonlocal SRE of any state of $2n$ qubits to the accuracy of Eq.~\eqref{eq:shellparticipation}.
We can also bound the nonlocal SRE by a single shell mass on each side.
Because the spectrum is sorted, the largest probability of the shell $D_{k}$ is at most the average probability of the shell $D_{k-1}$, so that $\hat w_{k}\le2\sum_{x\in D_{k-1}}\lambda_{x}$ for $k\ge2$, and together with $\hat w_{0}=\lambda_{0}$ and $\hat w_{1}\le\lambda_{0}$ this gives $\sum_{k}\hat w_{k}\le2$ for a normalized spectrum.
Hence $Q\le2(\max_{k}\hat w_{k})^{3}$, while trivially $\sum_{k}\check w_{k}^{4}\ge(\max_{k}\check w_{k})^{4}$.
Inserting these two estimates into Eqs.~\eqref{eq:shellfloor} and~\eqref{eq:shellsandwich} and using $2C_{\mathrm{sh}}<2^{11}$, we obtain
\begin{equation}
   3\log_{2}\frac1{\max_{k}\hat w_{k}}-11
   \ \le\ \mathcal M_{\mathrm{NL}}
   \ \le\ \mathcal M_{\mathrm{CB}}
   \ \le\ 4\log_{2}\frac1{\max_{k}\check w_{k}}.
   \label{eq:shellwindow}
\end{equation}
The lower bound grows when the weight of the spectrum is spread over many shells, so that no single shell mass is large, whereas the upper bound stays finite as soon as one shell carries a finite fraction of the weight.

The number of shells that carry appreciable mass is limited by their total number $n+1$ and by the entanglement entropy of the cut, and Eq.~\eqref{eq:shellfloor} converts either limit into an upper bound on the nonlocal SRE.
A spectrum with entanglement entropy $S_{1}$ keeps more than half of its weight on $O(S_{1})$ dyadic shells.  The fourth moment of the masses of these shells is at least their total mass to the fourth power divided by the cube of their number. Equation~\eqref{eq:shellfloor} then turns this fourth moment into a lower bound on the stabilizer purity, and hence into an upper bound on the nonlocal SRE.

\begin{theorem}[Nonlocal SRE is upper bounded by entanglement entropy]
  \label{thm:ceiling}
  Let $\ket{\Psi}$ be a state of $N=n_{A}+n_{B}$ qubits with sorted Schmidt spectrum $\boldsymbol\mu$ and entanglement entropy $S_{1}\equiv-\sum_{x}\lambda_{x}\log_{2}\lambda_{x}$ across
  the cut.
  Then
  \begin{align}
   \mathcal F_{\mathrm{CB}}(\boldsymbol\mu)\ &\ge\ \frac1{256\,(2S_{1}+3)^{3}},
   \qquad\text{hence}\notag\\
   \mathcal M_{\mathrm{NL}}\ \le\ \mathcal M_{\mathrm{CB}}\ &\le\ 3\log_{2}(2S_{1}+3)+8 .
   \label{eq:ceiling}
  \end{align}
  Moreover, with $n\equiv\min(n_{A},n_{B})$ the number of qubits of the smaller party, the Schmidt rank is at most $2^{n}$, so that $S_{1}\le n$, and counting the $n+1$ dyadic shells
  directly gives
  \begin{equation}
   \mathcal F_{\mathrm{CB}}(\boldsymbol\mu)\ \ge\ \frac1{16\,(n+1)^{3}},
   \qquad
   \mathcal M_{\mathrm{NL}}\ \le\ 3\log_{2}(n+1)+4 .
   \label{eq:ceilingn}
  \end{equation}
  \end{theorem}

The proof is given in Appendix~\ref{app:shell}.
No bipartite pure state of $2n$ qubits has a nonlocal SRE larger than $3\log_{2}n+O(1)$, irrespectively of its, possibly extensive, scaling of  SRE with system size.
An area law implies $\mathcal M_{\mathrm{NL}}=O(1)$, the logarithmic entanglement $S_{1}=\Theta(\ln N)$ of a one-dimensional critical state implies $\mathcal M_{\mathrm{NL}}\le3\log_{2}\ln N+O(1)$, and even a volume law only allows $\mathcal M_{\mathrm{NL}}\le3\log_{2}n+O(1)$.
The theorem uses the entropy only, and it therefore gives away information.
When the full Schmidt spectrum is known, the lower bound $\sum_{k}\check w_{k}^{4}$ on the stabilizer purity in Eq.~\eqref{eq:shellfloor} is sharper than the bound of Eq.~\eqref{eq:ceiling}, which replaces it by a function of $S_{1}$.

We note that corollary~4.1 of Ref.~\cite{Cao2025magical} (see also~\cite{Torre2026Spectrum}) gives, in base-two units,
\begin{equation}
  \mathcal M_{\mathrm{NL}}
  \le \min\bigl\{2S_{2},\,4(S_{0}-S_{1/2})\bigr\},
  \label{eq:caoceiling}
\end{equation}
where $S_{\alpha}$ are the R\'enyi entropies of the cut and $S_{0}$ is the logarithm of the Schmidt rank.
The first branch is a bound by the entanglement, $2S_{2}\le2S_{1}$, and therefore allows the nonlocal SRE to grow linearly with the entropy, whereas Eq.~\eqref{eq:ceiling} allows only $3\log_{2}S_{1}+O(1)$.
In terms of the stabilizer purity, Ref.~\cite{Cao2025magical} gives $\mathcal F_{*}\ge2^{-2S_{2}}$, while Theorem~\ref{thm:ceiling} gives $\mathcal F_{*}\ge[256(2S_{1}+3)^{3}]^{-1}$, so the improvement is exponential in the entropy.
For a one-dimensional critical state the two bounds scale as $O(\ln N)$ and $O(\log\ln N)$, and under a volume law as $O(n)$ and $O(\log N)$.

\subsection{Comparator bounds}
\label{ssec:comparator}
\label{ssec:principle}

Theorem~\ref{thm:resValue} fixes $c_{1}$ only within an interval of size $4$, which is the same for every spectrum and takes no advantage of the spectrum at hand.
Here, for an individual spectrum, we obtain a sharper lower bound by comparing the state $\ket{\Psi}$ with a nearby \emph{comparator state} $\ket{\phi_{\boldsymbol\nu}}=\sum_{x}\nu_{x}\ket{x,x}$, with $\boldsymbol\nu\ge0$ sorted and possibly unnormalized, whose optimal stabilizer purity $\mathcal F_{*}(\boldsymbol\nu)$ is known exactly.
To compare the two states, we write the stabilizer purity as a sum over the Pauli operators of the full $2n$-qubit system.
With $d^{2}=4^{n}$ the dimension of its Hilbert space, $\mathcal P_{2n}$ the Pauli basis, $U=V_{A}\otimes V_{B}^{*}$, and $\rho=\ketbra{\Psi}{\Psi}$ for $\ket{\Psi}=\sum_{x}\mu_{x}\ket{x,x}$,
\begin{equation}
   \mathcal F(V;\rho)=\frac1{d^{2}}\sum_{R\in\mathcal P_{2n}}\bigl[\Tr(RU\rho U^{\dagger})\bigr]^{4},
   \label{eq:pauliform}
\end{equation}
which reproduces $\Fmu(V_{A},V_{B})$ term by term.
At fixed local unitaries, $\mathcal F^{1/4}$ is therefore an $\ell_{4}$ norm of a linear image of the density matrix.
By the triangle inequality, the stabilizer purities of target and comparator differ, at every pair of local unitaries simultaneously, by at most a norm of the difference $\Delta=\ketbra{\Psi}{\Psi}-\ketbra{\phi_{\boldsymbol\nu}}{\phi_{\boldsymbol\nu}}$, and this norm is bounded by the two unitary invariants of the rank-two operator $\Delta$ (Fig.~\ref{fig:comparator}).
Since the bound is uniform over the unitary orbit, it transfers between the two optima even though they are attained at different unitaries.

\begin{figure}[t]
\centering
\includegraphics[width=\columnwidth]{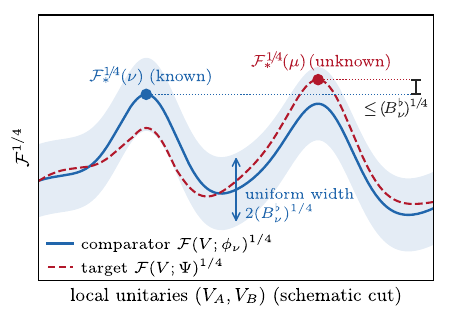}
\caption{The comparator bound. Over the local-unitary landscape
(schematic one-dimensional cut), the stabilizer purity of the target (dashed) deviates from that of a comparator (solid) by at most the uniform error term
$(B^{\flat}_{\nu})^{1/4}$ on the norm scale $\mathcal F^{1/4}$. The two
maxima may be attained at different unitaries, but the same error term bounds
their separation (right bracket): the unknown optimum
$\mathcal F_{*}(\boldsymbol\mu)$ lies within $(B^{\flat}_{\nu})^{1/4}$ of the
known optimum $\mathcal F_{*}(\boldsymbol\nu)$ on this scale, which yields a
lower bound on $\mathcal M_{\mathrm{NL}}$.}
\label{fig:comparator}
\end{figure}

\begin{theorem}[Comparator bound]
\label{thm:witness}
For a sorted comparator spectrum $\boldsymbol\nu\ge0$, possibly unnormalized, let $c_{\nu}=\sum_{x}\nu_{x}^{2}$, $z_{\nu}=(\sum_{x}\mu_{x}\nu_{x})^{2}$, and
\begin{equation}
 B^{\flat}_{\nu}=\bigl(1+c_{\nu}^{2}-2z_{\nu}\bigr)\bigl((1+c_{\nu})^{2}-4z_{\nu}\bigr).
 \label{eq:Bflat}
\end{equation}
Then, for every normalized $\boldsymbol\mu$,
\begin{equation}
 \bigl|\mathcal F_{*}(\boldsymbol\mu)^{1/4}-\mathcal F_{*}(\boldsymbol\nu)^{1/4}\bigr|\;\le\;\bigl(B^{\flat}_{\nu}\bigr)^{1/4}.
 \label{eq:transfer}
\end{equation}
If moreover $\hat{\boldsymbol\nu}=\boldsymbol\nu/\sqrt{c_{\nu}}$ is a dyadic staircase or has rank at most six, so that $\mathcal F_{*}(\hat{\boldsymbol\nu})=\mathcal F_{\mathrm{CB}}(\hat{\boldsymbol\nu})$ by Theorems~\ref{thm:resDyadic} and~\ref{thm:resRank6}, then
\begin{align}
 \mathcal M_{\mathrm{NL}}(\boldsymbol\mu)\;&\ge\;-\log_{2}U(\boldsymbol\nu),\notag\\
 U(\boldsymbol\nu)&=\min\Bigl\{1,\bigl[c_{\nu}\,\mathcal F_{\mathrm{CB}}(\hat{\boldsymbol\nu})^{1/4}+\bigl(B^{\flat}_{\nu}\bigr)^{1/4}\bigr]^{4}\Bigr\},
 \label{eq:Unu}
\end{align}
and consequently, for every family $\mathcal C$ of such comparators,
\begin{equation}
 -\log_{2}\inf_{\boldsymbol\nu\in\mathcal C}U(\boldsymbol\nu)\;\le\;\mathcal M_{\mathrm{NL}}(\boldsymbol\mu).
 \label{eq:bracket}
\end{equation}
\end{theorem}

The proof is given in Appendix~\ref{app:witness}, and Eq.~\eqref{eq:Unu} is the statement of Theorem~\ref{thm:resComparator}.
Theorem~\ref{thm:witness} replaces the minimization over the local unitaries by a search for a good comparator, and every comparator that is tried yields a valid lower bound, since the reported number is the exact evaluation of $U(\boldsymbol\nu)$.
A search that misses the best comparator therefore only loses tightness, while still resulting in a valid lower bound for the nonlocal SRE $\mathcal{M}_{NL}$.
By construction, the comparator states belong to the families in which the  exactness of the nonlocal SRE solution, $\mathcal F_{*}(\hat{\boldsymbol\nu})=\mathcal F_{\mathrm{CB}}(\hat{\boldsymbol\nu})$, is proved.

\textit{Dyadic staircases.---}
\label{ssec:dwitness}
The first family is $\mathcal C=\mathcal K_{n}$, the dyadic staircases of Sec.~\ref{sec:dyadic} [Fig.~\ref{fig:witnessgeom}(a)].
Theorem~\ref{thm:witness} gives
\begin{equation}
 \mathcal M_{\mathrm{NL}}(\boldsymbol\mu)\ \ge\ -\log_{2}\inf_{\boldsymbol\nu\in\mathcal K_{n}}U(\boldsymbol\nu),
 \label{eq:dyadicfloor}
\end{equation}
an optimization over the $n{+}1$ shell values of the comparator, in which every candidate $\boldsymbol\nu$ contributes the exactly evaluated $U(\boldsymbol\nu)$, with $\mathcal F_{\mathrm{CB}}(\hat{\boldsymbol\nu})$ from Eq.~\eqref{eq:channels}.
The optimization itself may be performed by any heuristic method.

\textit{Spectra of rank at most six.---}
\label{ssec:rkwitness}
The second family is $\mathcal C=\mathcal R_{k}=\{\boldsymbol\nu\ge0:\ \mathrm{rank}\,\boldsymbol\nu\le k\}$ with $k\le6$ [Fig.~\ref{fig:witnessgeom}(b)].
Its computational-basis value is independent of the dimension: placing the $k$ amplitudes on the first labels of $\Ftwo^{m}$ with $k\le2^{m}$,
\begin{equation}
 \mathcal F_{\mathrm{CB}}(\hat{\boldsymbol\nu})
 =\frac1{2^{m}}\sum_{a,c\in\Ftwo^{m}}
 \Bigl[\sum_{z}(-1)^{c\cdot z}\hat\nu_{z}\hat\nu_{z\oplus a}\Bigr]^{4},
 \label{eq:walsh}
\end{equation}
which is evaluated by fast Walsh--Hadamard transforms~\cite{Huang2025XOR,Sierant2026Computing,Xiao2026Sampling}.
The simplest member of the family requires no search.
  Taking $\boldsymbol\nu$ to be the $k$ leading amplitudes of $\boldsymbol\mu$ itself, with discarded weight $\epsilon_{k}\equiv\sum_{x\ge k}\mu_{x}^{2}$, we have $c_{\nu}=1-\epsilon_{k}$
  and $z_{\nu}=(1-\epsilon_{k})^{2}$, so that Eq.~\eqref{eq:transfer} becomes
\begin{equation}
  \begin{split}
   \bigl|\mathcal F_{*}(\boldsymbol\mu)^{1/4}-(1-\epsilon_{k})\,\mathcal F_{\mathrm{CB}}(\hat{\boldsymbol\nu})^{1/4}\bigr|
   &\le\beta(\epsilon_{k}),
  \end{split}
  \label{eq:headcert}
  \end{equation}
where 
   $\beta(\epsilon)\equiv\sqrt{\epsilon}\,\bigl(8-10\epsilon+3\epsilon^{2}\bigr)^{1/4}$.

The two families are suited to different spectra.
When most of the weight is concentrated on a few leading Schmidt values with distinct magnitudes, as is typical of area-law states, the six-term truncation keeps these values individually and loses only the discarded weight $\epsilon_{6}$, whereas a dyadic comparator is forced to take a single value on every shell.
We can quantify this advantage when $\epsilon_{6}\ll1$.
Let $X\equiv(1-\epsilon_{6})\,\mathcal F_{\mathrm{CB}}(\hat{\boldsymbol\nu})^{1/4}$ for the six-term truncation and $\beta_{6}\equiv\beta(\epsilon_{6})$.
Equation~\eqref{eq:headcert} places $\mathcal F_{*}(\boldsymbol\mu)^{1/4}$ in the interval $[X-\beta_{6},X+\beta_{6}]$.
The inequality behind Eq.~\eqref{eq:transfer} holds at every fixed pair of local unitaries, so applying it at $V_{A}=V_{B}=\mathbb 1$ places $\mathcal F_{\mathrm{CB}}(\boldsymbol\mu)^{1/4}$ in the same interval.
For $X>\beta_{6}$ we therefore obtain
  \begin{equation}
   \mathcal M_{\mathrm{CB}}-\mathcal M_{\mathrm{NL}}\ \le\ 4\log_{2}\frac{X+\beta_{6}}{X-\beta_{6}}
   \ =\ O\bigl(\sqrt{\epsilon_{6}}\bigr),
   \label{eq:arealaw}
  \end{equation}
so that the six leading Schmidt values determine the nonlocal SRE within an explicit error whenever the discarded weight is small.
When the weight of the spectrum is spread over many shells, as for the critical spectrum of Fig.~\ref{fig:witnessgeom}(a), the discarded weight $\epsilon_{6}$ grows with the system size, and the dyadic comparators give the better bound.


\section{Nonlocal SRE in many-body systems}
\label{sec:allrank}

In this section, we apply the bounds of Sec.~\ref{sec:witness} to two families of many-body states.
We begin with dimerized spin chains, in which the number of bonds crossing the cut sets the entanglement entropy and the bond weight sets the Schmidt spectrum.
For this family we show that states with the same logarithmic entanglement entropy have nonlocal SRE ranging from zero to $\frac32\log_2\ln N$, and that at a fixed number of crossing bonds the nonlocal SRE takes exact values in the area-law regime.
We then study the ground state of the transverse-field Ising chain, in which the nonlocal SRE grows as $\frac32\log_2\ln N$ at the critical point and reaches finite plateaus in the gapped phases.

\subsection{Dimer chains: from critical to area-law entanglement}
  \label{ssec:dimer}

We consider states of $N$ qubits in which the qubits are grouped into pairs, called dimers, and each dimer is in the same two-qubit state $\ket{\phi}$,
\begin{equation}
   \ket{\Psi}=\bigotimes_{\text{dimers }(i,j)}\ket{\phi}_{ij},
   \qquad
   \ket{\phi}=\sqrt p\,\ket{0,0}+\sqrt q\,\ket{1,1},
   \label{eq:dimerchain}
\end{equation}
with $\frac12\le p\le1$ and $q=1-p$.
Since every two-qubit state has Schmidt rank at most two, Eq.~\eqref{eq:dimerchain} covers every product of identical two-qubit states up to local unitaries.
A dimer that lies entirely in $A$ or entirely in $B$ does not contribute to the entanglement across the cut, so the Schmidt spectrum of $\ket{\Psi}$ is the $m$-fold tensor product of the spectrum $(\sqrt p,\sqrt q)$ of a single dimer, where $m$ is the number of dimers crossing the cut.
The number $m$ sets the entanglement entropy, $S_1=m\,h_2(p)$ with $h_2$ the binary entropy, while $p$ sets the shape of the Schmidt spectrum.
We first show that for $p=\frac12$ the nonlocal SRE vanishes for every $m$, whereas for $\frac12<p<1$ it equals $\frac32\log_2 m+O(1)$, and we then evaluate it exactly at fixed $m$.

The case $p=\frac12$ is realized by the random-singlet fixed point of strongly disordered spin chains~\cite{RefaelMoore2004}, in which the strong-disorder renormalization group describes the ground state as a product of two-qubit singlets.
If $m$ singlets cross the cut, the state is equivalent under local unitaries to $m$ Bell pairs together with unentangled qubits.
It is therefore a stabilizer state up to local unitaries, and
  \begin{equation}
   S_1=m,
   \qquad
   \mathcal M_{\mathrm{NL}}=0
   \label{eq:rszero}
  \end{equation}
exactly, for every realization of the disorder and every $m$.
At the fixed point, the mean number of crossing singlets grows as $\overline m=\Theta(\ln N)$~\cite{RefaelMoore2004}.
The entanglement entropy of this family thus scales as in a critical chain, while the nonlocal SRE vanishes identically, because the entropy is carried by an increasing number of Bell pairs with flat Schmidt spectra rather than by a Schmidt spectrum spread over many scales.
This statement concerns the idealized fixed-point wave function, and finite-disorder corrections deform the spectrum.

Replacing the singlets by imperfect dimers changes this behavior completely.
  
\begin{theorem}[Nonlocal SRE of imperfect dimers]
  \label{thm:dimer}
  Let $\ket{\Psi}$ be the state of Eq.~\eqref{eq:dimerchain} with $\frac12<p<1$ and $m$ dimers crossing the cut.
  Then, without assuming Conjecture~\ref{conj:cb},
  \begin{equation}
   \mathcal M_{\mathrm{NL}}(\Psi)
   =\frac32\log_2 m+O(1),
   \quad
   \mathcal M_{\mathrm{CB}}-\mathcal M_{\mathrm{NL}}=O(1),
   \label{eq:dimerlaw}
  \end{equation}
  where the constants implied by $O(1)$ depend on $p$ but not on $m$.
\end{theorem}

We prove Theorem~\ref{thm:dimer} in Appendix~\ref{app:dimer} with explicit constants.
The Schmidt probabilities of $\ket{\Psi}$ are $p^{m-j}q^{j}$ with multiplicity $\binom mj$, so that the total weight of level $j$ is given by the binomial distribution $\mathrm{Bin}(m,q)$.
Most of the weight is concentrated on $O(\sqrt m)$ levels of mass $O(m^{-1/2})$ each, and the probabilities of consecutive levels differ by the fixed factor $p/q$.
The weight is therefore spread over $\Theta(\sqrt m)$ dyadic shells, which gives $Q=\sum_k\hat w_k^4=\Theta(m^{-3/2})$, and Eq.~\eqref{eq:shellparticipation} yields Eq.~\eqref{eq:dimerlaw}.

For $m=\Theta(\ln N)$ crossing dimers, the entanglement entropy is $S_1=\Theta(\ln N)$ and $\mathcal M_{\mathrm{NL}}=\frac32\log_2\ln N+O(1)$.
In the basis of the individual dimers, the SRE is additive, $\mathcal M_2=m\,[-\log_2\gamma(p)]=\Theta(\ln N)$, where
\begin{equation}
   \gamma(p)\equiv1-4pq+16p^{2}q^{2}
   \label{eq:gamma2}
\end{equation}
is the stabilizer purity of a single dimer $\ket{\phi}$ in the computational basis.
Local unitaries thus remove all but a $\Theta(\log\log N)$ part of the SRE computed in the basis of the dimers.
As $p\to\frac12$, the constant in Eq.~\eqref{eq:dimerlaw} diverges, and the crossover to the vanishing value of Eq.~\eqref{eq:rszero} is controlled by $|p-\frac12|\sqrt m$.

For a cut crossing one or two dimers, the Schmidt rank is at most four, and Theorem~\ref{thm:resRank6} gives the nonlocal SRE exactly, $\mathcal M_{\mathrm{NL}}=-\log_2\gamma(p)$ for one dimer and $-2\log_2\gamma(p)$ for two, since the rank-four spectrum $(p^{2},pq,pq,q^{2})$ is already sorted in the tensor-product order.
This is the area-law regime of the family, realized by ground states of sums of commuting two-qubit projectors and by dimerized spin-$\frac12$ chains~\cite{Asoudeh2007dimer}, with the Majumdar--Ghosh chain~\cite{MajumdarGhosh1969} as the singlet case $p=\frac12$.

\subsection{The transverse-field Ising chain}
\label{ssec:critical}
\label{ssec:ising}
\label{ssec:gapped}

We consider the transverse-field Ising chain of $N$ qubits,
\begin{equation}
   H=-\sum_{i=1}^{N-1}X_{i}X_{i+1}-g\sum_{i=1}^NZ_{i},
   \label{eq:tfim}
\end{equation}
where $X_{i},Z_{i}\in\mathcal P_{N}$ are the Pauli strings that act as the Pauli matrices $X$ and $Z$ on qubit $i$ and as the identity on the other qubits.
The ground state is ferromagnetic for $g<1$, paramagnetic for $g>1$, and critical at $g=1$.
We study its nonlocal SRE across the cut between the two halves of the chain.
The SRE of this ground state in a fixed basis was studied in Refs.~\cite{Oliviero2022Ising,Haug23quanti}, whereas here we bound the nonlocal SRE, which involves the minimization over all local unitaries, in terms of the Schmidt spectrum alone.

\textit{Numerical approach.---}
The ground state of the transverse-field Ising chain~\eqref{eq:tfim} is a fermionic Gaussian state, so the reduced density matrix of half of the chain is fixed by the occupations of the single-particle modes of the entanglement Hamiltonian~\cite{Peschel2003RDM}, which we compute from the Majorana covariance matrix of the chain.
For $N\le10^{4}$ we diagonalize the $2N\times2N$ Majorana coupling matrix at every $g$ and obtain the Schmidt spectrum of the equal bipartition of the ground state.
At the critical point, $g=1$, the singular value decomposition of the coupling matrix is known in closed form for every $N$~\cite{Pfeuty1970}, and the resulting correlation matrix of the half chain is a sum of a Toeplitz and a Hankel matrix~\cite{FagottiCalabrese2011}, which allows us to obtain the Schmidt spectrum by a Lanczos iteration for system sizes up to $N=3\times10^{6}$.
For larger critical chains we use a model of the single-particle entanglement energies, whose leading term is the equidistant ladder with spacing proportional to $1/\ln N$~\cite{PeschelKaulkeLegeza1999,PeschelEisler2009} and whose finite-size corrections are fitted to the exact spectra, which extends the evaluation to $N\simeq10^{28}$. Appendix~\ref{app:numerics} describes the details of calculations in each of the system size regimes.

\begin{figure}[t]
\centering
\includegraphics[width=\columnwidth]{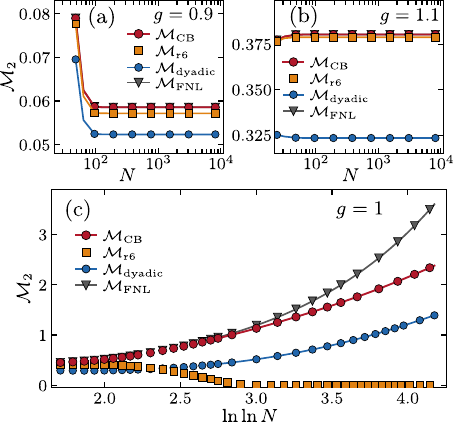}
\caption{Nonlocal SRE of the transverse-field Ising chain.
Panels (a) and (b): gapped phases at $g=0.9$ and $g=1.1$, from exact half-chain spectra of open chains up to $N=8192$.
Shown are the fermionic value $\mathcal M_{\mathrm{FNL}}$ of Eq.~\eqref{eq:mfnl}, the CB value $\mathcal M_{\mathrm{CB}}$, and the lower bounds $\mathcal M_{\mathrm{r6}}$ and $\mathcal M_{\mathrm{dyadic}}$ defined in the text.
Every quantity reaches an $N$-independent plateau once $N$ exceeds a few correlation lengths.
Panel (c): the critical point versus $\ln\ln N$, combining exact chains up to $N=3\times10^{6}$ with the ladder of entanglement energies of Eq.~\eqref{eq:ladder} up to $N\simeq10^{28}$.
The error of $\mathcal M_{\mathrm{CB}}$ due to the truncation of the spectrum is below $10^{-3}$, smaller than the symbols.}
\label{fig:ising}
\end{figure}

\textit{Quantities of interest.---}
We apply the bounds of Sec.~\ref{sec:witness} to determine the behavior of the nonlocal SRE across the phase diagram of the transverse-field Ising chain.

To find a lower bound on $\mathcal M_{\mathrm{NL}}$, we use Theorem~\ref{thm:witness} with the two families of comparators for which the CB value is proven to be optimal: the dyadic bound $\mathcal M_{\mathrm{dyadic}}$ is  Eq.~\eqref{eq:dyadicfloor} optimized over the $n+1$ shell values of the dyadic comparator, and the rank-six bound $\mathcal M_{\mathrm{r6}}$ uses the comparator of Schmidt rank six built from the six leading Schmidt values.
We use the CB value $\mathcal M_{\mathrm{CB}}$ of Eq.~\eqref{eq:cbgapM}, which we evaluate from the Schmidt spectrum with Eq.~\eqref{eq:walsh} to bound $\mathcal{M}_\mathrm{NL}$ from above.
Both lower bounds hold without any assumption, and together with $\mathcal M_{\mathrm{CB}}$ they enclose $\mathcal M_{\mathrm{NL}}$ in an interval, independently of the validity of  Conjecture~\ref{conj:cb}.

As an additional comparison, we consider the fermionic nonlocal SRE $\mathcal M_{\mathrm{FNL}}$, defined by restricting the local unitaries in Eq.~\eqref{eq:MNLsre} to fermionic Gaussian  unitaries~\cite{Iannotti2026Fermionic, Collura2026FreeFermion}.
For R\'enyi index $2$, the minimum is attained in the Gaussian Schmidt basis~\cite{Iannotti2026Fermionic} and equals
\begin{equation}
     \mathcal M_{\mathrm{FNL}}
     =-\log_{2}\prod_{j}\gamma(p_{j}),
     \label{eq:mfnl}
\end{equation}
where each entanglement mode has probabilities $p_j$ and $1-p_j$, and $\gamma$ is the stabilizer purity of a single dimer, Eq.~\eqref{eq:gamma2}.
In this basis, the reduced density matrix factorizes over the entanglement modes, so the many-body Schmidt probabilities are $\lambda_{\boldsymbol n}=\prod_j p_j^{1-n_j}(1-p_j)^{n_j}$, indexed by the binary mode occupations $\boldsymbol n$.
These probabilities need not decrease as the binary label increases.
The CB representative assigns the same Schmidt coefficients $\sqrt{\lambda_{\boldsymbol n}}$ in decreasing order to matching computational-basis labels. The two representatives thus have the same Schmidt spectrum but differ in its assignment to basis labels.
Since Gaussian local unitaries are included in the unrestricted minimization, $\mathcal M_{\mathrm{FNL}}$ is an upper bound on $\mathcal M_{\mathrm{NL}}$.
The fermionic value is additive over the modes and grows as $\Theta(\ln N)$ for the critical chain~\cite{Iannotti2026Fermionic}.

Figure~\ref{fig:ising} shows two upper bounds and two lower bounds on $\mathcal M_{\mathrm{NL}}$ in the two phases of the transverse-field Ising model and at the critical point $g=1$.

\textit{Gapped phases.---}
Away from the critical point the ground state obeys an area law, $S_{1}=O(1)$, so Theorem~\ref{thm:ceiling} bounds the nonlocal SRE by a constant at every $g\ne1$.
The Schmidt spectrum decays quickly there, and Eq.~\eqref{eq:arealaw} determines $\mathcal M_{\mathrm{NL}}$ within $O(\sqrt{\epsilon_{6}})$ of the CB value, with $\epsilon_{6}$ the weight beyond the six leading Schmidt values.
The rank-six bound $\mathcal M_{\mathrm{r6}}$ therefore lies close to the CB value: it differs from $\mathcal M_{\mathrm{CB}}$ by less than $10^{-2}$ at $|g-1|=0.05$ and by less than $2\times10^{-3}$ at $|g-1|=0.1$, with the difference decreasing rapidly deeper in either phase [right panel of Fig.~\ref{fig:overview} and Figs.~\ref{fig:ising}(a) and (b)].
Together, $\mathcal M_{\mathrm{r6}}$ and $\mathcal M_{\mathrm{CB}}$ enclose the nonlocal SRE of the Ising ground state in a narrow interval at every $g\ne1$. In passing, we note that the dyadic bound $\mathcal{M}_{\mathrm{dyadic}}$ is less tight away from the critical point, since a dyadic comparator cannot resolve the individual leading Schmidt values.

\textit{Critical point.---}
At $g=1$ the Schmidt spectrum of the ground state is spread over many dyadic shells.
Our analytical input is the Calabrese--Lefevre distribution of the entanglement spectrum~\cite{Calabrese2008spectrum}, which assumes the conformal form  $\mathrm{Tr}\rho_{A}^{q}=e^{-b(q-1/q)}$ of the moments of the reduced density matrix, with a single parameter $b=-\ln\lambda_{0}=\frac1{24}\ln N+O(1)$ for the half of an open chain with central charge $c=\tfrac12$~\cite{CalabreseCardy2004,Its2005XY,FagottiCalabrese2011}, and which has been tested against the entanglement spectra of critical chains~\cite{Pollmann2010spectra,Pollmann2009finite,Alba2018Cardy}.
The number of Schmidt probabilities larger than $\lambda$ is then $I_{0}\bigl(2\sqrt{b\ln(\lambda_{0}/\lambda)}\bigr)$, with $I_{0}$ the modified Bessel function, so for large $N$ the sorted probabilities are $\lambda_{r}\simeq\exp[-b-(\ln r)^{2}/4b]$. The mass of the shell $D_{k}$ is a Gaussian in $k$ centered at $2b/\ln2$ with a width of $\Theta(\sqrt b)$ shells: $\Theta(\sqrt{\ln N})$ shells carry masses $\hat w_{k}=\Theta((\ln N)^{-1/2})$ each, and the participation is $Q_{N}\equiv\sum_{k}\hat w_{k}^{4}=\Theta((\ln N)^{-3/2})$.

For this participation, Eqs.~\eqref{eq:shellfloor} and~\eqref{eq:shellsandwich} bound the stabilizer purity from both sides,
\begin{equation}
   \frac1{16}\,Q_{N}
   \ \le\ \mathcal F_{\mathrm{CB}}(N)
   \ \le\ \mathcal F_{*}(N)
   \ \le\ C_{\mathrm{sh}}\,Q_{N},
   \label{eq:isingbracket}
\end{equation}
and taking the negative base-two logarithm gives
  \begin{align}
   \mathcal M_{\mathrm{NL}}(N)
   &=\frac32\log_{2}\ln N+O(1),\notag\\
   \mathcal M_{\mathrm{CB}}(N)
   &=\frac32\log_{2}\ln N+O(1),
   \label{eq:isingcritical}
  \end{align}
so that the CB value and the minimum over local unitaries agree up to $O(1)$ without assuming Conjecture~\ref{conj:cb}.
The only input specific to the critical chain is the scaling $Q_{N}=\Theta((\ln N)^{-3/2})$ of the Calabrese--Lefevre profile. The coefficient $\frac32$ comes from the fourth power in $Q_{N}$ and from the width $\Theta(\sqrt{\ln N})$ of the profile.
The same profile led Ref.~\cite{Torre2026Spectrum} to propose a scaling scenario with a double-logarithmic growth of $\mathcal M_{\mathrm{CB}}$, conditional on an algebraic decay of the Pauli sum of Eq.~\eqref{eq:walsh} in $\ln N$.
Eq.~\eqref{eq:isingbracket} establishes that decay, fixes the coefficient, and extends the statement by providing both lower and upper bounds on $\mathcal M_{\mathrm{NL}}$.

In Fig.~\ref{fig:ising}(c), we present the size dependence of the bounds on the nonlocal SRE at the critical point, together with the fermionic value $\mathcal M_{\mathrm{FNL}}$.
We observe that the CB value grows without saturation up to the largest size $N\simeq10^{28}$, and that its local slope $d\mathcal M_{\mathrm{CB}}/d\log_{2}\ln N$ increases monotonically towards the asymptotic value $\frac32$ of Eq.~\eqref{eq:isingcritical} while remaining below it even there.
The slow onset of the asymptotic regime reflects the width of the profile, whose standard deviation is only about $0.42\sqrt{\ln N}$ shells, that is, three shells at $N\simeq10^{28}$: the asymptotic law assumes many shells of comparable mass, whereas at the sizes shown a few shells carry most of the weight and the largest shell mass may still be the largest Schmidt probability.

The two lower bounds are complementary.
The rank-six bound $\mathcal{M}_{\mathrm{r6} }$ resolves the leading Schmidt values individually and is the sharper one at moderate sizes, but its discarded weight $\epsilon_{6}$ grows with $N$, so the dyadic bound overtakes it near $N\approx3\times10^{4}$ and the rank-six bound decays to zero near $N\approx10^{8}$.
The dyadic bound $\mathcal M_{\mathrm{dyadic}}$ grows steadily, so the interval $\mathcal M_{\mathrm{dyadic}}\le\mathcal M_{\mathrm{NL}}\le\mathcal M_{\mathrm{CB}}$ remains of the order of one throughout, narrower than the window of $4+\log_{2}C_{\mathrm{sh}}$ that Eq.~\eqref{eq:shellparticipation} gives at any fixed size.

Finally, we observe that $\mathcal M_{\mathrm{FNL}}$ and $\mathcal M_{\mathrm{CB}}$ nearly coincide for smaller chains, while for $N\gtrsim10^{5}$ reassigning the Schmidt coefficients from the mode-occupation order to decreasing order produces an appreciably lower CB value.
At large $N$, the separation between  $\mathcal M_{\mathrm{FNL}}$ and  $\mathcal M_{\mathrm{NL}} \leq \mathcal{M}_{\mathrm{CB}}$  grows without bound, since $\mathcal M_{\mathrm{FNL}}$ keeps its $\Theta(\ln N)$ growth while $\mathcal M_{\mathrm{CB}}$ follows the $\Theta(\log\log N)$ scaling. This demonstrates that general local unitaries reduce the nonlocal SRE of the Gaussian minimization from logarithmic to doubly logarithmic in $N$.

\section{Conclusions}
\label{sec:conclusions}

\subsection{Summary}
\label{ssec:conclusion-summary}

In this work, we have developed an analytic framework for the nonlocal SRE of bipartite pure states, the stabilizer R\'enyi entropy minimized over all local unitaries.

\textit{Exact results.---} We have shown that the minimization over the two local unitaries reduces exactly to a minimization over a single unitary, which is the basis of all our results.
We have proved the CB optimality conjecture, which asserts that the sorted computational-basis representative state minimizes the SRE, for two families of states at every system size and every bipartition: the dyadic-staircase spectra, whose Schmidt values are constant on consecutive blocks of labels with sizes given by powers of two, and the spectra of Schmidt rank at most six.
For these families of states the minimization over local unitaries is solved exactly, and the nonlocal SRE is an explicit function of the Schmidt spectrum.

\textit{Bounds for arbitrary states.---}
An arbitrary Schmidt spectrum lies entrywise between two generally unnormalized dyadic staircases, obtained by replacing all entries in each dyadic shell with its smallest or largest Schmidt value.
The CB optimality for these staircase spectra, together with homogeneity and entrywise monotonicity of the optimized stabilizer purity, yields upper and lower bounds on the target's nonlocal SRE.
We have shown in this way that the nonlocal SRE of every pure state is fixed, up to bounded constants, by the largest Schmidt probability of each block, and that the CB value exceeds the nonlocal SRE by at most four, demonstrating that the CB optimality conjecture holds up to an additive constant.
  
Comparing a state with a nearby state of known nonlocal SRE gives sharper lower bounds for an individual spectrum. The dyadic shell construction also gives an upper bound logarithmic in the entanglement entropy, improving the asymptotic dependence of earlier linear bounds. In one dimension, this implies $O(1)$ nonlocal SRE for area-law states, at most $O(\log\log N)$ for critical states with logarithmic entanglement, and at most $O(\log N)$ for volume-law states.
The entanglement entropy thus constrains the nonlocal SRE from above, while the distribution of Schmidt weight across dyadic shells, or logarithmic rank scales, determines its value up to a bounded additive error.

\textit{Many-body systems.---}
We have applied these results to products of dimers and to the transverse-field Ising chain. When $m=\Theta(\ln N)$ dimers cross the cut, singlet products and products of imperfect dimers have the same logarithmic scaling of the entanglement entropy. Their nonlocal SRE behaves differently: it vanishes for singlets and grows as $\frac32\log_2\ln N+O(1)$ when the dimers are imperfect.

For the Ising chain, our numerical upper and lower bounds approach finite plateaus in the gapped phases and enclose the nonlocal SRE in narrow intervals.
At criticality, applying our universal bounds to the Calabrese--Lefevre form of the entanglement spectrum gives $\mathcal M_{\mathrm{NL}}=\frac32\log_2\ln N+O(1)$, a scaling supported by our numerical results.
Enlarging the minimization from local fermionic Gaussian unitaries to all local unitaries therefore changes the critical growth from logarithmic to doubly logarithmic in $N$~\cite{Iannotti2026Fermionic, Collura2026FreeFermion}.
These applications demonstrate a practical route to computing bounds on nonlocal SRE from accessible many-body Schmidt spectra.

\subsection{Outlook}
\label{ssec:conclusion-outlook}

\textit{CB optimality.---}
A central open problem is to settle Conjecture~\ref{conj:cb} for arbitrary Schmidt spectra. A related question is whether CB optimality extends to higher R\'enyi indices.
One obstacle to extending the argument of Sec.~\ref{sec:dyadic} is that its elementary rank bounds need not be saturated by the CB projectors at non-dyadic ranks.
The Schur--Weyl analysis of Sec.~\ref{sec:rank6} offers another route through sufficient conditions involving the antisymmetric sector.
Extending this approach requires new bounds: the present compression inequality holds up to rank seven but fails at rank eight, while the overlap inequality is established only up to rank six.

A proof would make the CB value an exact formula for the nonlocal SRE of every bipartite pure state.
A counterexample would exhibit a representative of the same local-unitary orbit with lower SRE than the sorted CB state.
At R\'enyi index $2$, Theorem~\ref{thm:resValue} already guarantees that the CB value and the nonlocal SRE share the leading term of any divergent scaling with system size, regardless of the conjecture's validity.

\textit{Universality at criticality and beyond.---}
Our derivation of the $\frac32\log_{2}\ln N$ law takes the Calabrese--Lefevre distribution as input.
This distribution is expected to hold for every one-dimensional conformal critical point, with the central charge entering only the additive constant~\cite{Calabrese2008spectrum,Alba2018Cardy}, so the coefficient $\frac32$ should be universal.
Testing this universality in interacting critical chains~\cite{Hoshino2026CFT,Catalano2026SPT,Korbany2025LongRange} is a natural next step.
Topologically ordered states in two dimensions~\cite{LevinWen2005,LuoWu2016} are another natural target, since their entanglement spectra have a structure similar to the dimers of Theorem~\ref{thm:dimer}.
More broadly, our bounds apply to any state whose Schmidt spectrum is known, exactly or numerically, and thus allow the nonlocal SRE to be tracked across phase diagrams, quenches, and eigenstates of many-body systems on the same footing as the entanglement entropy, for example in the thermalization and localization of isolated systems~\cite{Abanin2019, Sierant2025MBL} or in monitored dynamics~\cite{Potter2022mipt}.

\textit{Operational meaning.---}
The nonlocal SRE is equivalent, up to dimension-independent multiplicative constants, to the nonlocal min-relative entropy of nonstabilizerness (NMRE) $D_{\mathrm{min}}^{\mathrm{NL}}$, defined through the stabilizer fidelity~\cite{Sierant26exact}. The NMRE diverges along precisely those families of pure states that universally embezzle entanglement under local operations and classical communication~\cite{Sierant26exact, vanDam2003embezzling, vanLuijk2025embezzlers}.  The same criterion therefore holds for the nonlocal SRE.

A complementary interpretation relates $D_{\mathrm{min}}^{\mathrm{NL}}$ to the magic of purification, defined as the minimum $D_{\mathrm{min}}$ over all purifications of a density matrix. Specifically, $D_{\mathrm{min}}^{\mathrm{NL}}$ equals the minimum magic of purification along the unitary orbit of the reduced density matrix~\cite{Viscardi26nonlocal}. Identifying an operational task whose optimal performance is determined directly by the nonlocal SRE remains an open problem.

\begin{acknowledgments}
P.S. acknowledges collaboration with Sreemayee Aditya and Xhek Turkeshi on related topics, and useful discussions with members of Quantic group at BSC.
P.S. acknowledges fellowship within the “Generación D” initiative, Red.es, Ministerio para la Transformación Digital y de la Función Pública, for talent attraction (C005/24-ED CV1), funded by the European Union NextGenerationEU funds through PRTR.
The author used ChatGPT~5.5 and~5.6 (OpenAI) as an aid in formalizing the theorems and their proofs, and Claude Fable~5 and~5.1 (Anthropic) for assistance in manuscript preparation; all results were verified independently by the author, who takes full responsibility for the content of this manuscript.
\end{acknowledgments}
\textit{Note added.---}
While this manuscript was being finalized, Ref.~\cite{LiuCui2026Spectral} appeared, which announces related bounds on the computational-basis value of the nonlocal SRE in terms of shell-flat spectra.

\appendix

\section{Proof of the dyadic rank inequality}
  \label{app:rank}

In this appendix, we prove Theorem~\ref{thm:rank} in full.
Throughout, $d=2^{n}$ and the computational-basis labels are the elements of $\Ftwo^{n}$, that is, bit strings $x=(x_{1},\dots,x_{n})$ of length $n$, which we read as the binary integers $x=\sum_{t=1}^{n}x_{t}2^{t-1}$ from $0$ to $d-1$.
Addition in $\Ftwo^{n}$ is bitwise addition modulo two, that is, the XOR of the two bit strings, and this is the operation written $x{+}u$ in Eq.~\eqref{eq:Phibasis}.
A \emph{linear subspace} of $\Ftwo^{n}$ is a set of bit strings closed under this addition, and it automatically contains the all-zeros string $0=x{+}x$.
We write $L_{r}\equiv\{0,\dots,r-1\}$ for the support of the truncation projector $E_{r}$.
For dyadic $r=2^{\alpha}$ the set $L_{2^{\alpha}}$ consists of the bitstrings whose bits $x_{t}$ with $t>\alpha$ vanish, and it is therefore a linear subspace of $\Ftwo^{n}$, whereas for other values of $r$ the set $L_{r}$ is not closed under addition.
The sets $L_{1}\subset L_{2}\subset L_{4}\subset\cdots\subset L_{d}$ form a nested chain of such subspaces, and this nested linear structure is the only property of the label order used below.
As in Sec.~\ref{ssec:rank}, the two rank vectors are dyadic, $r_{j}=2^{\alpha_{j}}$ and $s_{j}=2^{\beta_{j}}$, and we write $\alpha=(\alpha_{1},\dots,\alpha_{4})\in\{0,\dots,n\}^{4}$ and $\beta=(\beta_{1},\dots,\beta_{4})$ for the two quadruples of exponents.
All linear algebra in this appendix takes place in $\Ftwo^{n}$ and its powers, never in the Hilbert space $\mathbb C^{d}$.

\subsection{The stabilizer subspace and the CB value}
  \label{app:lexvalue}

We start by recalling the stabilizer basis $\ket{\Phi_{u,v}}$ of Eq.~\eqref{eq:Phibasis}, which we refer to as the $\Phi$-basis in what follows.
The following lemma is standard~\cite{Zhu16fails, Gross21comm, Bittel26complete}, and we include its short proof to fix the normalization and the explicit basis used throughout.
\begin{lemma}[Stabilizer subspace]
  \label{lem:Phi}
The Pauli average $d^{-2}\sum_{P\in\mathcal P_{n}}P^{\otimes4}$ is the orthogonal projector onto $\mathrm{span}\{\ket{\Phi_{u,v}}\}$, that is, Eq.~\eqref{eq:PWdef} holds.
\end{lemma}

\begin{proof}
Write $P=X^{a}Z^{b}$ up to a phase, with $a,b\in\Ftwo^{n}$.
Acting on a basis term of Eq.~\eqref{eq:Phibasis}, the four $Z$-phases multiply to $(-1)^{b\cdot[x+(x+u)+(x+v)+(x+u+v)]}=1$, while $X^{a}$ shifts $x\mapsto x+a$ and permutes the terms of the sum.  Hence $P^{\otimes4}\ket{\Phi_{u,v}}=\ket{\Phi_{u,v}}$ for every $P$.
The set $\{P^{\otimes4}\}$ is an ordinary abelian group, since the fourth power removes the Pauli phases: for $P_{a}P_{b}=\omega(a,b)P_{a+b}$ with $\omega(a,b)\in\{\pm1,\pm i\}$ one has  $P_{a}^{\otimes4}P_{b}^{\otimes4}=\omega(a,b)^{4}P_{a+b}^{\otimes4}=P_{a+b}^{\otimes4}$.
The average $d^{-2}\sum_{P}P^{\otimes4}$ is therefore the orthogonal projector onto the common invariant subspace of this group, and each $\ket{\Phi_{u,v}}$ lies in it.
Its trace is $d^{-2}\sum_{P}(\Tr P)^{4}=d^{-2}\,d^{4}=d^{2}$, since only $P=\mathbb 1$ contributes.
The $d^{2}$ orthonormal vectors $\ket{\Phi_{u,v}}$ therefore exhaust it.
\end{proof}

  We now return to the rank inequality.
  For operators $M_{1},\dots,M_{4}$ on $\mathbb{C}^{d}$, a direct computation gives the matrix elements
  \begin{align}
   &\bra{\Phi_{u,v}}M_{1}\otimes M_{2}\otimes M_{3}\otimes M_{4}\ket{\Phi_{u',v'}}
   \notag\\
   &=\frac1d\sum_{x,x'}
   (M_{1})_{x,x'}(M_{2})_{x+u,x'+u'}
   (M_{3})_{x+v,x'+v'}\notag\\
   &\hspace{4.2em}\times(M_{4})_{x+u+v,x'+u'+v'} .
   \label{eq:Phimatel}
  \end{align}
Equation~\eqref{eq:Phimatel} gives the matrix, in the $\Phi$-basis, of the operator $\PW(M_{1}\otimes M_{2}\otimes M_{3}\otimes M_{4})\PW$, which we call the compression of $M_{1}\otimes\cdots\otimes M_{4}$ to the stabilizer subspace $W$ and denote, as in Sec.~\ref{ssec:rank}, by a subscript $W$.
If all four operators are diagonal in the computational basis, each term of the sum forces $x=x'$, $u=u'$, and $v=v'$, so the compressed operator is diagonal in the $\Phi$-basis.

In particular, for the computational-basis choice $A^{\mathrm{CB}}=\bigotimes_{j}E_{2^{\alpha_{j}}}$, writing $q_{1}=0$, $q_{2}=u$, $q_{3}=v$, $q_{4}=u{+}v$, the diagonal entries are
  \begin{equation}
   \bigl(A^{\mathrm{CB}}_{W}\bigr)_{(u,v),(u,v)}=\frac{a_{\alpha}(u,v)}{d},
   \qquad
   a_{\alpha}(u,v)=\bigl|\mathcal L_{\alpha}(u,v)\bigr|,
   \label{eq:adiag}
  \end{equation}
where $\mathcal L_{\alpha}(u,v)\equiv\{x:\,x+q_{j}\in L_{2^{\alpha_{j}}}\ \text{for all}\ j\}$ is the set of labels $x$ for which the $j$th entry of the quadruple $(x,x{+}u,x{+}v,x{+}u{+}v)$ lies in $L_{2^{\alpha_{j}}}$ for every $j$.

\subsection{The twenty bounds}
\label{app:cuts}
Throughout this subsection, $A=P_{1}\otimes P_{2}\otimes P_{3}\otimes P_{4}$ and $B=Q_{1}\otimes Q_{2}\otimes Q_{3}\otimes Q_{4}$ are built from arbitrary orthogonal projectors $P_{j}$ and $Q_{j}$ of arbitrary ranks $r_{j}$ and $s_{j}$, which need not be powers of two.
The twenty bounds derived here hold at every rank, and it is only in Appendix~\ref{app:mincut} that we restrict to $r_{j}=2^{\alpha_{j}}$ and $s_{j}=2^{\beta_{j}}$ and show that the tightest of them is attained.
We refer to the $j$th tensor factor of $(\mathbb C^{d})^{\otimes4}$ as the $j$th replica, so that $P_{j}$ and $Q_{j}$ are the two projectors acting on replica $j$.
The starting point is that the trace $\Tr(A_{W}B_{W})=\Tr[\PW A\PW B]$ is monotone in each of its two arguments, since it is the Hilbert--Schmidt inner product of $A$ with the positive-semidefinite operator $\PW B\PW$.
Replacing any $P_{j}$ or $Q_{j}$ by $\mathbb 1$ can therefore only increase it, and after a suitable replacement the trace can be evaluated in terms of the ranks alone.
Two kinds of replacement achieve this.
Removing both projectors of one replica gives the twelve \emph{drop-one} bounds of Eq.~\eqref{eq:dropone}
Keeping a single projector on each replica, three from $A$ and one from $B$ or the other way around, gives the eight \emph{cross} bounds of Eq.~\eqref{eq:crosscut}. We derive the two families in turn.

\begin{lemma}[Drop-one bounds]
  \label{lem:dropone}
  For every pair of distinct indices $i,a\in\{1,2,3,4\}$,
  \begin{equation}
   \Tr(A_{W}B_{W})\ \le\
   \frac{r_{a}s_{a}\prod_{j\ne i,a}\min(r_{j},s_{j})}{d^{2}} .
   \label{eq:dropone}
  \end{equation}
\end{lemma}

\begin{proof}
We replace both $P_{i}$ and $Q_{i}$ by $\mathbb 1$, which can only increase the trace, and evaluate the result. Expanding both copies of $\PW=d^{-2}\sum_{R\in\mathcal P_{n}}R^{\otimes4}$ over the Pauli strings, cf. Eq.~\eqref{eq:PWdef}, and collecting the four replicas separately gives
\begin{equation*}
   \Tr[\PW A\PW B]
   =\frac1{d^{4}}\sum_{R,S\in\mathcal P_{n}}\prod_{j=1}^{4}\Tr\bigl(RP_{j}SQ_{j}\bigr),
\end{equation*}
where from now on $P_{i}=Q_{i}=\mathbb 1$.
The factor of replica $i$ is $\Tr(RS)=d\,\delta_{R,S}$, because distinct Pauli operators are orthogonal, and it collapses the double sum to a single one,
\begin{equation*}
   \Tr[\PW A\PW B]
   =\frac1{d^{3}}\sum_{R\in\mathcal P_{n}}\prod_{j\ne i}t_{j}(R),
\end{equation*}
where $t_{j}(R)\equiv\Tr\bigl(P_{j}RQ_{j}R\bigr)$.
Each $t_{j}(R)$ is the trace of the product of two orthogonal projectors, $P_{j}$ and $RQ_{j}R$, of ranks $r_{j}$ and $s_{j}$, so that $0\le t_{j}(R)\le\min(r_{j},s_{j})$.
We use this bound on two of the three remaining replicas and sum the third one exactly.
For replica $a$, the Pauli average $\sum_{R}RQ_{a}R=d\,\Tr(Q_{a})\,\mathbb 1$ gives $\sum_{R}t_{a}(R)=d\,r_{a}s_{a}$, and Eq.~\eqref{eq:dropone} follows.
\end{proof}

  The cross bounds rest on the fact that $\PW$ carries no information about any single replica: each $\ket{\Phi_{u,v}}$ is maximally entangled between any one replica and the remaining
  three.
  We write $Q^{(i)}\equiv\mathbb 1^{\otimes(i-1)}\otimes Q\otimes\mathbb 1^{\otimes(4-i)}$ for an operator $Q$ placed on replica $i$.

  \begin{lemma}[Single-replica properties of $\PW$]
  \label{lem:onebody}
  For every operator $Q$ on $\mathbb C^{d}$ and every replica $i$,
  \begin{equation}
   \text{(i)}\ \ \PW\,Q^{(i)}\,\PW=\frac{\Tr Q}{d}\,\PW,
   \qquad
   \text{(ii)}\ \ \Tr_{i}\,\PW=\frac{\mathbb 1^{\otimes3}}{d}.
   \label{eq:onebody}
  \end{equation}
  \end{lemma}

  \begin{proof}
  We take $i=1$, the other replicas being analogous.
  For (i), Eq.~\eqref{eq:Phimatel} with $M_{1}=Q$ and $M_{2}=M_{3}=M_{4}=\mathbb 1$ contains three Kronecker deltas, $x+u=x'+u'$, $x+v=x'+v'$, and $x+u+v=x'+u'+v'$.
  The first two give $x+x'=u+u'=v+v'$, and inserting this into the third gives $x+x'=0$, hence $x=x'$, $u=u'$, and $v=v'$.
  The matrix of $\PW Q^{(1)}\PW$ in the $\Phi$-basis is therefore diagonal, with the constant entry $d^{-1}\sum_{x}Q_{x,x}=\Tr Q/d$.
  For (ii), tracing out the first replica of $\ketbra{\Phi_{u,v}}{\Phi_{u,v}}$ leaves $d^{-1}\sum_{x}\ketbra{x{+}u,x{+}v,x{+}u{+}v}{x{+}u,x{+}v,x{+}u{+}v}$.
  Summing over $(u,v)$ and using that $(x,u,v)\mapsto(x{+}u,x{+}v,x{+}u{+}v)$ is a bijection of $\Ftwo^{3n}$ gives $d^{-1}\sum_{a,b,c}\ketbra{a,b,c}{a,b,c}$.
  \end{proof}

  \begin{lemma}[Cross bounds]
  \label{lem:crosscut}
  For every $i\in\{1,2,3,4\}$,
  \begin{equation}
   \Tr(A_{W}B_{W})\le\frac{s_{i}\prod_{j\ne i}r_{j}}{d^{2}},
   \qquad
   \Tr(A_{W}B_{W})\le\frac{r_{i}\prod_{j\ne i}s_{j}}{d^{2}} .
   \label{eq:crosscut}
  \end{equation}
  \end{lemma}

  \begin{proof}
  For the first bound we keep $Q_{i}$ as the only projector of $B$ and the three $P_{j}$ with $j\ne i$ as the only projectors of $A$, replacing the other four by $\mathbb 1$, which can
  only increase the trace.
  After the replacement $B=Q_{i}^{(i)}$, so that $\PW B\PW=\frac{s_{i}}{d}\PW$ by property (i), and the trace becomes $\frac{s_{i}}{d}\Tr[A\PW]$.
  Since $A$ is now the identity on replica $i$, property (ii) gives $\Tr[A\PW]=\Tr[(\bigotimes_{j\ne i}P_{j})\Tr_{i}\PW]=d^{-1}\prod_{j\ne i}r_{j}$, which is the first bound.
  The second bound follows by exchanging the roles of $A$ and $B$.
  \end{proof}

 \subsection{The minimum of the twenty bounds at dyadic ranks}
  \label{app:mincut}

The twenty bounds of Lemmas~\ref{lem:dropone} and~\ref{lem:crosscut} hold at every rank.
We now restrict to dyadic ranks $r_{j}=2^{\alpha_{j}}$ and $s_{j}=2^{\beta_{j}}$ and show that the tightest of the twenty bounds is exactly the computational-basis value, which will enable us to complete the proof of Theorem~\ref{thm:rank}.
The argument has three steps.
First, we show that the computational-basis value counts the pairs of quadruples compatible with the two rank vectors, and that at dyadic ranks these pairs form a linear subspace $\mathcal S_{\alpha\beta}$ of $\Ftwo^{4n}$.
Second, we express $\dim\mathcal S_{\alpha\beta}$ as the minimum weight of a basis chosen among eight vectors of $\Ftwo^{4}$ carrying the weights $\alpha_{j}$ and $\beta_{j}$.
Third, we enumerate the bases and find that their minimum weight is the minimum of the twenty exponents.

  \textit{The computational-basis value as a count.---}
  With $(q_{1},q_{2},q_{3},q_{4})=(0,u,v,u{+}v)$ as in Eq.~\eqref{eq:adiag}, let
  \begin{equation}
   \mathcal S_{\alpha\beta}\equiv\bigl\{(x,u,v,y):\,
   x+q_{j}\in L_{2^{\alpha_{j}}},\ y+q_{j}\in L_{2^{\beta_{j}}}\ \text{for all}\ j\bigr\}
   \label{eq:Sdef}
  \end{equation}
  be the set of pairs of quadruples $(x,x{+}u,x{+}v,x{+}u{+}v)$ and $(y,y{+}u,y{+}v,y{+}u{+}v)$ with the same differences $u$ and $v$, the first compatible with the ranks $2^{\alpha_{j}}$
  and the second with the ranks $2^{\beta_{j}}$.
  At fixed $(u,v)$ the admissible $x$ form the set $\mathcal L_{\alpha}(u,v)$ and the admissible $y$ the set $\mathcal L_{\beta}(u,v)$, so that $|\mathcal S_{\alpha\beta}|=\sum_{u,v}a_{\alpha}(u,v)\,a_{\beta}(u,v)$.
  Since both reference operators are diagonal in the $\Phi$-basis, Eq.~\eqref{eq:adiag} gives
  \begin{equation}
   \Tr\bigl(A_{W}^{\mathrm{CB}}B_{W}^{\mathrm{CB}}\bigr)
   =\sum_{u,v}\frac{a_{\alpha}(u,v)}{d}\,\frac{a_{\beta}(u,v)}{d}
   =\frac{|\mathcal S_{\alpha\beta}|}{d^{2}} .
   \label{eq:Scount}
  \end{equation}
  This identity holds at every rank.
  At dyadic ranks each $L_{2^{\alpha}}$ is a linear subspace of $\Ftwo^{n}$, so the eight conditions in Eq.~\eqref{eq:Sdef} are linear in $(x,u,v,y)$, and $\mathcal S_{\alpha\beta}$ is a linear
  subspace of $\Ftwo^{4n}$ with $|\mathcal S_{\alpha\beta}|=2^{\dim\mathcal S_{\alpha\beta}}$.
  At the same ranks the twenty bounds of Eqs.~\eqref{eq:dropone} and~\eqref{eq:crosscut} read $2^{c}/d^{2}$ with the twenty exponents
  \begin{align}
   c_{i,a}&=\alpha_{a}+\beta_{a}+\sum_{j\ne i,a}\min(\alpha_{j},\beta_{j}),\notag\\
   c_{i}&=\beta_{i}+\sum_{j\ne i}\alpha_{j},
   \qquad
   c_{i}'=\alpha_{i}+\sum_{j\ne i}\beta_{j},
   \label{eq:exponents}
  \end{align}
  where $i\ne a$.
  Theorem~\ref{thm:rank} is therefore a consequence of the following statement.

  \begin{theorem}[Minimum of the twenty bounds]
  \label{thm:mincut}
  For all $\alpha,\beta\in\{0,\dots,n\}^{4}$, $\dim\mathcal S_{\alpha\beta}$ equals the minimum of the twenty exponents of Eq.~\eqref{eq:exponents}.
  \end{theorem}

  \textit{The dimension of $\mathcal S_{\alpha\beta}$ as a minimum-weight basis.---}
  The conditions in Eq.~\eqref{eq:Sdef} involve the eight linear forms
  \begin{equation*}
  \begin{aligned}
   A_{1}&=x, & A_{2}&=x{+}u, & A_{3}&=x{+}v, & A_{4}&=x{+}u{+}v,\\
   B_{1}&=y, & B_{2}&=y{+}u, & B_{3}&=y{+}v, & B_{4}&=y{+}u{+}v,
  \end{aligned}
  \end{equation*}
  and require $A_{j}\in L_{2^{\alpha_{j}}}$ and $B_{j}\in L_{2^{\beta_{j}}}$.
  We assign to each form the weight $w(A_{j})=\alpha_{j}$ and $w(B_{j})=\beta_{j}$, so that the condition on a form $e$ reads $e\in L_{2^{w(e)}}$, that is, bit $t$ of $e$ vanishes for
  every $t>w(e)$.
  Each form acts on the bit strings $x,u,v,y$ bit by bit, so we represent it by its coefficient vector in $\Ftwo^{4}$ with respect to $(x,u,v,y)$,
  \begin{equation*}
  \begin{aligned}
   a_{1}&=(1,0,0,0), & a_{2}&=(1,1,0,0), & a_{3}&=(1,0,1,0),\\
   a_{4}&=(1,1,1,0), & b_{1}&=(0,0,0,1), & b_{2}&=(0,1,0,1),\\
   b_{3}&=(0,0,1,1), & b_{4}&=(0,1,1,1),
  \end{aligned}
  \end{equation*}
  so that bit $t$ of $A_{2}=x{+}u$, for instance, is $a_{2}\cdot(x_{t},u_{t},v_{t},y_{t})=x_{t}+u_{t}$.
  These eight vectors span $\Ftwo^{4}$, and we call any four of them that are linearly independent a \emph{basis}.

  \begin{lemma}[Dimension as a minimum-weight basis]
  \label{lem:matroid}
  $\displaystyle \dim\mathcal S_{\alpha\beta}=\min_{\mathcal B\ \mathrm{basis}}\ \sum_{e\in\mathcal B}w(e)$.
  \end{lemma}

  \begin{proof}
  We first decouple the bit positions.
  The condition on the form $e$ constrains bit $t$ of $e$ only for $t>w(e)$, and bit $t$ of $e$ depends only on the four bits $(x_{t},u_{t},v_{t},y_{t})$.
  Hence $\mathcal S_{\alpha\beta}$ is the product over $t=1,\dots,n$ of the solution sets of the linear systems $\{e\cdot(x_{t},u_{t},v_{t},y_{t})=0:\,w(e)<t\}$ on $\Ftwo^{4}$, whose dimensions are
  $4-\mathrm{rk}\{e:\,w(e)<t\}$, so that
  \begin{equation}
   \dim\mathcal S_{\alpha\beta}=\sum_{t=1}^{n}\bigl(4-\mathrm{rk}\{e:\,w(e)<t\}\bigr).
   \label{eq:dimsum}
  \end{equation}
  We next compare this with the weight of a basis.
  Since $w(e)$ counts the levels $t\in\{1,\dots,n\}$ with $w(e)\ge t$, the weight of any basis $\mathcal B$ is
  \begin{equation}
  \begin{split}
   \sum_{e\in\mathcal B}w(e)&=\sum_{t=1}^{n}\bigl(4-|\mathcal B\cap\{e:\,w(e)<t\}|\bigr)\\
   &\ge\sum_{t=1}^{n}\bigl(4-\mathrm{rk}\{e:\,w(e)<t\}\bigr),
  \end{split}
  \label{eq:basisweight}
  \end{equation}
  because $\mathcal B\cap\{e:\,w(e)<t\}$ is a linearly independent subset of $\{e:\,w(e)<t\}$.
  The right-hand side equals $\dim\mathcal S_{\alpha\beta}$ by Eq.~\eqref{eq:dimsum}, so every basis has weight at least $\dim\mathcal S_{\alpha\beta}$.
  To find a basis of weight exactly $\dim\mathcal S_{\alpha\beta}$, we go through the eight vectors in order of nondecreasing weight and keep a vector whenever it is linearly independent of the vectors
  already kept.
  Once all vectors of weight below $t$ have been processed, the kept vectors of weight below $t$ form a maximal linearly independent subset of $\{e:\,w(e)<t\}$ and therefore have
  cardinality $\mathrm{rk}\{e:\,w(e)<t\}$.
  The kept set is thus a basis for which Eq.~\eqref{eq:basisweight} is an equality at every level $t$, and its weight equals $\dim\mathcal S_{\alpha\beta}$.
  \end{proof}

  \textit{The fifty-six bases.---}
  It remains to enumerate the bases and their weights.
  We write $p_{1}=(0,0)$, $p_{2}=(1,0)$, $p_{3}=(0,1)$, $p_{4}=(1,1)$ for the four points of $\Ftwo^{2}$, so that $a_{j}=(1,p_{j},0)$ and $b_{j}=(0,p_{j},1)$.

  \begin{lemma}[The fifty-six bases]
  \label{lem:bases}
  The bases among the eight vectors are exactly the eight \emph{cross bases} $\{B_{i}\}\cup\{A_{j}:\,j\ne i\}$ and $\{A_{i}\}\cup\{B_{j}:\,j\ne i\}$, and the forty-eight \emph{drop-one
  bases} $\{A_{a},B_{a},X_{j},X_{k}\}$ with $X_{j}\in\{A_{j},B_{j}\}$ and $X_{k}\in\{A_{k},B_{k}\}$, where $(i,a)$ runs over the twelve pairs of distinct indices and
  $\{j,k\}=\{1,2,3,4\}\setminus\{i,a\}$.
  \end{lemma}

  \begin{proof}
  Any three $A$-vectors $a_{p},a_{q},a_{r}$ are linearly independent, since the $3\times3$ block formed by their first three coordinates has the row $(1,1,1)$ above the columns
  $p_{p},p_{q},p_{r}$, and this block is nonsingular because three distinct points of $\Ftwo^{2}$ never lie on a line.
  Adjoining a $B$-vector adds a nonzero last coordinate, so every set of three $A$-vectors and one $B$-vector is a basis, and by symmetry so is every set of one $A$-vector and three
  $B$-vectors.
  These $16+16$ bases are cross bases when the index of the single vector lies outside the indices of the three, and drop-one bases with $X_{j}$ and $X_{k}$ on the same side otherwise.
  All four $A$-vectors are dependent, since $\sum_{j}a_{j}=0$, and likewise the four $B$-vectors.
  Finally, consider a set $\{A_{p},A_{q},B_{r},B_{s}\}$ with $p\ne q$ and $r\ne s$.
  Any three of its vectors are independent, being contained in one of the bases just found, so the set is dependent if and only if all four vectors sum to zero.
  Their sum is $(0,\,p_{p}+p_{q}+p_{r}+p_{s},\,0)$, which vanishes if and only if $p_{p}+p_{q}=p_{r}+p_{s}$.
  Each nonzero element of $\Ftwo^{2}$ is the sum of exactly two pairs of distinct points, and these two pairs are complementary, so the set is a basis if and only if $\{p,q\}$ and
  $\{r,s\}$ share exactly one index, say $a$, in which case it is the drop-one basis $\{A_{a},B_{a},A_{j},B_{k}\}$ with $\{j,k\}$ the two remaining indices.
  There are $24$ such sets, which completes the count $16+16+24=56$.
  \end{proof}

  \begin{proof}[Proof of Theorem~\ref{thm:mincut}]
  By Lemma~\ref{lem:matroid}, $\dim\mathcal S_{\alpha\beta}$ is the minimum weight of a basis, and by Lemma~\ref{lem:bases} every basis is a cross basis or a drop-one basis.
  The cross bases $\{B_{i}\}\cup\{A_{j}:\,j\ne i\}$ and $\{A_{i}\}\cup\{B_{j}:\,j\ne i\}$ have the weights $c_{i}$ and $c_{i}'$ of Eq.~\eqref{eq:exponents}.
  At fixed $(i,a)$ the four drop-one bases have the weights $\alpha_{a}+\beta_{a}+\theta_{j}+\theta_{k}$ with $\theta_{j}\in\{\alpha_{j},\beta_{j}\}$ and
  $\theta_{k}\in\{\alpha_{k},\beta_{k}\}$, the smallest of which is $c_{i,a}$.
  The minimum over all bases is therefore the minimum of the twenty exponents.
  \end{proof}

  \begin{proof}[Proof of Theorem~\ref{thm:rank}]
  By Lemmas~\ref{lem:dropone} and~\ref{lem:crosscut}, $\Tr(A_{W}B_{W})$ is at most the smallest of the twenty bounds, which at dyadic ranks is $2^{c}/d^{2}$ with $c$ the minimum of the
  twenty exponents.
  By Theorem~\ref{thm:mincut} this equals $2^{\dim\mathcal S_{\alpha\beta}}/d^{2}=|\mathcal S_{\alpha\beta}|/d^{2}$, which by Eq.~\eqref{eq:Scount} is $\Tr(A_{W}^{\mathrm{CB}}B_{W}^{\mathrm{CB}})$.
  \end{proof}
  
The proof uses dyadicity in exactly one place, namely the fact that $L_{2^{\alpha}}$ is the linear subspace of $\Ftwo^{n}$ cut out by the vanishing of the bits above $\alpha$.
This makes $\mathcal S_{\alpha\beta}$ a linear subspace, so that $|\mathcal S_{\alpha\beta}|$ is a power of two, and it decouples the bit positions in Lemma~\ref{lem:matroid}.
At non-dyadic ranks Eq.~\eqref{eq:Scount} still holds, but $|\mathcal S_{\alpha\beta}|$ need not be a power of two and the tightest of the twenty bounds is no longer attained.
At $d=4$, for the ranks $\boldsymbol r=(1,1,2,2)$ and $\boldsymbol s=(3,3,3,3)$, one finds $|\mathcal S_{\alpha\beta}|=5$, so that the computational-basis value is $5/16$, while the smallest of the twenty bounds is $6/16$.

\section{Schur--Weyl decomposition of the stabilizer purity}
\label{app:SW}

This appendix develops the Schur--Weyl analysis summarized in Sec.~\ref{sec:rank6}.
We first prove the three-sector decomposition of Eq.~\eqref{eq:split} and record the sector ranks and the sector values at a diagonal matrix (Appendix~\ref{app:sectors}).
We then bound the symmetric sector (Appendix~\ref{app:symSOS}) and the mixed sector (Appendix~\ref{app:mixed}) by antisymmetric data, and assemble the reduced form of Proposition~\ref{prop:master} (Appendix~\ref{app:reduced}).

\subsection{The three sectors}
\label{app:sectors}

The compressed representation~\eqref{eq:Gcompressed} makes an exact representation-theoretic decomposition available. The symmetric group $S_{4}$ acts on $(\mathbb{C}^{d})^{\otimes4}$ by permuting the four tensor factors: for $\pi\in S_{4}$, let $r(\pi)$ denote the unitary carrying $\ket{\psi_{1}}\otimes\cdots\otimes\ket{\psi_{4}}$ to $\ket{\psi_{\pi^{-1}(1)}}\otimes\cdots\otimes\ket{\psi_{\pi^{-1}(4)}}$. Under a finite-group action the space splits into orthogonal pieces, one for each irreducible representation; for $S_{4}$ these are labeled by the partitions $\lambda$ of $4$. The orthogonal projector onto the piece transforming as $\lambda$ is the character average
\begin{equation}
 \Pi_{\lambda}^{S_{4}}
 =\frac{\dim[\lambda]}{4!}\sum_{\pi\in S_{4}}\chi_{\lambda}(\pi)\,r(\pi),
 \label{eq:isotypic}
\end{equation}
with $\chi_{\lambda}$ the character of the irreducible representation $[\lambda]$. Being a combination of the $r(\pi)$, it commutes with every operator that commutes with all permutations---in particular with $\PW$ and with $N^{\otimes4}$.

Each sector is most usefully described by single-copy data. For a partition $\lambda$ let $\mathbb S_{\lambda}$ denote the associated Schur functor: $\mathbb S_{(4)}=\mathrm{Sym}^{4}$ and $\mathbb S_{(1^{4})}=\Lambda^{4}$, while $\mathbb S_{(2,2)}$ is the mixed functor associated with the partition $(2,2)$~\cite{FultonHarris}. For a matrix $N$, $\mathbb S_{\lambda}(N)$ is the operator that $N^{\otimes4}$ induces on a single copy of the module $\mathbb S_{\lambda}(\mathbb{C}^{d})\subset(\mathbb{C}^{d})^{\otimes4}$; for $\lambda=(4)$ and $(1^{4})$ this is the restriction of $N^{\otimes4}$ to the totally symmetric and totally antisymmetric subspaces, respectively. The Pauli average of the functor,
\begin{equation}
 \Pi_{\lambda}^{W}
 :=\frac1{d^{2}}\sum_{P\in\mathcal P_{n}}
 \mathbb S_{\lambda}(P),
 \label{eq:multprojector}
\end{equation}
is an orthogonal projector, the single-copy image of $\PW$ on the $\lambda$ sector, and the sector functionals are the compressions
\begin{equation}
 Q_{\lambda}(N)\equiv
 \bigl\|\Pi_{\lambda}^{W}\,\mathbb S_{\lambda}(N)\,\Pi_{\lambda}^{W}\bigr\|_{\mathrm{HS}}^{2}.
 \label{eq:singlecopysector}
\end{equation}

With $g_{2}:=(d-1)(d-2)/6$ the number of two-dimensional subspaces of $\Ftwo^{n}$, the ranks of the projectors $\Pi_{\lambda}^{W}$ are (see the end of this subsection)
\begin{align}
 \mathrm{rank}\,\Pi_{(4)}^{W}&=d+g_{2}=\frac{(d+1)(d+2)}6,\notag\\
 \mathrm{rank}\,\Pi_{(2,2)}^{W}&=(d-1)+2g_{2}=\frac{d^{2}-1}{3},\notag\\
 \mathrm{rank}\,\Pi_{(1^{4})}^{W}&=g_{2}.
 \label{eq:sectordims}
\end{align}
In the decomposition below the three sectors appear as mutually orthogonal subspaces of $W$, the $\lambda$ sector consisting of $\dim[\lambda]$ copies of a block of dimension $\mathrm{rank}\,\Pi_{\lambda}^{W}$; consistently, the dimensions saturate the stabilizer subspace:
\begin{equation}
 \mathrm{rank}\,\Pi_{(4)}^{W}+2\,\mathrm{rank}\,\Pi_{(2,2)}^{W}
 +\mathrm{rank}\,\Pi_{(1^{4})}^{W}=d^{2}=\dim W.
 \label{eq:dimcheck}
\end{equation}

\begin{proposition}[Exact three-sector decomposition]
\label{prop:SWsplit}
For every $d=2^{n}$, only three of the five irreducible representations of $S_{4}$ occur in $W=\mathrm{ran}\,\PW$, namely
\begin{equation}
 \lambda\in\{(4),(2,2),(1^{4})\}.
 \label{eq:threeisotypes}
\end{equation}
Consequently, Eq.~\eqref{eq:split} holds for every $V\in U(d)$, with $N=VMV^{\dagger}$.
\end{proposition}

\begin{proof}
Both $\PW$ and $N^{\otimes4}$ commute with $r(S_{4})$, so the compressed operator in Eq.~\eqref{eq:Gcompressed} is block diagonal over the $S_{4}$ sectors of $W$: with $K_{\lambda}:=\PW\Pi_{\lambda}^{S_{4}}=\Pi_{\lambda}^{S_{4}}\PW$, orthogonality of the blocks gives $\Gm(V)=d^{2}\sum_{\lambda}\|K_{\lambda}N^{\otimes4}K_{\lambda}\|_{\mathrm{HS}}^{2}$. In the stabilizer basis of Eq.~\eqref{eq:Phibasis}, each double transposition in the Klein subgroup
$V_{4}=\{e,(12)(34),(13)(24),(14)(23)\}$ fixes every $\ket{\Phi_{u,v}}$, since it only reparametrizes the summation variable $x$. Hence every permutation acts on $W$ only through its image in the quotient
$S_{4}/V_{4}\simeq S_{3}$, so the only $S_{4}$ representations that can occur in $W$ are those in which $V_{4}$ acts trivially: the trivial, the two-dimensional, and the sign representation of $S_{3}$, which, read as representations of $S_{4}$, are precisely $[4]$, $[2,2]$, and $[1^{4}]$. This proves Eq.~\eqref{eq:threeisotypes}. Finally, on the $\lambda$ sector Schur--Weyl duality factorizes $N^{\otimes4}=\mathbb S_{\lambda}(N)\otimes I_{[\lambda]}$ and the Pauli average factorizes $K_{\lambda}=\Pi_{\lambda}^{W}\otimes I_{[\lambda]}$, whence $\|K_{\lambda}N^{\otimes4}K_{\lambda}\|_{\mathrm{HS}}^{2}=\dim[\lambda]\,Q_{\lambda}(N)$; since $\dim[(4)]=\dim[(1^{4})]=1$ and $\dim[(2,2)]=2$, this gives Eq.~\eqref{eq:split}.

\end{proof}

Two further facts follow directly from the action of the Pauli group on the labels $(u,v)$ established in the proof of Lemma~\ref{lem:Phi} in Appendix~\ref{app:rank}.
The first is the set of sector ranks in Eq.~\eqref{eq:sectordims}.
The $S_{3}$ action on the labels $(u,v)$, described in the proof of Proposition~\ref{prop:SWsplit}, has one singleton orbit $(0,0)$, one three-point orbit $\{(a,0),(0,a),(a,a)\}$ for each of the $d-1$ directions $a\ne0$, and one six-point regular orbit for each of the $g_{2}$ two-dimensional subspaces $L=\langle a,b\rangle$.
The corresponding permutation representations of $S_{3}$ decompose as $\mathbf1$, $\mathbf1\oplus\mathrm{std}$, and $\mathbf1\oplus\mathrm{sgn}\oplus2\,\mathrm{std}$, respectively.  Collecting the trivial, standard, and sign pieces, which build the $(4)$, $(2,2)$, and $(1^{4})$ sectors, gives $\mathrm{rank}\,\Pi_{(4)}^{W}=1+(d-1)+g_{2}$, $\mathrm{rank}\,\Pi_{(2,2)}^{W}=(d-1)+2g_{2}$, and $\mathrm{rank}\,\Pi_{(1^{4})}^{W}=g_{2}$.

The second consequence is the value of each sector of Proposition~\ref{prop:SWsplit} for a diagonal matrix $M=\mathrm{diag}(\mu_{x})$, that is, for $V=\mathbb 1$ in Eq.~\eqref{eq:Gcompressed}.
Define
  \begin{equation}
   w_{0}:=\frac1d\sum_{x}\mu_{x}^{4},\qquad
   w_{a}:=\frac1d\sum_{x}\mu_{x}^{2}\mu_{x+a}^{2}\quad(a\ne0),
   \label{eq:waweights}
  \end{equation}
and, for $L=\langle a,b\rangle$ with $\dim L=2$,
  \begin{equation}
   w_{L}:=\frac1d\sum_{x}
   \mu_{x}\mu_{x+a}\mu_{x+b}\mu_{x+a+b},
   \label{eq:wLweight}
  \end{equation}
which is independent of the ordered basis $(a,b)$ of $L$.
Then
  \begin{align}
   \Qsym(M)&=w_{0}^{2}+\sum_{a\ne0}w_{a}^{2}
                +\sum_{\dim L=2}w_{L}^{2},\notag\\
   \Qtt(M)&=\sum_{a\ne0}w_{a}^{2}
                +2\sum_{\dim L=2}w_{L}^{2},\notag\\
   \Qa(M)&=\sum_{\dim L=2}w_{L}^{2},
   \label{eq:diagsectors}
  \end{align}
  and hence
  $\Gm(\mathbb 1)=d^{2}\bigl[w_{0}^{2}+3\sum_{a\ne0}w_{a}^{2}
  +6\sum_{\dim L=2}w_{L}^{2}\bigr]$, in direct agreement with the original Pauli fourth moment.
In terms of the plane sums $s_{L}$ of Appendix~\ref{app:shell}, $w_{L}=4s_{L}/d$.

\subsection{The symmetric sector}
\label{app:symSOS}
\label{sec:symmetric}

The goal of this subsection is to bound the symmetric functional $\Qsym$ in terms of the antisymmetric functional $\Qa$ and the antisymmetric stabilizer overlap $\cF$, at every $d=2^{n}$ and without finite enumeration; along the unitary orbit this removes the symmetric sector as an independent obstruction. The bound proceeds in two steps, quantified by two scalar functionals. The first is the Pauli fourth moment
\begin{equation}
 \Fpm(C):=\sum_{P\in\mathcal P_n}[\Tr(PC)]^{4};
 \label{eq:Fpmdef}
\end{equation}
a Pauli--Fourier collapse (Theorem~\ref{thm:symcollapse} below) shows that $\Qsym(N)-\Qa(N)\le d^{-3}\Fpm(N^{2})$, with equality at every computational diagonal. The second step dominates the fourth moment itself by $\cF$ and spectral invariants, through the defect
\begin{align}
 \Delta_d(C):={}&3d\,[\Tr(C^{2})]^{2}-2d\,\Tr(C^{4})\notag\\
 &+6d^{2}\cF(C)-\Fpm(C),
 \label{eq:Deltadef}
\end{align}
which Theorem~\ref{thm:symSOS} exhibits as a sum of squares. The defect vanishes at the reference point: for every computational diagonal $M$, Walsh orthogonality and the classification of zero-sum quadruples as planes give
\begin{equation}
 \Fpm(M)=3d[\Tr(M^{2})]^{2}-2d\Tr(M^{4})+6d^{2}\cF(M),
 \label{eq:diagFpm}
\end{equation}
so that $\Delta_d(M)=0$.

Index the Hermitian Paulis by the symplectic vector space
$\mathsf V=\Ftwo^{2n}$, writing $P_aP_b=(-1)^{[a,b]}P_bP_a$ and
$c_a=\Tr(P_aC)$.  Let
\begin{equation}
 \Afr(C):=\sum_{\{a,b\}:\,[a,b]=1}c_a^2c_b^2
 \label{eq:Afrdef}
\end{equation}
be the sum over unordered anticommuting pairs.  For $s\ne0$ let
$E_s=\{\{a,a+s\}:[a,s]=1\}$ and define the edge vector
$z_s(\{a,a+s\})=c_ac_{a+s}$.

\begin{theorem}[Pauli sum of squares]
\label{thm:symSOS}
For every $d=2^n$, $n\ge2$, and every Hermitian $C$, there are explicit
orthogonal projections $\mathsf Q_s$ on $\mathbb R^{E_s}$, of rank
$d(d-2)/8$, such that
\begin{equation}
 \Delta_d(C)=\frac4d\,\Afr(C)
 +\frac{16}{d}\sum_{s\ne0}\|\mathsf Q_s z_s(C)\|^2\ge0 .
 \label{eq:DeltaSOS}
\end{equation}
Equality holds if and only if the nonzero Pauli support of $C$ is pairwise
commuting, equivalently if $C$ is diagonal in a stabilizer basis.
\end{theorem}

The proof rests on an exact expansion of $\Delta_{d}$ over anticommuting pairs and noncommuting isotropic affine planes, regrouped into signed edge blocks whose universal spectrum follows from character orthogonality; it is given below.

\begin{theorem}[Symmetric sector collapse]
\label{thm:symcollapse}
For \emph{every} $d=2^{n}$ and every Hermitian $N$,
\begin{align}
 \Qsym(N)-\Qa(N)&=\frac1{d^4}\sum_{P,Q\in\mathcal P_n}[\Tr(PNQN)]^2\notag\\
 &\quad\times\Tr[(PNQN)^2]\ \le\ \frac1{d^3}\Fpm(N^2),
 \label{eq:symbound}
\end{align}
with equality for every computational diagonal $N$.
\end{theorem}

The identity and the bound of Eq.~\eqref{eq:symbound} follow from a Pauli--Fourier trade, proved at the end of this subsection.

Theorems~\ref{thm:symSOS} and~\ref{thm:symcollapse} combine especially cleanly on a unitary orbit. Subtracting the diagonal equality case of Eq.~\eqref{eq:symbound} at $M$ from the bound at $N=VMV^{\dagger}$ shows that any orbit increase of $\Qsym$ is paid for by $\Qa$ and the fourth moment; and since $\Delta_{d}\ge0$ with equality at diagonals, the alignment $\cF(N^{2})\le\cF(M^{2})$ already forces $\Fpm(N^{2})\le\Fpm(M^{2})$. The symmetric sector is therefore reduced, at every $d=2^{n}$, to the antisymmetric sector and the linear alignment $\cF(N^{2})\le\cF(M^{2})$. The mixed sector obeys an analogous reduction, but through a different representation-theoretic mechanism.

\subsubsection{Structure of the proof of Theorem~\ref{thm:symSOS}}

The structural steps are recorded first; the parts below supply the coefficient-level derivations.

Pauli conjugation invariance
forces every quartic monomial in $\Delta_d$ to have total label zero.  The
$c_a^4$ terms cancel, commuting pair terms cancel, and the remaining terms are
anticommuting pairs together with four distinct labels forming a noncommuting
isotropic affine plane.  If $\mathscr B$ denotes the latter planes and
$\varepsilon_F\in\{\pm1\}$ is the cyclic Pauli-product sign, the exact expansion is
\begin{equation}
 \Delta_d(C)=\frac{4(3d-4)}{d^2}\Afr(C)
 +\frac{64}{d^2}\sum_{F\in\mathscr B}\varepsilon_F\prod_{a\in F}c_a .
 \label{eq:DeltaPlane}
\end{equation}
Grouping the plane terms by their two anticommuting parallel classes gives
\begin{equation}
 \Delta_d(C)=\frac{16}{d^2}\sum_{s\ne0}
 z_s^{\mathsf T}\!\left(\frac{3d-4}{4}I+A_s\right)z_s,
 \label{eq:Deltablock}
\end{equation}
where $A_s$ is the signed adjacency matrix of the $s$-edge block.  A diagonal
switching identifies $A_s$ with $I-H_q$, where $H_q$ is a quadratic-form
character matrix on a nondegenerate symplectic space of size $d^2/4$.  Character
orthogonality yields
\begin{equation}
 A_s^2=2A_s+\left(\frac{d^2}{4}-1\right)I.
 \label{eq:Asquad}
\end{equation}
Hence
\begin{equation}
 \mathsf Q_s:=\frac1d\left[A_s+\left(\frac d2-1\right)I\right]
 \label{eq:Qsdef}
\end{equation}
is an orthogonal projection and
$(3d-4)I/4+A_s=dI/4+d\mathsf Q_s$.  Inserting this into
Eq.~\eqref{eq:Deltablock} and using
$\sum_{s\ne0}\|z_s\|^2=\Afr(C)$ proves Eq.~\eqref{eq:DeltaSOS} and its equality
statement.
Let $\mathsf V=\Ftwo^{2n}$ label the Hermitian Paulis and write
$C=d^{-1}\sum_a c_aP_a$.

\subsubsection{The antisymmetric stabilizer projector and the diagonal identity}

The fourth exterior power removes the Pauli cocycle: if
$P_aP_b=\omega(a,b)P_{a+b}$ with $\omega(a,b)\in\{\pm1,\pm i\}$, then
$\Lambda^4(P_a)\Lambda^4(P_b)=\omega(a,b)^4\Lambda^4(P_{a+b})$.  Thus
$a\mapsto\Lambda^4(P_a)$ is an ordinary unitary representation and
\begin{equation}
 \PiF=\frac1{d^2}\sum_{a\in\mathsf V}\Lambda^4(P_a)
\end{equation}
is its invariant-space projector.  In the computational wedge basis $e_F$,
$|F|=4$, its diagonal is $4/d$ when $F$ is a plane and zero otherwise.
Indeed, a nonzero translation stabilizes a four-set precisely when it is a plane; its translation stabilizer is then the four-element direction plane, and
the $Z$ character is trivial because $\sum_{x\in F}x=0$.  Consequently
\begin{equation}
 \mathrm{rank}\,\PiF=\frac{(d-1)(d-2)}6,
 \qquad
 \cF(C)=\frac1{d^2}\sum_{p\in\mathsf V}e_4(P_pC),
 \label{eq:appcFe4}
\end{equation}
where
\begin{equation}
 24e_4(X)=t_1^4-6t_1^2t_2+3t_2^2+8t_1t_3-6t_4,
 \qquad t_j=\Tr(X^j).
 \label{eq:Newton4}
\end{equation}
For $M=\mathrm{diag}(\mu_x)$, Walsh orthogonality gives
\begin{align}
 \Fpm(M)
 &=d\!\sum_{x_1+\cdots+x_4=0}\mu_{x_1}\mu_{x_2}\mu_{x_3}\mu_{x_4}\notag\\
 &=d\left[p_4+3(p_2^2-p_4)+24\sum_{F\ \mathrm{plane}}\prod_{x\in F}\mu_x\right],
\end{align}
while Eq.~\eqref{eq:appcFe4} gives
$\cF(M)=(4/d)\sum_F\prod_{x\in F}\mu_x$.  This proves
Eq.~\eqref{eq:diagFpm}.

\subsubsection{Coefficient expansion of $\Delta_d$}

Conjugation by $P_p$ maps $c_a\mapsto(-1)^{[p,a]}c_a$ and leaves $\Delta_d$
invariant.  Nondegeneracy of the symplectic form therefore forces each surviving
quartic monomial to have total label zero.  There are only three multiplicity
types: $c_a^4$, $c_a^2c_b^2$, and four distinct labels summing to zero.

For $c_a^4$ one has
\begin{align}
 [c_a^4](\Tr C^2)^2&=\frac1{d^2},&
 [c_a^4]\Tr C^4&=\frac1{d^3},\notag\\
 [c_a^4]\cF&=\frac{(d-1)(d-2)}{6d^4}.&&
\end{align}
so its coefficient in $\Delta_d$ is
$3/d-2/d^2+(d-1)(d-2)/d^2-1=0$.

For $a\ne b$, put $\sigma=(-1)^{[a,b]}$.  The six trace words give
\begin{equation}
 [c_a^2c_b^2]\Tr C^4=\frac{4+2\sigma}{d^3},
 \qquad [c_a^2c_b^2](\Tr C^2)^2=\frac2{d^2}.
\end{equation}
Using Eqs.~\eqref{eq:appcFe4}--\eqref{eq:Newton4}, character orthogonality gives
\begin{equation}
 [c_a^2c_b^2]\cF=\frac{\sigma(2-d)}{d^4}.
\end{equation}
Hence
\begin{equation}
 [c_a^2c_b^2]\Delta_d=(1-\sigma)\frac{6d-8}{d^2},
 \label{eq:paircoeffapp}
\end{equation}
which vanishes for commuting pairs and equals $4(3d-4)/d^2$ for
anticommuting pairs.

Now let $F=\{u_1,u_2,u_3,u_4\}$ contain four distinct labels with
$u_1+u_2+u_3+u_4=0$.  Write
$F=a+\mathrm{span}\{r,t\}$ and set
$\eta=[r,t]$, $\alpha=[a,r]$, $\beta=[a,t]$.  Its anticommutation graph is
exactly one of
\begin{align}
 \varnothing &: \eta=\alpha=\beta=0,\notag\\
 C_4 &: \eta=0,\quad(\alpha,\beta)\ne(0,0),\notag\\
 K_3\sqcup K_1 &: \eta=1.
 \label{eq:graphtri}
\end{align}
Thus the noncommuting isotropic planes are precisely the $C_4$ case.  For such a
plane define $\varepsilon_F$ by the Pauli product in any cyclic ordering,
$P_{u_1}P_{u_2}P_{u_3}P_{u_4}=\varepsilon_F I$; cyclic shifts and reversal
preserve this real sign.

For an arbitrary zero-sum four-set let $\lambda_\pi$ be the scalar phase of the
ordered product and define
\begin{align}
 S_0(F)&=\sum_{\pi\in S_4}\lambda_\pi,\\
 S_3(F)&=\sum_{\pi\in S_4}
 (-1)^{[u_{\pi(4)},u_{\pi(1)}+u_{\pi(3)}]}\lambda_\pi .
\end{align}
Expanding the $t_1t_3$ and $t_4$ terms in Newton's identity and averaging over the
Pauli root gives
\begin{equation}
 \left[\prod_{u\in F}c_u\right]\Delta_d
 =\frac2{d^2}[S_3(F)-S_0(F)].
 \label{eq:S30coeff}
\end{equation}
For the empty and four-cycle graphs, complementary pairs have equal commutation
type.  The signed phase
\begin{equation}
 \mu_\pi=(-1)^{[v_{4},v_{1}+v_{3}]}\lambda_\pi,
 \qquad v_{j}:=u_{\pi(j)},
\end{equation}
is therefore invariant under the adjacent transpositions generating $S_4$.
It equals $\tau$ on the empty graph and $\varepsilon_F$ on a cyclic ordering of
$C_4$.  Since each unordered pair occupies positions $\{2,4\}$ in four words,
\begin{equation}
 (S_0,S_3)=
 \begin{cases}
  (24\tau,24\tau),&\varnothing,\\
  (-8\varepsilon_F,24\varepsilon_F),&C_4.
 \end{cases}
\end{equation}
For $K_3\sqcup K_1$, reversal changes the sign associated with the odd number
of anticommuting pairs, so every $\lambda_\pi$ is purely imaginary.  The
coefficient in Eq.~\eqref{eq:S30coeff} is a coefficient of the real polynomial
$\Delta_d$ in the independent real variables $c_u$ and hence is real; therefore
$S_3-S_0=0$ (and $S_0=0$ separately by swapping two triangle vertices).  Thus
Eq.~\eqref{eq:S30coeff} is $64\varepsilon_F/d^2$ precisely on the $C_4$ planes
and zero otherwise.  Together with Eq.~\eqref{eq:paircoeffapp} this proves
Eq.~\eqref{eq:DeltaPlane}.

\subsubsection{Signed blocks and the universal spectrum}

For $s\ne0$ define $E_s$ and $z_s$ as defined above.  There are $d^2/2$
solutions of $[a,s]=1$, paired by $a\leftrightarrow a+s$, so $|E_s|=d^2/4$ and
\begin{equation}
 \sum_{s\ne0}\|z_s\|^2=\Afr(C).
 \label{eq:appzA}
\end{equation}
For distinct edges $e=\{a,a+s\}$ and $f=\{b,b+s\}$, their endpoints form a
noncommuting isotropic plane.  Set $(A_s)_{e,f}=\varepsilon_{F(e,f)}$ and
$(A_s)_{e,e}=0$.  Every $C_4$ plane has exactly two anticommuting parallel
classes, and each is counted in both matrix orders; hence each plane monomial
occurs four times in $\sum_s z_s^{\mathsf T}A_sz_s$.  This proves
Eq.~\eqref{eq:Deltablock}.

It remains to diagonalize $A_s$.  Choose $a_0$ with $[a_0,s]=1$ and a complement
$U_s$ of $\langle s\rangle$ in $s^\perp$.  The restriction
$\omega=[\cdot,\cdot]|_{U_s}$ is nondegenerate and
$U_s\cong\Ftwo^{2n-2}$.  Every edge has the unique representative
$e_u=\{a_0+u,a_0+u+s\}$, $u\in U_s$.  Define
$\rho_s(u)\in\{\pm1\}$ by
\begin{equation}
 P_{a_0+u}P_{a_0+u+s}=i\rho_s(u)P_s .
\end{equation}
A comparison with a cyclic ordering of the corresponding four-cycle gives
\begin{align}
 (A_s)_{u,v}&=-\rho_s(u)\rho_s(v)
 (-1)^{\ell(u)+\ell(v)+\omega(u,v)},\notag\\
 \ell(u)&=[a_0,u].
 \label{eq:Asentriesapp}
\end{align}
Choose a quadratic refinement $q$ of $\omega$ and switch by
$D_{u,u}=\rho_s(u)(-1)^{q(u)+\ell(u)}$.  Then
\begin{equation}
 DA_sD=I-H_q,
 \qquad (H_q)_{u,v}=(-1)^{q(u+v)}.
\end{equation}
Quadratic polarization and character orthogonality yield
\begin{align}
 (H_q^2)_{u,w}
 &=(-1)^{q(u+w)+\omega(u,w)}
 \sum_v(-1)^{\omega(u+w,v)}\notag\\
 &=\frac{d^2}{4}\,\delta_{u,w}.
\end{align}
Therefore Eq.~\eqref{eq:Asquad} holds.  Its two eigenvalues are
$1\pm d/2$; since $\Tr A_s=0$, their multiplicities are respectively
$d(d-2)/8$ and $d(d+2)/8$.  Equation~\eqref{eq:Qsdef} is consequently the
spectral projector onto the $1+d/2$ eigenspace.  Substitution into
Eq.~\eqref{eq:Deltablock}, together with Eq.~\eqref{eq:appzA}, proves the SOS
Eq.~\eqref{eq:DeltaSOS}.  Since its first term is strictly coercive in every
anticommuting product $c_ac_b$, equality is equivalent to pairwise commuting
Pauli support, which lies in a maximal isotropic subspace and is therefore
stabilizer-diagonal.

\subsubsection{Pauli--Fourier proof of Theorem~\ref{thm:symcollapse}}
\label{app:symcollapse}

We prove Theorem~\ref{thm:symcollapse}.  For $P,Q\in\mathcal P_n$
write $X_{P,Q}=PNQN$ and $t_k=\Tr(X_{P,Q}^k)$.  Expanding both copies of $\PW$
and the central $S_4$ projector gives, for every $\lambda\vdash4$,
\begin{equation}
 Q_\lambda(N)=\frac1{d^4}\sum_{P,Q}
 \frac1{24}\sum_{\pi\in S_4}\chi^\lambda(\pi)
 \prod_{c\in\mathrm{cyc}(\pi)}t_{|c|}.
 \label{eq:appcycles}
\end{equation}
Thus the symmetric and antisymmetric sectors contain $h_4$ and $e_4$,
respectively.  Since
\begin{equation}
 h_4-e_4=\frac12(t_1^2t_2+t_4),
 \label{eq:h4e4}
\end{equation}
it suffices to evaluate the two Pauli sums.

Let $\chi(P,R)=\pm1$ according as $P$ and $R$ commute or anticommute.  The
characters obey
\begin{equation}
 \sum_P\chi(P,R)\chi(P,S)=d^2\delta_{R,S}.
 \label{eq:charorth}
\end{equation}
For Hermitian $A$, set $a_R=\Tr(RA)$.  From
$A=d^{-1}\sum_Ra_RR$ one obtains
\begin{equation}
 \Tr(APAP)=\frac1d\sum_R\chi(P,R)a_R^2.
 \label{eq:APAP}
\end{equation}
Squaring and summing Eq.~\eqref{eq:APAP} gives
\begin{equation}
 \sum_P[\Tr(APAP)]^2=\sum_P[\Tr(AP)]^4.
 \label{eq:Fourier1app}
\end{equation}
A direct four-word expansion, followed by Eq.~\eqref{eq:charorth}, likewise gives
\begin{align}
 \sum_P\Tr[(AP)^4]
 &=\sum_P[\Tr(AP)]^2\Tr(APAP)\notag\\
 &=\frac1d\sum_{P,R}\chi(P,R)a_P^2a_R^2.
 \label{eq:Fourier2app}
\end{align}
Applying Eq.~\eqref{eq:Fourier2app} with $A=NQN$ to Eq.~\eqref{eq:h4e4} proves
\begin{align}
 \Qsym(N)-\Qa(N)
 &=\frac1{d^4}\sum_{P,Q}[\Tr(PNQN)]^2\notag\\
 &\qquad\times\Tr[(PNQN)^2].
\end{align}
Since $a_P^2a_R^2\ge0$ and $\chi(P,R)\le1$,
\begin{equation}
 \frac1d\sum_{P,R}\chi(P,R)a_P^2a_R^2
 \le\frac1d\left(\sum_Pa_P^2\right)^2
 =d[\Tr(A^2)]^2,
\end{equation}
where the last equality is Pauli Parseval.  For $A=NQN$,
$\Tr(A^2)=\Tr(N^2QN^2Q)$.  Summing over $Q$ and using
Eq.~\eqref{eq:Fourier1app} for $N^2$ gives
\begin{equation}
 \Qsym(N)-\Qa(N)\le d^{-3}\Fpm(N^2),
\end{equation}
as claimed.

Finally suppose $N=M=\mathrm{diag}(\mu_x)$.  For
$Q\propto X^aZ^b$, the Pauli support of $A=MQM$ lies among
$X^aZ^{b'}$, with coefficient proportional to the Walsh transform of
$g(x)=\mu_x\mu_{x+a}$.  Since $g(x+a)=g(x)$, a surviving coefficient satisfies
$a\cdot b'=a\cdot b$.  Any two surviving labels therefore obey
$a\cdot(b'+b'')=0$ and commute.  Hence every $\chi(P,R)$ on the support in
Eq.~\eqref{eq:Fourier2app} equals $+1$, proving equality for every computational
diagonal $M$.

\subsection{The mixed sector}
\label{app:mixed}
\label{sec:mixed}

The $(2,2)$ sector admits an equally explicit certificate, obtained below by realizing the sector inside the exterior square $\Lambda^{2}\mathbb C^{d}$: the antisymmetric Paulis form an orthonormal frame of $\Lambda^{2}\mathbb C^{d}$, and two elementary sums of squares over that frame---an edge/triangle inequality and a linear bridge---combine into the following bound.

\begin{theorem}[Dimension-free mixed-sector certificate]
\label{thm:mixedSOS}
For every $d=2^n$, $n\ge2$, and every Hermitian $N$, define
\begin{align}
 D_{22}(N):={}&\frac8{d^2}\Tr[(\Lambda^2N)^4]
 +\frac{12}{d}\cF(N^2)\notag\\
 &+4\Qa(N)-2\,\Qtt(N).
 \label{eq:D22def}
\end{align}
Then $D_{22}(N)$ is an explicit sum of squares over the exterior-square Pauli frame constructed below [Eq.~\eqref{eq:D22cert}]; in particular,
\begin{equation}
 \Qtt(N)\le\frac4{d^2}\Tr[(\Lambda^2N)^4]
 +\frac{6}{d}\cF(N^2)+2\Qa(N).
 \label{eq:22bound}
\end{equation}
Equality holds for every real computational diagonal $N$.
\end{theorem}

Because the first term in Eq.~\eqref{eq:22bound} is a unitary invariant---$\Tr[(\Lambda^{2}N)^{4}]=\tfrac12\{[\Tr(N^{4})]^{2}-\Tr(N^{8})\}$---subtracting the diagonal equality case at $M$ from the bound at $N=VMV^{\dagger}$ shows that any orbit increase of $\Qtt$ is controlled by $\cF$ and $\Qa$ alone: if $\cF(N^{2})\le\cF(M^{2})$ and $\Qa(N)\le\Qa(M)$, then $\Qtt(N)\le\Qtt(M)$. The mixed sector thus carries no obstruction beyond the two antisymmetric/Fano quantities.

The remainder of this subsection proves Theorem~\ref{thm:mixedSOS} and, in particular, identifies the
exterior-square construction analytically with the $(2,2)$ Schur--Weyl sector of
Proposition~\ref{prop:SWsplit}.  Throughout, Pauli labels lie in
$\mathsf V=\Ftwo^{2n}$ with quadratic form $q(x,z)=x\cdot z$ and symplectic form
$[(x,z),(x',z')]=x\cdot z'+z\cdot x'$.

\subsubsection{Exterior-square Pauli frame and the two multiplicity spaces}
\label{app:mixedframe}

The Hermitian Pauli satisfies $P_a^{\mathsf T}=(-1)^{q(a)}P_a$.  Hence
$\mathcal A=\{a:q(a)=1\}$ labels precisely the antisymmetric Paulis.  Since
$\sum_a(-1)^{q(a)}=d$,
\begin{equation}
 |\mathcal A|=\frac{d^2-d}{2}=\dim\Lambda^2\mathbb C^d,
\end{equation}
and $E_a=d^{-1/2}P_a$, $a\in\mathcal A$, is an orthonormal basis of
$\Lambda^2\mathbb C^d$ under the antisymmetric-matrix realization
$\Lambda^2(A)\alpha=A\alpha A^{\mathsf T}$.

For Hermitian $N$
define the matrix of $\Lambda^2N$ in this frame and its two Hadamard squares,
\begin{align}
 K_{ab}(N)&:=\langle E_a,\Lambda^2(N)E_b\rangle,
 \notag\\
 B&:=K\circ K,\qquad C:=K^2\circ K^2 .
 \label{eq:KBCmain}
\end{align}
For $a\ne b$ put
\begin{equation}
 \varepsilon_{ab}:=-(-1)^{[a,b]},\qquad
 (\Sigma_d)_{aa}=0,\quad (\Sigma_d)_{ab}=\varepsilon_{ab},
 \label{eq:SigmaMain}
\end{equation}
and define
\begin{equation}
 \rho_\varepsilon(z):=
 \begin{cases}\Re z,&\varepsilon=+1,\\ \Im z,&\varepsilon=-1.\end{cases}
 \label{eq:rhoMain}
\end{equation}
Finally let
\begin{align}
 \mathcal S_d(K)&:=\Tr C-\Tr B^2+\Tr(\Sigma_dB^2),
 \label{eq:Smain}\\
 \mathcal R_d(K)&:=\sum_{a<b}
 \rho_{\varepsilon_{ab}}\!\bigl((K^2)_{ab}\bigr)^2 .
 \label{eq:Rmain}
\end{align}

The certificate of Theorem~\ref{thm:mixedSOS} to be established is
\begin{equation}
 D_{22}(N)=\frac4d\,\mathcal S_d(K(N))
 +\frac{32}{d^2}\,\mathcal R_d(K(N))\ge0 .
 \label{eq:D22cert}
\end{equation}

For a Pauli $P_p$,
\begin{equation}
 \Lambda^2(P_p)E_a=(-1)^{q(p)+[a,p]}E_a.
 \label{eq:L2PauliApp}
\end{equation}
Thus the basis vector $E_a\odot E_b$ of
$\operatorname{Sym}^2(\Lambda^2\mathbb C^d)$ carries character
$(-1)^{[a+b,p]}$.  By nondegeneracy of the symplectic form it is fixed for every
$p$ if and only if $a=b$.  Therefore the Pauli-fixed subspace is
\begin{equation}
 \mathcal M=\operatorname{span}\{E_a\otimes E_a:a\in\mathcal A\}.
 \label{eq:MfixedApp}
\end{equation}
Let $\mathscr D:\mathbb C^{\mathcal A}\to\mathcal M$ be the isometry
$\mathscr D e_a=E_a\otimes E_a$.

Throughout this subsection, $\alpha\wedge\beta$ denotes one half of the ordinary wedge product of two-forms, which keeps the constants below simple.
Consider the wedge map
\begin{equation}
 \mathscr W:\operatorname{Sym}^2(\Lambda^2\mathbb C^d)\longrightarrow
 \Lambda^4\mathbb C^d,
 \qquad \mathscr W(\alpha\otimes\beta)=\alpha\wedge\beta .
 \label{eq:wedgeMapApp}
\end{equation}
It is surjective and satisfies
$\mathscr W\operatorname{Sym}^2(\Lambda^2A)=\Lambda^4(A)\mathscr W$.  Consequently
\begin{equation}
 \ker\mathscr W\simeq\mathbb S_{(2,2)}\mathbb C^d,
 \qquad
 (\ker\mathscr W)^\perp\simeq\Lambda^4\mathbb C^d,
 \label{eq:plethApp}
\end{equation}
which is the elementary realization of
$\operatorname{Sym}^2(\Lambda^2)=\mathbb S_{(2,2)}\oplus\Lambda^4$.  Both summands
are Pauli invariant, hence their orthogonal projectors preserve the fixed subspace
$\mathcal M$.  Under $\mathscr D$, let $\mathsf P_{22}$ and $\mathsf P_{\rm F}$ be
the resulting complementary projections on $\mathbb C^{\mathcal A}$.

This identifies the frame sectors with the ranges of the sector projectors
$\Pi_{\lambda}^{W}$ of Proposition~\ref{prop:SWsplit}: the Pauli-fixed part of $\ker\mathscr W$ is
$\mathrm{ran}\,\Pi_{(2,2)}^{W}$, while the Pauli-fixed part of the complementary summand maps
isomorphically to $\mathrm{ran}\,\PiF\subset\Lambda^4\mathbb C^d$.  Therefore,
for $K=K(N)$ and $B=K\circ K$,
\begin{align}
 \Qtt(N)&=\|\mathsf P_{22}B\mathsf P_{22}\|_{\mathrm{HS}}^2,
 \label{eq:Q22FrameApp}\\
 \Qa(N)&=\|\PiF\Lambda^4(N)\PiF\|_{\mathrm{HS}}^2 .
\end{align}
This is the same $Q_{(2,2)}$ as in Eq.~\eqref{eq:singlecopysector};
the multiplicity factor $\dim[2,2]=2$ enters the sector split explicitly in
Eq.~\eqref{eq:split}.

For later use, because $K(A)$ is the matrix of $\Lambda^2A$ in an orthonormal
basis,
\begin{align}
 K(A^2)&=K(A)^2,\notag\\
 I_4(N)&:=\Tr K(N)^4=\Tr[(\Lambda^2N)^4]\notag\\
 &=\frac12\{[\Tr(N^4)]^2-\Tr(N^8)\}.
 \label{eq:I4App}
\end{align}

\subsubsection{Wedge-square Gram matrix and the exact intertwiner}
\label{app:mixedintertwiner}

Define $G:\mathbb C^{\mathcal A}\to\Lambda^4\mathbb C^d$ by
$Ge_a=E_a\wedge E_a$.  The basic identity for antisymmetric matrices is
\begin{equation}
 \langle A\wedge A,B\wedge B\rangle
 =\frac12(\Tr A^\dagger B)^2-\Tr[(A^\dagger B)^2].
 \label{eq:wedgeGramApp}
\end{equation}
It follows by expanding both wedge squares in coordinates and grouping the $24$
permutations into the three pairings of four indices.  Applying
Eq.~\eqref{eq:wedgeGramApp} to $E_a,E_b$ gives
\begin{align}
 \Gamma:=G^\dagger G&=\frac{d-2}{2d}I+\frac1d\Sigma_d,
 \label{eq:GammaApp}\\
 (\Sigma_d)_{ab}&=-(-1)^{[a,b]}\qquad(a\ne b).\notag
\end{align}

The operator $\mathscr W\mathscr W^\dagger$ is a positive $U(d)$-equivariant
endomorphism of the irreducible representation $\Lambda^4\mathbb C^d$, hence is
scalar by Schur's lemma.  Since
\begin{equation}
 \rank\PiF=\frac{(d-1)(d-2)}6,
 \qquad
 \Tr\Gamma=\frac{(d-1)(d-2)}4,
\end{equation}
the scalar is $3/2$.  Thus
\begin{align}
 G^\dagger G&=\frac32\mathsf P_{\rm F},\qquad
 GG^\dagger=\frac32\PiF,\notag\\
 \Sigma_d&=\frac{3d}{2}\mathsf P_{\rm F}-\frac{d-2}{2}I .
 \label{eq:FanoFrameApp}
\end{align}
In particular
$\rank\mathsf P_{\rm F}=(d-1)(d-2)/6$ and
$\rank\mathsf P_{22}=(d^2-1)/3$, as required by
Eq.~\eqref{eq:sectordims}.

For Hermitian $A$ write $\widehat B(A)=K(A)\circ K(A)$.  If
$F_b=\Lambda^2(A)E_b$, Eq.~\eqref{eq:wedgeGramApp} gives
\begin{equation}
 (G^\dagger\Lambda^4(A)G)_{ab}
 =\frac12K(A)_{ab}^2-\Tr[(E_aF_b)^2].
\end{equation}
Expanding $F_b$ in the Pauli frame and using
\begin{equation}
 \Tr(E_aE_cE_aE_{c'})=\frac{(-1)^{[a,c]}}{d}\,\delta_{cc'}
\end{equation}
yields
\begin{equation}
 \Tr[(E_aF_b)^2]
 =\frac1d[(I-\Sigma_d)\widehat B(A)]_{ab}.
\end{equation}
Therefore
\begin{equation}
 G^\dagger\Lambda^4(A)G
 =\Gamma\widehat B(A)
 =\frac32\mathsf P_{\rm F}\widehat B(A).
 \label{eq:intertwinerApp}
\end{equation}
The left-hand side is Hermitian, so
$[\mathsf P_{\rm F},\widehat B(A)]=0$; moreover
\begin{equation}
 G\widehat B(A)G^\dagger
 =\frac32\PiF\Lambda^4(A)\PiF .
 \label{eq:intertwiner2App}
\end{equation}
Taking traces and Hilbert--Schmidt norms in these identities gives the exact sector
dictionary.  For
\begin{align}
 x&=\Tr(\mathsf P_{22}B^2),&
 y&=\Tr(\mathsf P_{\rm F}B^2),\notag\\
 t&=\Tr(\mathsf P_{22}C),&
 u&=\Tr(\mathsf P_{\rm F}C).
 \label{eq:xyzuApp}
\end{align}
one has
\begin{equation}
 x=\Qtt(N),\qquad y=\Qa(N),\qquad u=\cF(N^2).
 \label{eq:dictionaryApp}
\end{equation}
For example, the first identity follows from
$[\mathsf P_{22},B]=0$ and Eq.~\eqref{eq:Q22FrameApp}; the second follows from
Eq.~\eqref{eq:intertwiner2App}; and the third follows by tracing
Eq.~\eqref{eq:intertwinerApp} at $A=N^2$.

\subsubsection{Triality curvature: an edge/triangle sum of squares}
\label{app:mixedtriality}

For a Hermitian matrix $K$, write $k_a=K_{aa}\in\mathbb R$ and $z_{ab}=K_{ab}$.
The elementary identity
\begin{equation}
 |z|^2+\varepsilon\Re(z^2)=2\rho_\varepsilon(z)^2
 \label{eq:rhoApp}
\end{equation}
will be used repeatedly.  Since
$C_{aa}=((K^2)_{aa})^2$, direct expansion gives
\begin{align}
 \Tr C&=\sum_{a,b,c}|K_{ab}|^2|K_{ac}|^2,\notag\\
 \Tr B^2&=\sum_{a,b}|K_{ab}|^4,\notag\\
 \Tr(\Sigma_dB^2)&=\sum_{a\ne b}\sum_c
 \varepsilon_{ab}K_{bc}^2K_{ca}^2 .
 \label{eq:threeTracesApp}
\end{align}
Grouping the summands by supports of size one, two, and three yields cancellation of
the one-index terms and the exact SOS
\begin{align}
 \mathcal S_d(K)={}&4\sum_{a<b}(k_a^2+k_b^2)
 \rho_{\varepsilon_{ab}}(z_{ab})^2\notag\\
 &+4\sum_{a<b<c}\Big[
 \rho_{\varepsilon_{ab}}(z_{ac}\overline{z_{bc}})^2
 +\rho_{\varepsilon_{ac}}(z_{ab}z_{bc})^2\notag\\
 &\hspace{7.2em}
 +\rho_{\varepsilon_{bc}}(z_{ab}\overline{z_{ac}})^2\Big]\ge0.
 \label{eq:trialitySOSApp}
\end{align}
Indeed, the two-index contribution for $\{a,b\}$ is
$2(k_a^2+k_b^2)[|z_{ab}|^2+\varepsilon_{ab}\Re z_{ab}^2]$, while the
three-index contribution is the sum of the three analogous expressions obtained by
pairing each modulus product with the Seidel sign of the opposite edge.

Using Eq.~\eqref{eq:FanoFrameApp} and
$\mathsf P_{22}+\mathsf P_{\rm F}=I$, Eq.~\eqref{eq:Smain} becomes
\begin{equation}
 \mathcal S_d(K)=t+u-\frac d2x+dy.
 \label{eq:SsectorApp}
\end{equation}

\subsubsection{The linear bridge and completion of the certificate}
\label{app:mixedbridge}

For any Hermitian matrix $\widehat K$ and
$\widehat B=\widehat K\circ\widehat K$,
\begin{align}
 &\Tr\widehat K^2-\Tr\widehat B+
 \Tr(\Sigma_d\widehat B)\notag\\
 &\qquad=4\sum_{a<b}
 \rho_{\varepsilon_{ab}}(\widehat K_{ab})^2\ge0.
 \label{eq:linearBridgeApp}
\end{align}
To see this, note that the first two traces give
$2\sum_{a<b}|\widehat K_{ab}|^2$ and the Seidel term gives
$2\sum_{a<b}\varepsilon_{ab}\Re(\widehat K_{ab}^2)$, then use
Eq.~\eqref{eq:rhoApp} termwise.

Apply Eq.~\eqref{eq:linearBridgeApp} to
$\widehat K=K(N^2)=K(N)^2$ and $\widehat B=C$.  Using
Eqs.~\eqref{eq:FanoFrameApp}, \eqref{eq:xyzuApp}, and \eqref{eq:I4App} gives
\begin{equation}
 I_4(N)-\frac d2t+du
 =4\mathcal R_d(K)\ge0.
 \label{eq:bridgeSectorApp}
\end{equation}
Combining Eqs.~\eqref{eq:SsectorApp} and~\eqref{eq:bridgeSectorApp},
\begin{align}
 \frac4d\mathcal S_d+\frac{32}{d^2}\mathcal R_d
 &=\frac4d\left(t+u-\frac d2x+dy\right)
 +\frac8{d^2}\left(I_4-\frac d2t+du\right)\notag\\
 &=\frac8{d^2}I_4+\frac{12}{d}u+4y-2x,
\end{align}
which is Eq.~\eqref{eq:D22cert} by the dictionary
Eq.~\eqref{eq:dictionaryApp}.

\subsubsection{Sharpness at computational diagonals}
\label{app:mixedsharp}

Let $M=\operatorname{diag}(\mu_r)_{r\in\Ftwo^n}$ be real diagonal.  If
$a=(x_a,z_a)$ and $b=(x_b,z_b)$ lie in $\mathcal A$, a direct matrix-element
calculation gives
\begin{equation}
 K(M)_{ab}=0\quad\text{unless }x_a=x_b .
 \label{eq:KblockApp}
\end{equation}
For fixed nonzero $x$, the corresponding block consists of the $d/2$ labels
$(x,z)$ with $x\cdot z=1$, all its entries are real, and for distinct labels in
the same block
\begin{equation}
 [a,b]=x\cdot z_b+z_a\cdot x=0,
 \qquad \varepsilon_{ab}=-1.
 \label{eq:epsBlockApp}
\end{equation}
Thus every edge square in Eq.~\eqref{eq:trialitySOSApp} vanishes.  If three labels
lie in one block, all three matrix entries are real and all three signs are $-1$,
so every triangle square also vanishes; if they do not lie in one block, at least
two of the three off-diagonal entries vanish by Eq.~\eqref{eq:KblockApp}.
Therefore $\mathcal S_d(K(M))=0$.  The same block structure holds for $K(M)^2$,
so Eq.~\eqref{eq:Rmain} also gives $\mathcal R_d(K(M))=0$.  Hence
$D_{22}(M)=0$, proving the equality statement in Theorem~\ref{thm:mixedSOS}.

\subsection{The reduced form of the stabilizer purity}
\label{app:reduced}
\label{ssec:master}

We now assemble the reduced form of Proposition~\ref{prop:master}.

Together, Theorems~\ref{thm:symSOS}, \ref{thm:symcollapse}, and~\ref{thm:mixedSOS} control both non-antisymmetric sectors by antisymmetric data.  Define the remaining symmetric correction
\begin{equation}
 D_{\mathrm{sym}}(N):=\frac1{d^3}\Fpm(N^2)-\Qsym(N)+\Qa(N),
 \label{eq:Dsymdef}
\end{equation}
which is nonnegative by Eq.~\eqref{eq:symbound}; both $D_{\mathrm{sym}}$ and
$D_{22}$ vanish at every computational diagonal.

The reduced form of the diagonal functional now assembles by pure substitution. Inserting Eqs.~\eqref{eq:Dsymdef} and~\eqref{eq:D22def} into the sector decomposition~\eqref{eq:split} eliminates $\Qsym$ and $\Qtt$ in favor of $\Qa$, $\cF$, the fourth moment, and the two defects; the wedge trace is spectral, $\Tr[(\Lambda^{2}N)^{4}]=\tfrac12\{[\Tr(N^{4})]^{2}-\Tr(N^{8})\}$; and Eq.~\eqref{eq:Deltadef} trades the remaining $\Fpm(N^{2})$ for $\cF(N^{2})$ and the third defect, $\Delta_{d}(N^{2})$. The spectral terms combine to $7[\Tr(N^{4})]^{2}-6\Tr(N^{8})$, the $V$-dependent terms to $6d^{2}\Qa(N)+18d\,\cF(N^{2})$, and the three defects collect with negative signs.

This is Eq.~\eqref{eq:master} of Proposition~\ref{prop:master}, with the spectral invariant $\mathcal I(\boldsymbol\mu)$ of Eq.~\eqref{eq:specinv} and
\begin{equation}
 \mathcal D(N)\equiv d^{2}D_{\mathrm{sym}}(N)+d^{2}D_{22}(N)
 +\frac1d\,\Delta_{d}(N^{2})\ \ge\ 0 ,
 \label{eq:defectdef}
\end{equation}
which vanishes at every computational diagonal because each of its three terms does.

All $V$-dependence of $\Gm$ now resides in the two nonnegative antisymmetric-sector quantities $\Qa(N)$ and $\cF(N^{2})$ and in the nonnegative $\mathcal D$.

\section{The antisymmetric stabilizer subspace and the proofs of the two sector inequalities}
\label{app:flatcode}

This appendix proves Theorems~\ref{thm:rsevencap} and~\ref{thm:cfsix}.
We first establish the structural facts about the antisymmetric stabilizer projector summarized in Sec.~\ref{ssec:mechanism} (Appendix~\ref{ssec:code}) and state the two operator inequalities (Appendix~\ref{ssec:caps}), with the exterior-algebra proofs of both collected in Appendices~\ref{app:flatmaps}--\ref{app:blossomproof}.
We then derive the two sector inequalities, Theorem~\ref{thm:cfsix} from the matching polytope of $K_{6}$ (Appendix~\ref{ssec:matching}) and Theorem~\ref{thm:rsevencap} from the self-similarity of the antisymmetric stabilizer subspace (Appendices~\ref{ssec:ocseven} and~\ref{app:selfsim}).

\subsection{The antisymmetric stabilizer subspace}
\label{ssec:code}

We use the exterior algebra of $\mathbb C^{d}$: for a $k$-subset $S=\{s_{1}<\cdots<s_{k}\}\subset\Ftwo^{n}$, ordered as binary integers, $e_{S}:=e_{s_{1}}\wedge\cdots\wedge e_{s_{k}}$, and the $e_{S}$ form an orthonormal basis of $\Lambda^{k}\mathbb C^{d}$, with inner product $\langle\cdot,\cdot\rangle$ antilinear in the first slot. The interior product $\iota_{v}=\sum_{z}v_{z}\,\iota_{e_{z}}$ is \emph{bilinear} in $v$, with $\iota_{u\wedge v}=\iota_{v}\iota_{u}$, and the Hermitian annihilation and creation operators are $a_{v}:=\iota_{\bar v}$ and $a_{v}^{\dagger}=v\wedge\cdot$, with $\{a_{u},a_{v}^{\dagger}\}=\langle u,v\rangle\mathbb 1$ and $n_{x}:=a_{x}^{\dagger}a_{x}$; Appendix~\ref{app:flatconv} collects the remaining conventions. On the combinatorial side, a four-subset $F\subset\Ftwo^{n}$ with $\bigoplus_{z\in F}z=0$ is exactly a plane, and it determines its direction $\mathrm{dir}\,F=\{z\oplus z':z,z'\in F\}$; the analogous statements hold for affine subspaces of dimension three, eight-point sets $\Delta=x+L$ with $\dim L=3$, which we call \emph{three-spaces}. The projector $\PiF$ of Eq.~\eqref{eq:flatprojector} is extended by zero off $\Lambda^{4}\mathbb C^{d}$.

The proofs of Theorems~\ref{thm:rsevencap} and~\ref{thm:cfsix} rest on five structural facts about $\PiF$, each false for a generic subspace of $\Lambda^{4}\mathbb C^{d}$, and each with a short geometric reason. We present them as a dictionary, proving on the spot only the two that take a few lines; the remaining proofs are given in the subsections below.

\emph{(i) The projector has a sign-free basis.} The first fact fixes the signs of the basis vectors: the Pauli average selects, for each direction, the \emph{all-plus} combination of the wedges of parallel planes.

\begin{lemma}[Translation parity]
\label{lem:transpar}
For every plane $F$ and every $u\in\Ftwo^{n}$,
$\Lambda^{4}(X^{u})\,e_{F}=+\,e_{u\oplus F}$.
\end{lemma}

\begin{proof}
Translation by $u$ maps the sorted tuple of $F$ to a tuple of $u\oplus F$ whose sorting sign is $(-1)^{\#\mathrm{inv}}$, where an unordered pair $\{f,f'\}\subset F$ is inverted precisely when translation reverses its order. Two labels first differ at the most significant bit of $f\oplus f'$, and translation flips their order if and only if $u$ has a one at that bit: whether a pair inverts depends only on its difference $f\oplus f'$. In a plane every nonzero difference $w\in L=\mathrm{dir}\,F$ occurs in exactly two of the six pairs, so the inversions come in pairs and their number is even.
\end{proof}

\begin{corollary}[Basis of the antisymmetric stabilizer subspace]
\label{cor:codeform}
For each two-dimensional $L\le\Ftwo^{n}$ set
\begin{equation}
 \psi_{L}:=\sqrt q\sum_{F\parallel L}e_{F},
\end{equation}
a unit vector ($d/4=q^{-1}$ parallel planes, disjoint supports). Then $\PiF=\sum_{L}\ketbra{\psi_{L}}{\psi_{L}}$: the $\psi_{L}$ are an orthonormal basis of $W_{\mathrm F}:=\mathrm{ran}\,\PiF$, the antisymmetric stabilizer subspace, of dimension $g_{2}$.
\end{corollary}

\begin{proof}
In Eq.~\eqref{eq:flatprojector}, $\Lambda^{4}$ of the Hermitian Pauli $i^{u\cdot z}X^{u}Z^{z}$ carries the phase $i^{4u\cdot z}=1$, acts on $e_{S}$ as $(-1)^{z\cdot\sigma(S)}\Lambda^{4}(X^{u})e_{S}$ with $\sigma(S)=\bigoplus_{s\in S}s$, and the $z$-average annihilates every $e_{S}$ with $\sigma(S)\ne0$, i.e., every wedge coordinate that is not a plane. On a plane coordinate the $z$-average is trivial and Lemma~\ref{lem:transpar} gives $\PiF e_{F}=\frac1d\sum_{u}e_{u\oplus F}=q\sum_{F'\parallel F}e_{F'}=\sqrt q\,\psi_{\mathrm{dir}F}$, which is the stated projector. Orthonormality holds because a plane determines its direction, and the dimension count is Eq.~\eqref{eq:sectordims}.
\end{proof}

\emph{(ii) The one-hole layer is unitary.} The second fact is a Steiner-type completion property: every coordinate triple completes uniquely to a plane, so contracting one mode out of $W_{\mathrm F}$ sweeps $\Lambda^{3}\mathbb C^{d}$ exactly once, with the exact count $d\,g_{2}=\binom d3$.

\begin{theorem}[Unitary annihilation tensor]
\label{thm:annU}
The annihilation tensor
\begin{equation}
 D:\mathbb C^{d}\otimes W_{\mathrm F}\to\Lambda^{3}\mathbb C^{d},
 \qquad D(v\otimes\psi)=\iota_{v}\psi ,
\end{equation}
is $\sqrt{q}$ times a signed permutation in the bases $\{e_{z}\otimes\psi_{L}\}$ and $\{e_{T}\}_{|T|=3}$; consequently
\begin{equation}
 D^{\dagger}D=DD^{\dagger}=q\,\mathbb 1 .
 \label{eq:annU}
\end{equation}
\end{theorem}

\begin{proof}
Fix $z$ and $L$, and let $C_{z}(L)$ be the unique plane of direction $L$ containing $z$. In $\iota_{e_{z}}\psi_{L}$ only the plane $C_{z}(L)$ contributes, giving $\pm\sqrt{q}\,e_{T}$ with $T=C_{z}(L)\setminus\{z\}$: every column is $\pm$ one basis vector. The assignment $(z,L)\mapsto T$ is a bijection onto the coordinate three-subsets: given $T=\{t_{1},t_{2},t_{3}\}$, the only candidate is $z=t_{1}\oplus t_{2}\oplus t_{3}$, which does not lie in $T$, the four-set $T\cup\{z\}$ has vanishing XOR and hence is a plane, and $L=\mathrm{dir}(T\cup\{z\})$ is forced. The counts match, $d\,g_{2}=\binom d3$, so the matrix of $q^{-1/2}D$ is square with distinct signed basis columns: a signed permutation, whence both identities in \eqref{eq:annU}.
\end{proof}

\begin{corollary}[Uniform occupation]
\label{cor:occup}
For $\psi,\varphi\in W_{\mathrm F}$ and $u,v\in\mathbb C^{d}$:
$\langle a_{u}\psi,a_{v}\varphi\rangle=q\,\langle v,u\rangle\langle\psi,\varphi\rangle$; equivalently $\PiF a_{u}^{\dagger}a_{v}\PiF=q\langle v,u\rangle\PiF$.
\end{corollary}

\begin{proof}
$a_{u}\psi=D(\bar u\otimes\psi)$ and $D^{\dagger}D=q\,\mathbb 1$.
\end{proof}

\emph{(iii) The two-hole layer is exchangeable.} For $\alpha\in\Lambda^{2}\mathbb C^{d}$ define the double contraction restricted to $W_{\mathrm F}$, $b_{\alpha}:=\iota_{\alpha}|_{W_{\mathrm F}}:W_{\mathrm F}\to\Lambda^{2}\mathbb C^{d}$, linear in $\alpha$.

\begin{theorem}[Pair exchange]
\label{thm:pairex}
For all $\alpha,\beta\in\Lambda^{2}\mathbb C^{d}$,
\begin{equation}
 b_{\alpha}^{\dagger}b_{\beta}=b_{\bar\beta}^{\dagger}b_{\bar\alpha}
 \qquad\text{as operators on }W_{\mathrm F} .
 \label{eq:pairex}
\end{equation}
\end{theorem}

For a generic subspace of four-forms no such exchange holds. The proof (Appendix~\ref{app:pairexproof}) is local to a single three-space: a nonvanishing matrix element forces the two contracted edges and their common complementary edge to be, together with a fourth edge, the four parallel edges of one three-space, and the residual wedge signs cancel by an \emph{even-shuffle} lemma---coordinate order restricted to a three-space is itself affine, so splits into parallel planes shuffle evenly (Lemma~\ref{lem:evenshuffle}). The conjugations in Eq.~\eqref{eq:pairex} are load-bearing: both sides are antilinear in $\alpha$ and linear in $\beta$, the unconjugated exchange has mismatched sesquilinearity and fails at complex arguments, and every frame-summed identity below is arranged so that only the conjugated form enters.

\emph{(iv) The two-contraction channels are a compressed swap and an eight-form.} With the creation map
\begin{equation}
 C:\mathbb C^{d}\otimes W_{\mathrm F}\to\Lambda^{5}\mathbb C^{d},\qquad
 C(v\otimes\psi)=v\wedge\psi ,
 \label{eq:creationdef}
\end{equation}
two composite channels of $W_{\mathrm F}$ are exactly solvable. First, the Hermitian exchange Gram $\mathcal X_{x}$ of the two-hole maps $c_{x,\alpha}=\iota_{\bar\alpha}\iota_{\bar x}|_{W_{\mathrm F}}$ [Eq.~\eqref{eq:Xxdef}] is $q$ times a self-adjoint \emph{compression of the literal swap} of two factors $W_{\mathrm F}$, $\mathcal X_{x}=q\,\mathsf S_{x}$ (Proposition~\ref{prop:swapid}); its spectrum is therefore controlled for free, $-q\,\mathbb 1\preceq\mathcal X_{x}\preceq q\,\mathbb 1$. Second, the fully contracted two-hole channel $\mathscr R:=DC^{\dagger}=\sum_{z}a_{z}\PiF a_{z}$ is the flattening of a canonical \emph{reference eight-form}:
\begin{equation}
 \mathscr R=q\,K_{\Theta},\qquad
 \Theta_{n}:=\!\!\sum_{\Delta\ \text{three-space}}\!\! e_{\Delta}\in\Lambda^{8}\mathbb C^{d},
 \label{eq:eightform}
\end{equation}
where $K_{\Theta}:\Lambda^{5}\mathbb C^{d}\to\Lambda^{3}\mathbb C^{d}$ has matrix elements $\langle e_{T},K_{\Theta}e_{S}\rangle:=\langle e_{S}\wedge e_{T},\Theta_{n}\rangle$ on real coordinate wedges (Proposition~\ref{prop:eightform}). Because $\Theta_{n}$ is a form, exchanging one particle against one hole across $\mathscr R$ is antisymmetric, and this yields a particle--hole \emph{transport identity} (Theorem~\ref{thm:transport}); transport is what decouples the occupancy amplitudes inside the five-plane bound below.

\emph{(v) The construction is self-similar across affine subspaces.} The same construction inside any $k$-dimensional affine subspace $H\subseteq\Ftwo^{n}$---one all-plus unit vector per direction $L\le\mathrm{dir}H$, supported on the planes of direction $L$ inside $H$---defines a local projector $\Pi_{H}$ with parameter $q_{H}=4/2^{k}$ [Eq.~\eqref{eq:localflat}], and facts (i)--(ii) hold verbatim inside $H$. Moreover, the subspace built inside $H$ resolves \emph{exactly} into those built inside its affine hyperplanes $H''$: for $k\ge4$,
\begin{equation}
 q_{H}\,\mathbb 1-\Pi_{H}
 =\frac{4}{2^{k}-8}\sum_{H''}
 \Bigl(\,2q_{H}\,P^{(4)}_{H''}-\Pi_{H''}\Bigr),
 \label{eq:resfinal}
\end{equation}
with $P^{(4)}_{H''}$ the projector onto $\Lambda^{4}\mathbb C^{H''}$ (Proposition~\ref{prop:hyperres}). Both sides are defect operators of the same shape at successive scales, so positivity propagates down to an eight-point base case; this recursion is the entire proof of the rank-seven bound (Appendix~\ref{ssec:ocseven}).

\subsection{Two operator inequalities}
\label{ssec:caps}

Two of the matching facets of Appendix~\ref{ssec:matching} require operator norm bounds on maps built from $W_{\mathrm F}$; the third needs only the unitarity of fact (ii). The proofs of the two inequalities (Appendices~\ref{app:twolinkproof} and~\ref{app:pcfiveproof}) are canonical anticommutation bookkeeping on top of the dictionary: the subspace $W_{\mathrm F}$ enters only through facts (ii)--(iv).

\begin{theorem}[Two-link bound]
\label{thm:twolink}
For every $d=2^{n}$, $n\ge2$, and every unit $x\in\mathbb C^{d}$, the two-link map
\begin{equation}
 \widetilde T_{x}:\Lambda^{2}x^{\perp}\otimes W_{\mathrm F}\to\Lambda^{5}x^{\perp},\quad
 \widetilde T_{x}(\alpha\otimes\psi)=\alpha\wedge a_{x}\psi ,
 \label{eq:twolinkdef}
\end{equation}
---two free wedge legs linked by one contraction of the $W_{\mathrm F}$ factor---admits the positive-semidefinite certificate
\begin{equation}
 2q\,\mathbb 1-\widetilde T_{x}^{\dagger}\widetilde T_{x}
 =\mathcal H_{x}^{\dagger}\mathcal H_{x}+q\,(\mathbb 1-\mathsf S_{x})\ \succeq\ 0 ,
 \label{eq:twolink}
\end{equation}
where $\mathcal H_{x}$ is the holomorphic hopping map and $\mathsf S_{x}$ the compressed swap [Eqs.~\eqref{eq:hopdef} and~\eqref{eq:swapdef}]; hence $\|\widetilde T_{x}\|^{2}\le2q$ and also $\widetilde T_{x}\widetilde T_{x}^{\dagger}\preceq2q\,\mathbb 1$.
\end{theorem}

The certificate is an exact three-layer CAR expansion of $\|\widetilde T_{x}\Psi\|^{2}$ (Proposition~\ref{prop:carid}): the zero-contraction layer is $q\,\mathbb 1$ by uniform occupation; pair exchange converts the one-contraction cross-Gram into the perfect square $-\mathcal H_{x}^{\dagger}\mathcal H_{x}$; and the two-contraction layer is exactly the compressed swap of fact (iv), $\mathcal X_{x}=q\,\mathsf S_{x}\preceq q\,\mathbb 1$.

\begin{theorem}[Five-plane operator bound]
\label{thm:pcfive}
For every $d=2^{n}$, $n\ge3$, and all orthonormal $x,y\in\mathbb C^{d}$, with $R:=(\mathbb Cx\oplus\mathbb Cy)^{\perp}$, the map $B_{xy}:=a_{y}a_{x}C:\mathbb C^{d}\otimes W_{\mathrm F}\to\Lambda^{3}R$ obeys
\begin{equation}
 B_{xy}B_{xy}^{\dagger}\ \preceq\ q\,\mathbb 1 ,
 \qquad\text{equivalently}\qquad
 C^{\dagger}n_{x}n_{y}C\ \preceq\ q\,\mathbb 1 .
 \label{eq:pcfive}
\end{equation}
\end{theorem}

The proof (Appendix~\ref{app:pcfiveproof}) decomposes $W_{\mathrm F}$ by $x,y$ occupancy, $\psi=\psi_{0}+x\wedge A_{x}\psi+y\wedge A_{y}\psi+x\wedge y\wedge A_{2}\psi$, and shows that transport decouples the double-hole block from the single-hole ones, $\mathsf MA_{x}=A_{x}\mathsf Q$ with $\mathsf Q=A_{2}^{\dagger}A_{2}$ and $\mathsf M$ the frame-summed double-hole Gram. The subspace $E=\mathrm{ran}\,A_{x}\oplus\mathrm{ran}\,A_{y}$ is then reducing with $B_{xy}B_{xy}^{\dagger}|_{E}=q\,\mathbb 1$ exactly---so the constant is sharp---while on $E^{\perp}$ an occupation bound from pair exchange controls the remainder.

\begin{corollary}[Five-plane bound]
\label{cor:fiveplane}
Let $Y\subset\mathbb C^{d}$ be five-dimensional with orthonormal basis $v_{1},\dots,v_{5}$ and unit volume form $\Omega_{Y}=v_{1}\wedge\cdots\wedge v_{5}$. Then
\begin{equation}
 \sum_{j=1}^{5}\bigl\|\PiF\bigl(v_{1}\wedge\cdots\widehat{v_{j}}\cdots\wedge v_{5}\bigr)\bigr\|^{2}
 =\|C^{\dagger}\Omega_{Y}\|^{2}\ \le\ q ,
 \label{eq:fiveplane}
\end{equation}
the hat denoting omission.
\end{corollary}

\subsection{Maps used in the proofs}
\label{app:flatmaps}

The following subsections prove the deferred statements of Appendices~\ref{ssec:code} and~\ref{ssec:caps}: the even-shuffle lemma and pair exchange (Theorem~\ref{thm:pairex}), the compressed swap and the reference eight-form of fact (iv), the two operator inequalities (Theorems~\ref{thm:twolink} and~\ref{thm:pcfive}), and the wedge-level identification behind the blossom row of Proposition~\ref{prop:rows}. Throughout, the only inputs from $W_{\mathrm F}$ are Theorem~\ref{thm:annU} and---where explicitly invoked---Theorem~\ref{thm:pairex}; every other step is generic exterior-algebra bookkeeping. For orientation, the maps used below and their types:
\begin{center}
\begin{tabular}{lll}
 map & domain & codomain\\
 \hline\\[-7pt]
 $D$ & $\mathbb C^{d}\otimes W_{\mathrm F}$ & $\Lambda^{3}\mathbb C^{d}$\\
 $C$ & $\mathbb C^{d}\otimes W_{\mathrm F}$ & $\Lambda^{5}\mathbb C^{d}$\\
 $\widetilde T_{x}$ & $\Lambda^{2}x^{\perp}\otimes W_{\mathrm F}$ & $\Lambda^{5}x^{\perp}$\\
 $\mathcal H_{x}$ & $\Lambda^{2}x^{\perp}\otimes W_{\mathrm F}$ & $\mathbb C^{d-1}\otimes\Lambda^{2}\mathbb C^{d}$\\
 $c_{x,\alpha}$,\ $\mathsf L$ & $W_{\mathrm F}$ & $\mathbb C^{d}$,\ $\Lambda^{2}\mathbb C^{d}$\\
 $B_{xy}$ & $\mathbb C^{d}\otimes W_{\mathrm F}$ & $\Lambda^{3}R$\\
 $C_{2}$,\ $\mathcal G_{x}$ & $\bigoplus_{a}W_{\mathrm F}$ & $\Lambda^{3}R$,\ $\Lambda^{2}\mathbb C^{d}$\\
 $\mathsf Q$,\ $\mathsf M$ & $W_{\mathrm F}$,\ $\Lambda^{3}R$ & $W_{\mathrm F}$,\ $\Lambda^{3}R$\\
\end{tabular}
\end{center}
Here $R:=(\mathbb Cx\oplus\mathbb Cy)^{\perp}$ as in the five-plane block, whose frame indexes the tuples of $\bigoplus_{a}W_{\mathrm F}$.

\subsection{Conventions and mode resolutions}
\label{app:flatconv}

In addition to the conventions of Appendix~\ref{ssec:code}: $(\iota_{v})^{\dagger}=\bar v\wedge\cdot$, so that $\{a_{u},a_{v}^{\dagger}\}=\langle u,v\rangle\mathbb 1$ and $\{a_{u},a_{v}\}=0$; and for a unit $x$, $\Lambda^{k}\mathbb C^{d}=(x\wedge\Lambda^{k-1}x^{\perp})\oplus\Lambda^{k}x^{\perp}$ orthogonally, $a_{x}(x\wedge\sigma+\rho)=\sigma$ for $\sigma\in\Lambda^{k-1}x^{\perp}$, $\rho\in\Lambda^{k}x^{\perp}$, and $x\wedge\cdot$ is an isometry on $\Lambda x^{\perp}$.

\begin{lemma}[Adjoint mode resolutions]
\label{lem:adjres}
For $\tau\in\Lambda^{3}\mathbb C^{d}$ and $\xi\in\Lambda^{5}\mathbb C^{d}$,
\begin{equation}
 D^{\dagger}\tau=\sum_{z}e_{z}\otimes\PiF(e_{z}\wedge\tau),\qquad
 C^{\dagger}\xi=\sum_{z}e_{z}\otimes\PiF(a_{z}\xi),
 \label{eq:adjres}
\end{equation}
with $a_{z}:=\iota_{e_{z}}$. Consequently $\mathscr R=DC^{\dagger}=\sum_{z}a_{z}\PiF a_{z}$.
\end{lemma}

\begin{proof}
$\langle v\otimes\psi,D^{\dagger}\tau\rangle
=\langle\iota_{v}\psi,\tau\rangle
=\langle\psi,\bar v\wedge\tau\rangle
=\sum_{z}\overline{v_{z}}\langle\psi,e_{z}\wedge\tau\rangle$ for $\psi\in W_{\mathrm F}$, and analogously $\langle v\otimes\psi,C^{\dagger}\xi\rangle=\langle v\wedge\psi,\xi\rangle=\langle\psi,a_{v}\xi\rangle=\sum_{z}\overline{v_{z}}\langle\psi,a_{z}\xi\rangle$. Applying $D$ to the second resolution gives $\mathscr R$.
\end{proof}

\begin{lemma}[Basis independence of mode sums]
\label{lem:modesum}
For any orthonormal basis $\{f_{j}\}$ of $\mathbb C^{d}$ and any operator $X$ on the exterior algebra,
\begin{equation}
\begin{gathered}
 \sum_{j}\iota_{f_{j}}\,X\,a_{f_{j}}=\sum_{z}a_{z}\,X\,a_{z},
 \qquad
 \sum_{j}a_{f_{j}}^{\dagger}\,X\,a_{f_{j}}=\sum_{z}a_{z}^{\dagger}\,X\,a_{z} ,\\
 \sum_{j}a_{f_{j}}\,X\,a_{f_{j}}^{\dagger}=\sum_{z}a_{z}\,X\,a_{z}^{\dagger} .
\end{gathered}
\end{equation}
\end{lemma}

\begin{proof}
In all three sums the coefficient of the corresponding coordinate term is $\sum_{j}(f_{j})_{i}\overline{(f_{j})_{k}}=\delta_{ik}$ by unitarity. (The same-variance sums $\sum_{j}\iota_{f_{j}}X\iota_{f_{j}}$ are \emph{not} basis-independent; every frame-summed identity in this appendix is of one of the three displayed types, and this is load-bearing at complex frames.)
\end{proof}

\subsection{Even shuffle and pair exchange}
\label{app:pairexproof}

\begin{lemma}[Even shuffle]
\label{lem:evenshuffle}
Let $\Delta$ be a three-space, partitioned into the two parallel planes $G\sqcup G'$ of direction $L\le\mathrm{dir}\Delta$. Then $e_{G}\wedge e_{G'}=+\,e_{\Delta}$.
\end{lemma}

\begin{proof}
First, coordinate order on $\Delta$ is affine: choose a basis $v_{1},v_{2},v_{3}$ of $\mathrm{dir}\Delta$ with distinct leading bits $j_{1}>j_{2}>j_{3}$, each $v_{k}$ vanishing at the other two leading bits, and the point $a\in\Delta$ with vanishing bits at $j_{1},j_{2},j_{3}$. Two points with coefficient vectors $c\ne c'$ first differ at bit $j_{k}$ for the minimal $k$ with $c_{k}\ne c_{k}'$, where the bit value is $c_{k}$; hence integer order on $\Delta$ is lexicographic in $(c_{1},c_{2},c_{3})$ and the position of a point, $\mathrm{pos}_{\Delta}=4c_{1}+2c_{2}+c_{3}$, has affine binary digits. Now the shuffle sign of moving sorted $G$ in front of sorted $G'$ is $(-1)^{N}$ with $N=\sum_{g\in G}\mathrm{pos}_{\Delta}(g)-(0{+}1{+}2{+}3)$. Modulo $2$ only $\sum_{g}c_{3}(g)$ survives, and an affine $\Ftwo$-valued function restricted to the plane $G$ is constant or balanced, so this sum is even; with $0{+}1{+}2{+}3=6$ even, $N$ is even.
\end{proof}

\begin{proof}[Proof of Theorem~\ref{thm:pairex}]
By sesquilinearity it suffices to prove $b_{P}^{\dagger}b_{Q}=b_{Q}^{\dagger}b_{P}$ for coordinate edges $P=\{p_{1}<p_{2}\}$, $Q$; the general case then follows from
$b_{\alpha}^{\dagger}b_{\beta}=\sum_{P,Q}\overline{\alpha_{P}}\beta_{Q}\,b_{P}^{\dagger}b_{Q}
=\sum_{P,Q}\overline{\alpha_{P}}\beta_{Q}\,b_{Q}^{\dagger}b_{P}=b_{\bar\beta}^{\dagger}b_{\bar\alpha}$. Note first that $b_{P}\psi_{L}=0$ unless $w_{P}:=p_{1}\oplus p_{2}\in L$, in which case exactly one plane survives: with $F=P\sqcup\mathcal E$ the plane of direction $L$ containing $P$ (the complementary edge $\mathcal E$ is parallel to $P$, since $F$ has vanishing XOR),
$b_{P}\psi_{L}=\sqrt q\,\eta(P,\mathcal E)\,e_{\mathcal E}$, where $\eta(A,B)$ denotes the shuffle sign $e_{A}\wedge e_{B}=\eta(A,B)e_{A\sqcup B}$; the contraction itself leaves no extra sign by the convention $\iota_{e_{P}}=\iota_{e_{p_{2}}}\iota_{e_{p_{1}}}$, which strips leading factors.

Now compare $\langle b_{P}\psi_{L},b_{Q}\psi_{M}\rangle$ with $\langle b_{Q}\psi_{L},b_{P}\psi_{M}\rangle$; for $P=Q$ they coincide trivially, so assume $P\ne Q$. A nonzero left side requires a common complementary edge $\mathcal E$, parallel to both $P$ and $Q$; then $w_{P}=w_{Q}=:w$, the edges $P,Q,\mathcal E$ are distinct parallel lines of direction $\{0,w\}$, and their affine span is a three-space $\Delta$ partitioned into $P,Q,\mathcal E$ and a fourth parallel edge $\mathcal K$. The plane of direction $L$ containing $Q$ is $\Delta\setminus(P\sqcup\mathcal E)=Q\sqcup\mathcal K$, and the plane of direction $M$ containing $P$ is $P\sqcup\mathcal K$; the same configuration with $\mathcal K$ in place of $\mathcal E$ shows that the two sides vanish together. When nonzero,
\begin{align}
 \langle b_{P}\psi_{L},b_{Q}\psi_{M}\rangle&=q\,\eta(P,\mathcal E)\eta(Q,\mathcal E),\notag\\
 \langle b_{Q}\psi_{L},b_{P}\psi_{M}\rangle&=q\,\eta(Q,\mathcal K)\eta(P,\mathcal K).
\end{align}
Lemma~\ref{lem:evenshuffle} applied to the two parallel splits of $\Delta$ gives
$e_{P\sqcup\mathcal E}\wedge e_{Q\sqcup\mathcal K}=e_{\Delta}=e_{Q\sqcup\mathcal E}\wedge e_{P\sqcup\mathcal K}$; substituting $e_{P\sqcup\mathcal E}=\eta(P,\mathcal E)e_{P}\wedge e_{\mathcal E}$ etc.\ and commuting the disjoint even-degree blocks yields
$\eta(P,\mathcal E)\eta(Q,\mathcal K)=\eta(Q,\mathcal E)\eta(P,\mathcal K)$, i.e., $\eta(P,\mathcal E)\eta(Q,\mathcal E)=\eta(P,\mathcal K)\eta(Q,\mathcal K)$ after multiplying by two squared signs. The two matrix elements coincide.
\end{proof}

\subsection{The two-link bound}
\label{app:twolinkproof}

Fix a unit $x\in\mathbb C^{d}$, put $R=x^{\perp}$ with orthonormal basis $\{r_{a}\}_{a=1}^{d-1}$, and write $a_{a}:=a_{r_{a}}$. Elements of $\Lambda^{2}R\otimes W_{\mathrm F}$ are written $\Psi=\sum_{a<b}r_{a}\wedge r_{b}\otimes\psi_{ab}$ with $\psi_{ba}=-\psi_{ab}$, and the two-link map \eqref{eq:twolinkdef} lands in $\Lambda^{5}R$ since $a_{x}\psi\in\Lambda^{3}x^{\perp}$. Define the \emph{holomorphic hopping map} and the \emph{two-hole maps}
\begin{align}
 (\mathcal H_{x}\Psi)_{a}&=\sum_{b}\iota_{r_{b}}\iota_{x}\,\psi_{ab},\notag\\
 c_{x,\alpha}&=\iota_{\bar\alpha}\,\iota_{\bar x}\big|_{W_{\mathrm F}}:W_{\mathrm F}\to\mathbb C^{d}
 \quad(\text{antilinear in }\alpha),
 \label{eq:hopdef}
\end{align}
and the Hermitian operator $\mathcal X_{x}$ on $\Lambda^{2}R\otimes W_{\mathrm F}$ by its matrix elements on frame edges,
\begin{equation}
 \langle\alpha\otimes\psi,\ \mathcal X_{x}(\beta\otimes\varphi)\rangle
 =\langle c_{x,\beta}\psi,\ c_{x,\alpha}\varphi\rangle .
 \label{eq:Xxdef}
\end{equation}

\begin{proposition}[Exact CAR expansion]
\label{prop:carid}
On $\Lambda^{2}R\otimes W_{\mathrm F}$,
\begin{equation}
 \widetilde T_{x}^{\dagger}\widetilde T_{x}
 =q\,\mathbb 1-\mathcal H_{x}^{\dagger}\mathcal H_{x}+\mathcal X_{x} .
 \label{eq:carid}
\end{equation}
\end{proposition}

\begin{proof}
All four operators are Hermitian, so it suffices to match quadratic forms. With $\varphi_{ab}:=a_{x}\psi_{ab}$, the representation $\widetilde T_{x}\Psi=\tfrac12\sum_{a,b}a_{a}^{\dagger}a_{b}^{\dagger}\varphi_{ab}$ and three anticommutations give the normal-ordered kernel
\begin{align}
 a_{b}a_{a}a_{c}^{\dagger}a_{d}^{\dagger}
 ={}&\delta_{ac}\delta_{bd}-\delta_{ad}\delta_{bc}
 -\delta_{ac}a_{d}^{\dagger}a_{b}+\delta_{bc}a_{d}^{\dagger}a_{a}\notag\\
 &+\delta_{ad}a_{c}^{\dagger}a_{b}-\delta_{bd}a_{c}^{\dagger}a_{a}
 +a_{c}^{\dagger}a_{d}^{\dagger}a_{b}a_{a} ,
\end{align}
whose three groups are evaluated separately in $\|\widetilde T_{x}\Psi\|^{2}=\tfrac14\sum_{a,b,c,d}\langle\varphi_{ab},a_{b}a_{a}a_{c}^{\dagger}a_{d}^{\dagger}\varphi_{cd}\rangle$.

\emph{Zero contractions.} By antisymmetry and Corollary~\ref{cor:occup}, the two delta products give $\tfrac12\sum_{a,b}\|a_{x}\psi_{ab}\|^{2}=q\|\Psi\|^{2}$.

\emph{One contraction.} Set $S:=\sum_{a,b,c}\langle a_{b}\varphi_{ac},a_{c}\varphi_{ab}\rangle$. Each of the four single-delta terms contributes $-\tfrac14 S$: the first directly after relabeling, the second and third after one relabeling that costs one antisymmetry sign, and the fourth after two relabelings whose signs cancel; the total is $-S$. On $W_{\mathrm F}$ the composite $a_{b}a_{x}=\iota_{\bar r_{b}}\iota_{\bar x}$ equals $b_{\bar x\wedge\bar r_{b}}$, and Theorem~\ref{thm:pairex} with $\alpha=\bar x\wedge\bar r_{b}$, $\beta=\bar x\wedge\bar r_{c}$ converts the cross-Gram into a square:
\begin{equation}
 \langle a_{b}a_{x}\psi_{ac},a_{c}a_{x}\psi_{ab}\rangle
 =\langle \iota_{r_{c}}\iota_{x}\psi_{ac},\ \iota_{r_{b}}\iota_{x}\psi_{ab}\rangle ,
\end{equation}
so that summing over $b,c$ factorizes $S=\sum_{a}\|(\mathcal H_{x}\Psi)_{a}\|^{2}=\|\mathcal H_{x}\Psi\|^{2}$.

\emph{Two contractions.} The remaining term is $\tfrac14\sum_{a,b,c,d}\langle a_{d}a_{c}\varphi_{ab},a_{b}a_{a}\varphi_{cd}\rangle$; since $a_{b}a_{a}\varphi_{cd}=c_{x,r_{a}\wedge r_{b}}\psi_{cd}$ and the summand is invariant under $a\leftrightarrow b$ and $c\leftrightarrow d$, it equals the sum over unordered frame edges, which is $\langle\Psi,\mathcal X_{x}\Psi\rangle$ by Eq.~\eqref{eq:Xxdef}.
\end{proof}

\begin{proposition}[Compressed swap]
\label{prop:swapid}
Let $J_{x}\alpha=x\wedge\alpha$ (an isometry $\Lambda^{2}R\to\Lambda^{3}\mathbb C^{d}$), let
\begin{equation}
 \mathcal V_{x}:=q^{-1/2}D^{\dagger}J_{x}:\Lambda^{2}R\to\mathbb C^{d}\otimes W_{\mathrm F} ,
\end{equation}
let $\Sigma_{W_{\mathrm F}}$ be the literal swap of two factors $W_{\mathrm F}$, and set
\begin{equation}
 \mathsf S_{x}:=(\mathcal V_{x}^{\dagger}\otimes\mathbb 1)
 (\mathbb 1\otimes\Sigma_{W_{\mathrm F}})(\mathcal V_{x}\otimes\mathbb 1)
 \ \ \text{on }\Lambda^{2}R\otimes W_{\mathrm F} .
 \label{eq:swapdef}
\end{equation}
Then $\mathcal V_{x}^{\dagger}\mathcal V_{x}=\mathbb 1$, so $\mathsf S_{x}$ is a self-adjoint compression of a unitary with $-\mathbb 1\preceq\mathsf S_{x}\preceq\mathbb 1$, and
\begin{equation}
 \mathcal X_{x}=q\,\mathsf S_{x} .
 \label{eq:swapid}
\end{equation}
\end{proposition}

\begin{proof}
Isometry: $\mathcal V_{x}^{\dagger}\mathcal V_{x}=q^{-1}J_{x}^{\dagger}DD^{\dagger}J_{x}=J_{x}^{\dagger}J_{x}=\mathbb 1$ by Eq.~\eqref{eq:annU}; here the surjectivity half, $DD^{\dagger}=q\mathbb 1$ on all of $\Lambda^{3}\mathbb C^{d}$, is what enters. By Lemma~\ref{lem:adjres}, $\mathcal V_{x}\alpha=q^{-1/2}\sum_{z}e_{z}\otimes p_{z}(\alpha)$ with $p_{z}(\alpha):=\PiF(e_{z}\wedge x\wedge\alpha)$, so
\begin{equation}
 q\,\langle\alpha\otimes\psi,\mathsf S_{x}(\beta\otimes\varphi)\rangle
 =\sum_{z}\langle p_{z}(\alpha),\varphi\rangle\langle\psi,p_{z}(\beta)\rangle .
\end{equation}  
The factors are evaluated by moving $e_{z}$ to the last position, which produces a sign because $x\wedge\alpha$ is a three-form, and by taking the adjoint of the contraction: for $\varphi\in W_{\mathrm F}$,
$\langle p_{z}(\alpha),\varphi\rangle=\langle e_{z}\wedge x\wedge\alpha,\varphi\rangle =-\langle(x\wedge\alpha)\wedge e_{z},\varphi\rangle =-\langle e_{z},\iota_{\bar\alpha}\iota_{\bar x}\varphi\rangle =-\langle e_{z},c_{x,\alpha}\varphi\rangle$, using $(\iota_{\gamma})^{\dagger}=\bar\gamma\wedge\cdot$ for the three-form $\gamma=\bar x\wedge\bar\alpha$. The two signs cancel in the product, and completeness of $\{e_{z}\}$ gives $q\langle\alpha\otimes\psi,\mathsf S_{x}(\beta\otimes\varphi)\rangle =\langle c_{x,\beta}\psi,c_{x,\alpha}\varphi\rangle$, which is Eq.~\eqref{eq:Xxdef}.
\end{proof}

\begin{proof}[Proof of Theorem~\ref{thm:twolink}]
Substitute Eq.~\eqref{eq:swapid} into Eq.~\eqref{eq:carid}: this rearranges into Eq.~\eqref{eq:twolink}, whose two summands are positive semidefinite ($\mathcal H_{x}^{\dagger}\mathcal H_{x}\succeq0$ and $\mathbb 1-\mathsf S_{x}\succeq0$). The final claim holds because $AA^{\dagger}$ and $A^{\dagger}A$ share their nonzero spectrum.
\end{proof}

\subsection{The reference eight-form and transport}
\label{app:eightform}

\begin{proposition}[Reference-form identity]
\label{prop:eightform}
Eq.~\eqref{eq:eightform} holds for every $d=2^{n}$, $n\ge3$.
\end{proposition}

\begin{proof}
By definition, $\langle e_{T},K_{\Theta}e_{S}\rangle=\mathrm{sgn}(S,T)$ if $S\sqcup T$ is a three-space, where $e_{S}\wedge e_{T}=\mathrm{sgn}(S,T)\,e_{S\sqcup T}$ for disjoint $S,T$, and zero otherwise. We compute $\langle e_{T},\mathscr Re_{S}\rangle$ from $\mathscr R=\sum_{z}a_{z}\PiF a_{z}$ (Lemma~\ref{lem:adjres}). A nonzero $z$-term requires $z\in S$, the four-set $G:=S\setminus\{z\}$ to be a plane ($a_{z}e_{S}=(-1)^{p_{S}(z)-1}e_{G}$, with $p_{S}(z)$ the position of $z$ in $S$), and, after $\PiF e_{G}=q\sum_{F\parallel G}e_{F}$ and the outer contraction, $T=F\setminus\{z\}$ for the unique plane $F=C_{z}(\mathrm{dir}\,G)$. Then $F$ and $G$ are disjoint parallel planes, so $\Delta:=F\sqcup G=S\sqcup T$ is a three-space. Conversely, if $S\sqcup T=\Delta$ is a three-space, the only candidate is $z=\bigoplus_{t\in T}t$, and it works: $z\in S\setminus T$, $F=T\cup\{z\}$ is a plane, and $G=\Delta\setminus F=S\setminus\{z\}$ is the plane parallel to it. Each matrix element therefore receives exactly one contribution, of magnitude $q$. For its sign, the two contractions contribute $(-1)^{p_{S}(z)-1+p_{F}(z)-1}$; writing $e_{S}=(-1)^{p_{S}(z)-1}e_{z}\wedge e_{G}$ and $e_{F}=(-1)^{p_{F}(z)-1}e_{z}\wedge e_{T}$ and using Lemma~\ref{lem:evenshuffle} in the form $e_{G}\wedge e_{F}=e_{\Delta}$, one moves $e_{z}$ across the even-degree $e_{G}$ at no cost and finds $e_{S}\wedge e_{T}=(-1)^{p_{S}(z)-1+p_{F}(z)-1}e_{\Delta}$: the contribution sign is exactly $\mathrm{sgn}(S,T)$.
\end{proof}

\begin{proposition}[Particle--hole antisymmetry]
\label{prop:phanti}
For all coordinate modes $j,k$: $a_{j}\mathscr Ra_{k}^{\dagger}+a_{k}\mathscr Ra_{j}^{\dagger}=0$.
\end{proposition}

\begin{proof}
By Proposition~\ref{prop:eightform} it suffices for $K_{\Theta}$. On real wedges,
$\langle e_{B},a_{j}K_{\Theta}a_{k}^{\dagger}e_{A}\rangle
=\langle e_{k}\wedge e_{A}\wedge e_{j}\wedge e_{B},\Theta_{n}\rangle
=\langle e_{k}\wedge e_{j}\wedge e_{A}\wedge e_{B},\Theta_{n}\rangle$
for $|A|=4$, $|B|=2$, moving the one-form $e_{j}$ across the even-degree $e_{A}$; adding the $(j\leftrightarrow k)$ term produces the factor $e_{k}\wedge e_{j}+e_{j}\wedge e_{k}=0$.
\end{proof}

\begin{theorem}[Transport identity]
\label{thm:transport}
For all $u,v\in\mathbb C^{d}$, as operators on the exterior algebra (the only nonvanishing component maps $\Lambda^{3}\mathbb C^{d}$ to $\Lambda^{2}\mathbb C^{d}$),
\begin{equation}
 \sum_{z}a_{z}\,\iota_{u}\,\PiF\,a_{u}^{\dagger}a_{v}^{\dagger}\,a_{z}
 =\iota_{v}\iota_{u}\,\PiF\,a_{u}^{\dagger} ,
 \label{eq:transport}
\end{equation}
every $u$- and $v$-slot being \emph{linear} (bilinear contractions, wedge creations).
\end{theorem}

\begin{proof}
For coordinate modes $j,k,l$, two anticommutations normal-order the mode sum:
\begin{align}
 F(j,k,l):={}&\sum_{z}a_{z}a_{j}\PiF a_{k}^{\dagger}a_{l}^{\dagger}a_{z}\notag\\
 ={}&a_{l}a_{j}\PiF a_{k}^{\dagger}-a_{k}a_{j}\PiF a_{l}^{\dagger}
 -a_{j}\,\mathscr R\,a_{k}^{\dagger}a_{l}^{\dagger},
\end{align}
using $\sum_{z}a_{z}a_{j}\PiF a_{z}=-a_{j}\mathscr R$. Hence the defect $G_{l}(j,k):=F(j,k,l)-a_{l}a_{j}\PiF a_{k}^{\dagger}$ satisfies $G_{l}(j,k)+G_{l}(k,j)=0$ for every $l$: its first part is antisymmetric by $\{a_{j},a_{k}\}=0$ and its second by Proposition~\ref{prop:phanti}. Expanding Eq.~\eqref{eq:transport} in coordinates, the difference of the two sides is $\sum_{l}v_{l}\sum_{j,k}u_{j}u_{k}G_{l}(j,k)=0$, a symmetric coefficient array against an antisymmetric kernel. Polarization in the linear slots is legitimate precisely because the identity is bilinear in $u$; the Hermitian variant, with $a_{u}$ in place of $\iota_{u}$, has the non-symmetric coefficients $\overline{u_{j}}u_{k}$ and is false at complex $u$.
\end{proof}

\subsection{The five-plane bound}
\label{app:pcfiveproof}

Fix orthonormal $x,y\in\mathbb C^{d}$ and let $R:=(\mathbb Cx\oplus\mathbb Cy)^{\perp}$ with orthonormal basis $\{r_{a}\}_{a=1}^{d-2}$. Every $\psi\in W_{\mathrm F}$ decomposes orthogonally by $x,y$ occupancy,
\begin{equation}
 \psi=\psi_{0}+x\wedge A_{x}\psi+y\wedge A_{y}\psi+x\wedge y\wedge A_{2}\psi ,
 \label{eq:occdec}
\end{equation}
with $\psi_{0}\in\Lambda^{4}R$, $A_{x}\psi,A_{y}\psi\in\Lambda^{3}R$, $A_{2}\psi\in\Lambda^{2}R$; directly from the decomposition,
\begin{align}
 &a_{x}\psi=A_{x}\psi+y\wedge A_{2}\psi,\quad
 a_{y}\psi=A_{y}\psi-x\wedge A_{2}\psi,\notag\\
 &A_{2}=a_{y}a_{x}\big|_{W_{\mathrm F}},\qquad
 A_{x}^{\dagger}\theta=\PiF(x\wedge\theta),\ \ \text{etc.}
 \label{eq:amps6}
\end{align}
Define, on tuples $\Psi=(\psi_{a})_{a}$, $\psi_{a}\in W_{\mathrm F}$, and on $W_{\mathrm F}$,
\begin{align}
 &C_{2}\Psi=\sum_{a}r_{a}\wedge A_{2}\psi_{a},\qquad
 \mathcal G_{x}\Psi=\sum_{a}\iota_{r_{a}}\iota_{x}\,\psi_{a},\notag\\
 &\mathsf L:=\iota_{y}\iota_{x}\big|_{W_{\mathrm F}}=b_{x\wedge y},\quad
 \mathsf Q:=A_{2}^{\dagger}A_{2},\quad
 \mathsf M:=C_{2}C_{2}^{\dagger},
 \label{eq:decmaps}
\end{align}
and analogously $\mathcal G_{y}$. Four exchange identities decouple these objects.

\begin{lemma}[Pair-exchange intertwinings]
\label{lem:LQ6}
$\mathsf L^{\dagger}\mathsf L=\mathsf Q$, and $A_{x}^{\dagger}C_{2}=\mathsf L^{\dagger}\mathcal G_{x}$, $A_{y}^{\dagger}C_{2}=\mathsf L^{\dagger}\mathcal G_{y}$.
\end{lemma}

\begin{proof}
Theorem~\ref{thm:pairex} at $\alpha=\beta=x\wedge y$ gives $b_{x\wedge y}^{\dagger}b_{x\wedge y}=b_{\bar x\wedge\bar y}^{\dagger}b_{\bar x\wedge\bar y}$, and $b_{\bar x\wedge\bar y}=a_{y}a_{x}|_{W_{\mathrm F}}=A_{2}$ by Eq.~\eqref{eq:amps6}: the first identity. For the second, with $\varphi\in W_{\mathrm F}$,
\begin{align}
 \langle\varphi,A_{x}^{\dagger}C_{2}\Psi\rangle
 &=\sum_{a}\langle\varphi,\ x\wedge r_{a}\wedge A_{2}\psi_{a}\rangle\notag\\
 &=\sum_{a}\langle b_{\bar x\wedge\bar r_{a}}\varphi,\ b_{\bar x\wedge\bar y}\psi_{a}\rangle\notag\\
 &=\sum_{a}\langle b_{x\wedge y}\varphi,\ b_{x\wedge r_{a}}\psi_{a}\rangle
 =\langle\mathsf L\varphi,\mathcal G_{x}\Psi\rangle ,
\end{align}
by Theorem~\ref{thm:pairex} and $b_{x\wedge r_{a}}=\iota_{r_{a}}\iota_{x}$; the $y$-version is identical, with the \emph{same} $\mathsf L$.
\end{proof}

\begin{lemma}[Transport intertwinings]
\label{lem:HC2}
$\mathcal G_{x}C_{2}^{\dagger}=\mathsf LA_{x}^{\dagger}$ and $\mathcal G_{y}C_{2}^{\dagger}=\mathsf LA_{y}^{\dagger}$ on $\Lambda^{3}R$.
\end{lemma}

\begin{proof}
From Eq.~\eqref{eq:amps6}, $(C_{2}^{\dagger}\theta)_{a}=\PiF a_{x}^{\dagger}a_{y}^{\dagger}a_{r_{a}}\theta$, so $\mathcal G_{x}C_{2}^{\dagger}\theta=\sum_{a}\iota_{r_{a}}\iota_{x}\PiF a_{x}^{\dagger}a_{y}^{\dagger}a_{r_{a}}\theta$. Complete $\{r_{a}\}$ by $x$ and $y$: the added terms vanish on $\Lambda^{3}R$ (the $x$-term contains $\iota_{x}\iota_{x}=0$, and both end in $a_{x}\theta=0$ or $a_{y}\theta=0$), so by Lemma~\ref{lem:modesum} the sum equals the coordinate sum, and Theorem~\ref{thm:transport} with $(u,v)=(x,y)$ evaluates it as $\iota_{y}\iota_{x}\PiF a_{x}^{\dagger}\theta=\mathsf LA_{x}^{\dagger}\theta$. For the $y$-version, anticommute $a_{x}^{\dagger}a_{y}^{\dagger}=-a_{y}^{\dagger}a_{x}^{\dagger}$ and apply the theorem with $(u,v)=(y,x)$; the two sign flips cancel and the same $\mathsf L$ appears.
\end{proof}

\begin{proposition}[Decoupling]
\label{prop:decouple}
$\mathsf MA_{x}=A_{x}\mathsf Q$ and $\mathsf MA_{y}=A_{y}\mathsf Q$.
\end{proposition}

\begin{proof}
Adjoints of Lemmas~\ref{lem:LQ6} and~\ref{lem:HC2} give $C_{2}^{\dagger}A_{x}=\mathcal G_{x}^{\dagger}\mathsf L$ and $C_{2}\mathcal G_{x}^{\dagger}=A_{x}\mathsf L^{\dagger}$; hence $\mathsf MA_{x}=C_{2}(C_{2}^{\dagger}A_{x})=C_{2}\mathcal G_{x}^{\dagger}\mathsf L=A_{x}\mathsf L^{\dagger}\mathsf L=A_{x}\mathsf Q$.
\end{proof}

\begin{proposition}[Occupation bound]
\label{prop:occbound}
With $C_{R}\Psi:=\sum_{a}r_{a}\wedge\psi_{a}$,
\begin{equation}
 C_{R}^{\dagger}\,n_{x}\,C_{R}=q\,\mathbb 1-\mathcal G_{x}^{\dagger}\mathcal G_{x}\ \preceq\ q\,\mathbb 1 ,
\end{equation}
and consequently $\mathsf M\preceq q\,\mathbb 1$.
\end{proposition}

\begin{proof}
With $\varphi_{a}:=a_{x}\psi_{a}$, since $a_{x}(r_{a}\wedge\psi_{a})=-r_{a}\wedge\varphi_{a}$,
\begin{equation}
 \langle C_{R}\Psi,n_{x}C_{R}\Psi\rangle
 =\sum_{a,b}\langle\varphi_{a},(\delta_{ab}-a_{r_{b}}^{\dagger}a_{r_{a}})\varphi_{b}\rangle ,
\end{equation}
whose diagonal part is $q\|\Psi\|^{2}$ by Corollary~\ref{cor:occup} and whose cross-Gram, after the pair exchange $\langle b_{\bar x\wedge\bar r_{b}}\psi_{a},b_{\bar x\wedge\bar r_{a}}\psi_{b}\rangle=\langle b_{x\wedge r_{a}}\psi_{a},b_{x\wedge r_{b}}\psi_{b}\rangle$, factorizes into $\|\mathcal G_{x}\Psi\|^{2}$. (With free input of exterior degree one there is no two-contraction level, so no swap term appears.) For the consequence, two anticommutations give $C_{2}=a_{y}a_{x}C_{R}$, whence
\begin{equation}
 C_{2}^{\dagger}C_{2}=C_{R}^{\dagger}a_{x}^{\dagger}a_{y}^{\dagger}a_{y}a_{x}C_{R}
 \preceq C_{R}^{\dagger}n_{x}C_{R}\preceq q\,\mathbb 1 ,
\end{equation}
and $\mathsf M=C_{2}C_{2}^{\dagger}$ shares the nonzero spectrum of $C_{2}^{\dagger}C_{2}$.
\end{proof}

\begin{proof}[Proof of Theorem~\ref{thm:pcfive}]
Inserting Eq.~\eqref{eq:occdec} into $B_{xy}(v\otimes\psi)=a_{y}a_{x}(v\wedge\psi)$ resolves $B_{xy}$ over the slots of $\mathbb C^{d}=\mathbb Cx\oplus\mathbb Cy\oplus R$ as $A_{y}$, $-A_{x}$, and $C_{2}$, respectively; the slots being orthogonal,
\begin{equation}
 B_{xy}B_{xy}^{\dagger}=A_{x}A_{x}^{\dagger}+A_{y}A_{y}^{\dagger}+\mathsf M .
 \label{eq:BBsplit}
\end{equation}
From the orthogonality of the components in Eq.~\eqref{eq:amps6} and Corollary~\ref{cor:occup},
\begin{equation}
 A_{x}^{\dagger}A_{x}=A_{y}^{\dagger}A_{y}=q\,\mathbb 1-\mathsf Q,
 \qquad A_{x}^{\dagger}A_{y}=0 .
 \label{eq:AxAx}
\end{equation}
Let $E:=\mathrm{ran}\,A_{x}\oplus\mathrm{ran}\,A_{y}$, an orthogonal sum. For $u\in W_{\mathrm F}$, Eqs.~\eqref{eq:BBsplit}--\eqref{eq:AxAx} and Proposition~\ref{prop:decouple} give
\begin{equation}
 B_{xy}B_{xy}^{\dagger}A_{x}u=A_{x}(q\,\mathbb 1-\mathsf Q)u+A_{x}\mathsf Qu=q\,A_{x}u ,
\end{equation}
and likewise for $A_{y}$: the subspace $E$ is invariant with $B_{xy}B_{xy}^{\dagger}|_{E}=q\,\mathbb 1$, hence reducing. On $E^{\perp}=\ker A_{x}^{\dagger}\cap\ker A_{y}^{\dagger}$ the first two terms of Eq.~\eqref{eq:BBsplit} vanish, $E^{\perp}$ is $\mathsf M$-invariant by Proposition~\ref{prop:decouple}, and $\mathsf M\preceq q\,\mathbb 1$ by Proposition~\ref{prop:occbound}. Both reducing blocks obey the bound. The equivalence follows from $n_{x}n_{y}=a_{x}^{\dagger}a_{y}^{\dagger}a_{y}a_{x}$ for $x\perp y$, so that $C^{\dagger}n_{x}n_{y}C=B_{xy}^{\dagger}B_{xy}$. The bound is attained whenever $E\ne0$, so the constant cannot be improved; $E\ne0$ is explicit at coordinate pairs: for $x=e_{0}$, $y=e_{1}$ and any direction $L$ with $1\notin L$, $A_{x}\psi_{L}=\pm\sqrt q\,e_{L\setminus\{0\}}\ne0$.
\end{proof}

\begin{proof}[Proof of Corollary~\ref{cor:fiveplane}]
By Lemma~\ref{lem:adjres}, $\|C^{\dagger}\Omega_{Y}\|^{2}=\sum_{z}\|\PiF(a_{z}\Omega_{Y})\|^{2}
=\langle\Omega_{Y},\sum_{z}a_{z}^{\dagger}\PiF a_{z}\Omega_{Y}\rangle$, a Hermitian-variance mode sum, basis-independent by Lemma~\ref{lem:modesum}; evaluating it in an orthonormal basis extending $\{v_{j}\}$, the completion terms annihilate $\Omega_{Y}$ and $a_{v_{j}}\Omega_{Y}=\iota_{\bar v_{j}}\Omega_{Y}=\pm\,v_{1}\wedge\cdots\widehat{v_{j}}\cdots\wedge v_{5}$, giving the first equality in Eq.~\eqref{eq:fiveplane}. For the bound, take $x=v_{1}$, $y=v_{2}$, $\eta=v_{3}\wedge v_{4}\wedge v_{5}$; then $\langle v\otimes\psi,C^{\dagger}\Omega_{Y}\rangle=\langle v\wedge\psi,x\wedge y\wedge\eta\rangle=\langle B_{xy}(v\otimes\psi),\eta\rangle$, so $C^{\dagger}\Omega_{Y}=B_{xy}^{\dagger}\eta$ and $\|C^{\dagger}\Omega_{Y}\|^{2}=\langle\eta,B_{xy}B_{xy}^{\dagger}\eta\rangle\le q$ by Theorem~\ref{thm:pcfive}.
\end{proof}

\subsection{The blossom identification}
\label{app:blossomproof}

\begin{lemma}[Blossom identification]
\label{lem:blossomid}
Let $u_{1},\dots,u_{6}$ be orthonormal, $r\in[6]$, and let $\Omega$ be the unit volume form of $\mathrm{span}\{u_{i}:i\ne r\}\subseteq u_{r}^{\perp}$. Then
\begin{equation}
 \bigl\|\widetilde T_{u_{r}}^{\dagger}\Omega\bigr\|^{2}
 =\sum_{\substack{A\in\binom{[6]}4\\ r\in A}}\|\PiF u_{A}\|^{2} .
\end{equation}
\end{lemma}

\begin{proof}
For $\alpha$ in the orthonormal basis of $\Lambda^{2}u_{r}^{\perp}$ induced by an orthonormal basis of $u_{r}^{\perp}$ containing the $u_{i}$, $i\ne r$, and $\psi\in W_{\mathrm F}$, wedge--contraction adjointness gives
\begin{equation}
 \langle\alpha\otimes\psi,\ \widetilde T_{u_{r}}^{\dagger}\Omega\rangle
 =\langle\alpha\wedge a_{u_{r}}\psi,\ \Omega\rangle
 =\langle\psi,\ u_{r}\wedge\iota_{\bar\alpha}\Omega\rangle ,
\end{equation}
so the $\alpha$-component of $\widetilde T_{u_{r}}^{\dagger}\Omega$ is $\PiF(u_{r}\wedge\iota_{\bar\alpha}\Omega)$. It vanishes unless $\alpha=u_{i}\wedge u_{j}$ with $i,j\ne r$, in which case $u_{r}\wedge\iota_{\bar\alpha}\Omega=\pm u_{\widehat{ij}}$, the wedge over $\widehat{ij}=[6]\setminus\{i,j\}\ni r$; summing the squared norms over the orthonormal $\alpha$ gives the claim.
\end{proof}

\subsection{The matching polytope: proof of Theorem~\ref{thm:cfsix}}
\label{ssec:matching}

We use the notation of Sec.~\ref{ssec:mechanism}: $N\succeq0$ has rank at most six, $u_{1},\dots,u_{6}$ are orthonormal eigenvectors of $N$ for the eigenvalues $t_{1}\ge\cdots\ge t_{6}\ge0$, extended arbitrarily if the rank is smaller, and $p_{A}=\|\PiF u_{A}\|^{2}$ and $t_{A}$ are as in Eq.~\eqref{eq:cFexp}.
Complementation $e\mapsto\widehat e:=[6]\setminus e$ is a bijection between the fifteen edges of the complete graph $K_{6}$ and the fifteen four-subsets of $[6]$, and the edge variables are
\begin{equation}
 x_{e}:=\frac{p_{\widehat e}}{q},\qquad e\in\tbinom{[6]}2 .
 \label{eq:edgevars}
\end{equation}
The three properties of $W_{\mathrm F}$ listed in Sec.~\ref{ssec:mechanism} give the three families of constraints of the following proposition: the degree constraint comes from the five-plane bound (Corollary~\ref{cor:fiveplane}), the triangle constraint from the one-hole unitarity (Theorem~\ref{thm:annU}), and the five-vertex blossom constraint from the two-link bound (Theorem~\ref{thm:twolink}).

\begin{proposition}[Matching rows]
\label{prop:rows}
For every $i\in[6]$, every triple $\{i,j,k\}\subset[6]$, and every $r\in[6]$,
\begin{equation}
 \sum_{j\ne i}x_{ij}\le1,\qquad
 x_{ij}+x_{ik}+x_{jk}\le1,\qquad
 \sum_{\substack{i'<j'\\ i',j'\ne r}}x_{i'j'}\le2 .
 \label{eq:matchrows}
\end{equation}
\end{proposition}

\begin{proof}
Each row is one inequality tested on a frame form. \emph{Degree}: the five four-subsets of $Y_{i}:=\mathrm{span}\{u_{j}:j\ne i\}$ drawn from the frame are exactly the $\widehat{ij}$, $j\ne i$, so Corollary~\ref{cor:fiveplane} applied to $Y_{i}$ gives $\sum_{j\ne i}p_{\widehat{ij}}\le q$. \emph{Triangle}: with $\mathcal R:=[6]\setminus\{i,j,k\}$ and $u_{\mathcal R}$ the corresponding unit three-form, Theorem~\ref{thm:annU} gives the exact three-body marginal
\begin{equation}
 q=\|D^{\dagger}u_{\mathcal R}\|^{2}
 =\sum_{z}\|\PiF(e_{z}\wedge u_{\mathcal R})\|^{2} ,
 \label{eq:threebody}
\end{equation}
a Hermitian-variance mode sum, evaluable in any orthonormal basis (Lemmas~\ref{lem:adjres} and~\ref{lem:modesum}); retaining only the completions $u_{i},u_{j},u_{k}$ and discarding the rest leaves $p_{\widehat{jk}}+p_{\widehat{ik}}+p_{\widehat{ij}}\le q$. \emph{Blossom} (the five-vertex odd-set row): with $x=u_{r}$ and $\Omega$ the unit volume form of $\mathrm{span}\{u_{i}:i\ne r\}$, the identification $\|\widetilde T_{u_{r}}^{\dagger}\Omega\|^{2}=\sum_{A\ni r}p_{A}$ (Lemma~\ref{lem:blossomid}) and Theorem~\ref{thm:twolink} give $\sum_{A\ni r}p_{A}\le2q$.
\end{proof}

\begin{proof}[Proof of Theorem~\ref{thm:cfsix}]
The vector $x=(x_{e})$ is nonnegative and, by Proposition~\ref{prop:rows}, satisfies the degree constraints $x(\delta(i))\le1$ and the odd-set constraints $x(E(U))\le(|U|-1)/2$ for every odd $U\subseteq[6]$: the nontrivial sizes are $|U|=3$ (triangle rows) and $|U|=5$ (blossom rows). By Edmonds' matching-polytope theorem~\cite{Edmonds1965matching}, $x$ is a convex combination of indicator vectors of matchings of $K_{6}$. For the nonnegative weights $w_{e}:=t_{\widehat e}$, Eqs.~\eqref{eq:cFexp} and~\eqref{eq:edgevars} give
\begin{equation}
 \cF(N)=q\sum_{e}w_{e}x_{e}\ \le\ q\max_{\mathcal M}\sum_{e\in\mathcal M}w_{e},
\end{equation}
the maximum over matchings, attained at a perfect matching since $w\ge0$. Assume first $t_{i}>0$ for all $i$ and set $a_{i}:=t_{i}^{-1}$, so that $w_{ij}=(\prod_{k}t_{k})\,a_{i}a_{j}$ with $a_{1}\le\cdots\le a_{6}$. If a perfect matching pairs $6$ with $i<5$ and $5$ with $j<5$, replacing those two edges by $\{5,6\}$ and $\{i,j\}$ changes the weight by $(\prod_{k}t_{k})(a_{5}-a_{i})(a_{6}-a_{j})\ge0$; repeating on $\{1,2,3,4\}$ shows that $\{12,34,56\}$ is optimal. Hence
\begin{equation}
 \cF(N)\le q\,(w_{12}+w_{34}+w_{56})
 =q\,(t_{3456}+t_{1256}+t_{1234}),
\end{equation}
which is Eq.~\eqref{eq:cf6}; vanishing $t_{i}$ follow by continuity. Finally, the right-hand side is the CB value: the four-subsets of the six-set $\{0,\dots,5\}\subset\Ftwo^{n}$ with vanishing XOR are exactly $\{0,1,2,3\}$, $\{0,1,4,5\}$, and $\{2,3,4,5\}$, each a plane with $\|\PiF e_{F}\|^{2}=q$ and all other coordinate $p_{A}$ vanishing, so $\cF(T)=q\,(t_{1}t_{2}t_{3}t_{4}+t_{1}t_{2}t_{5}t_{6}+t_{3}t_{4}t_{5}t_{6})$.
\end{proof}

\subsection{The rank-seven compression bound}
\label{ssec:ocseven}

The second sector inequality follows from an operator bound for compressions of the antisymmetric stabilizer projector to low-dimensional supports, which we now state.

\begin{theorem}[Rank-seven compression bound]
\label{thm:ocseven}
For every $d=2^{n}$, $n\ge3$, and every subspace $\mathcal T\subset\mathbb C^{d}$ with $\dim\mathcal T\le7$,
\begin{equation}
 P_{\Lambda^{4}\mathcal T}\,\PiF\,P_{\Lambda^{4}\mathcal T}\ \preceq\ q\,P_{\Lambda^{4}\mathcal T} .
\end{equation}
The constant is sharp: if $\mathcal T$ contains the coordinate directions of a plane $F$, then $\langle e_{F},\PiF e_{F}\rangle=q$.
\end{theorem}

The proof (Appendix~\ref{app:selfsim}) is an induction over affine subspaces, in the stronger map form: for every $k$-dimensional affine subspace $H$ ($3\le k\le n$) and every linear map $A:E\to\mathbb C^{H}$ with $\dim E\le7$,
\begin{equation}
 (\Lambda^{4}A)^{\dagger}\,\Pi_{H}\,(\Lambda^{4}A)\ \preceq\
 q_{H}\,(\Lambda^{4}A)^{\dagger}(\Lambda^{4}A) .
 \label{eq:mapcap}
\end{equation}
The resolution~\eqref{eq:resfinal}, congruenced by $\Lambda^{4}A$, expresses the defect at scale $k$ as a positive combination of defects at scale $k-1$; at the eight-point base a rank-seven map misses one direction of the three-space, and uniform occupation proves the bound there. Taking $k=n$ and $A$ the inclusion $\mathcal T\hookrightarrow\mathbb C^{d}$ gives Eq.~\eqref{eq:ocseven}. Notably, pair exchange never enters: the rank-seven bound is a one-hole-layer statement, while the matching facets above are two-hole-layer statements.

Theorem~\ref{thm:rsevencap}, including its equality case, follows from Theorem~\ref{thm:ocseven} by the argument given in Sec.~\ref{ssec:mechanism}.

\subsection{Self-similarity of the antisymmetric stabilizer subspace and the proof of Theorem~\ref{thm:ocseven}}
\label{app:selfsim}

This subsection proves Theorem~\ref{thm:ocseven} in the map form~\eqref{eq:mapcap}. Pair exchange never enters: the recursion uses only the one-hole layer of $W_{\mathrm F}$---Theorem~\ref{thm:annU} and uniform occupation, applied inside every affine subspace---together with two counting identities.

For a $k$-dimensional affine subspace $H\subseteq\Ftwo^{n}$ ($2\le k\le n$) write $\mathbb C^{H}:=\mathrm{span}\{e_{z}:z\in H\}$ and $q_{H}:=4/2^{k}$, and define the local basis vectors and projector
\begin{equation}
 \psi^{H}_{L}:=\sqrt{q_{H}}\!\!\sum_{F\subset H,\ F\parallel L}\!\!e_{F},
 \qquad
 \Pi_{H}:=\sum_{L\le\mathrm{dir}H}\ketbra{\psi^{H}_{L}}{\psi^{H}_{L}},
 \label{eq:localflat}
\end{equation}
the sums over the two-dimensional subspaces $L$ of $\mathrm{dir}\,H$ and, for each $L$, over the $2^{k-2}$ planes of direction $L$ inside $H$; all wedge coordinates carry the ambient orientation. For $H=\Ftwo^{n}$ this is $\PiF$ by Corollary~\ref{cor:codeform}.

\begin{lemma}[Local unitarity]
\label{lem:localU}
Theorem~\ref{thm:annU} holds verbatim for every local subspace $W_{\mathrm F}^{H}:=\mathrm{ran}\,\Pi_{H}$: the tensor $D_{H}(v\otimes\psi)=\iota_{v}\psi$ obeys $D_{H}^{\dagger}D_{H}=D_{H}D_{H}^{\dagger}=q_{H}\,\mathbb 1$, and consequently $\Pi_{H}a_{u}^{\dagger}a_{v}\Pi_{H}=q_{H}\langle v,u\rangle\Pi_{H}$ for $u,v\in\mathbb C^{H}$.
\end{lemma}

\begin{proof}
The bijection argument of Theorem~\ref{thm:annU} is intrinsic to the affine structure of $H$: for a three-subset $T\subset H$, the completion $z=t_{1}\oplus t_{2}\oplus t_{3}$ is an affine combination of an odd number of points of $H$ and therefore lies in $H$, and the count $2^{k}\dim W_{\mathrm F}^{H}=\binom{2^{k}}{3}$ is the same Gaussian-binomial identity at scale $k$.
\end{proof}

\begin{lemma}[Base case $k=3$]
\label{lem:basecap}
Eq.~\eqref{eq:mapcap} holds for every three-space $H$.
\end{lemma}

\begin{proof}
Here $\mathbb C^{H}\cong\mathbb C^{8}$ and $q_{H}=\tfrac12$. Since $\mathrm{rank}\,A\le7<8$, the range of $A$ lies in $S:=x^{\perp}\cap\mathbb C^{H}$ for some unit $x\in\mathbb C^{H}$. On the four-particle sector of $\mathbb C^{H}$ the projector onto $\Lambda^{4}S$ is $\mathbb 1-n_{x}$, and the local uniform-occupation identity (Lemma~\ref{lem:localU}) gives $\Pi_{H}n_{x}\Pi_{H}=\tfrac12\Pi_{H}$, whence
\begin{equation}
 \Pi_{H}\,P_{\Lambda^{4}S}\,\Pi_{H}=\tfrac12\,\Pi_{H}
 \ \Longrightarrow\
 P_{\Lambda^{4}S}\,\Pi_{H}\,P_{\Lambda^{4}S}\preceq\tfrac12\,P_{\Lambda^{4}S},
\end{equation}
the implication because the two compressed products are $A^{\dagger}A$ and $AA^{\dagger}$ for $A=\Pi_{H}P_{\Lambda^{4}S}$, hence share their nonzero spectrum. For a general map, write the polar decomposition $A=U|A|$ with $U$ a partial isometry into $S$; compressing the display above to $\mathrm{ran}\,\Lambda^{4}U$ and congruencing by $\Lambda^{4}|A|$ gives Eq.~\eqref{eq:mapcap}.
\end{proof}

\begin{proposition}[Hyperplane resolution]
\label{prop:hyperres}
Let $H$ be a $k$-dimensional affine subspace, $k\ge4$, let $m:=2^{k-2}$, let $H''$ range over the affine hyperplanes of $H$, and let $P^{\mathrm{pl}}_{H}$ be the coordinate projector onto the plane wedges $e_{F}$, $F\subset H$. Then, on $\Lambda^{4}\mathbb C^{H}$,
\begin{align}
 \sum_{H''}P^{(4)}_{H''}&=\Bigl(\frac m2-1\Bigr)\mathbb 1+\frac m2\,P^{\mathrm{pl}}_{H},
 \label{eq:res1}\\
 \sum_{H''}\Pi_{H''}&=P^{\mathrm{pl}}_{H}+(m-2)\,\Pi_{H} .
 \label{eq:res2}
\end{align}
Consequently the resolution~\eqref{eq:resfinal} holds.
\end{proposition}

\begin{proof}
Both sides of Eq.~\eqref{eq:res1} are diagonal in the wedge basis of $\Lambda^{4}\mathbb C^{H}$, and the entry at $e_{S}$ counts the hyperplanes of $H$ containing the affine span of $S$: a $j$-dimensional affine subspace of $H$ lies in $2^{k-j}-1$ of them. A plane ($j=2$) is counted $m-1$ times; a four-set that is not a plane has affine span of dimension exactly three (four distinct points of affine rank two would themselves form a plane) and is counted $m/2-1$ times; since $m-1=(m/2-1)+m/2$, this is Eq.~\eqref{eq:res1}.

For Eq.~\eqref{eq:res2}, fix a direction $L\le\mathrm{dir}H$ and restrict to the block $E_{L}$ spanned by the $m$ coordinates $e_{C_{v}}$, $C_{v}$ the planes of direction $L$ in $H$; distinct direction blocks are orthogonal and exhaust the supports of both sides beyond $P^{\mathrm{pl}}_{H}$'s diagonal. A hyperplane $H''$ meets the block iff $L\le\mathrm{dir}H''$, in which case the planes of direction $L$ inside $H''$ form an affine hyperplane $h$ of the $(k{-}2)$-dimensional quotient geometry, and the $L$-component of $\Pi_{H''}$ is the rank-one projector onto $\sqrt{2/m}\sum_{v\in h}e_{C_{v}}$, with all-plus coefficients by construction~\eqref{eq:localflat}. A point of the quotient lies on $m-1$ of its hyperplanes and a pair of distinct points on $m/2-1$, so on $E_{L}$ the left side of \eqref{eq:res2} has diagonal $\tfrac2m(m-1)$ and off-diagonal $\tfrac2m(\tfrac m2-1)$, which is exactly $\mathbb 1_{E_{L}}+(m-2)\ketbra{\psi^{H}_{L}}{\psi^{H}_{L}}$; summing over $L$ gives \eqref{eq:res2}. Combining, with $q_{H}=1/m$ and $2q_{H}=2/m$,
\begin{align}
 \sum_{H''}\Bigl(\frac2m P^{(4)}_{H''}-\Pi_{H''}\Bigr)
 &=\Bigl(1-\frac2m\Bigr)\mathbb 1+P^{\mathrm{pl}}_{H}-P^{\mathrm{pl}}_{H}-(m-2)\Pi_{H}\notag\\
 &=(m-2)\Bigl(\frac1m\,\mathbb 1-\Pi_{H}\Bigr),
\end{align}
which is Eq.~\eqref{eq:resfinal} since $4/(2^{k}-8)=1/(m-2)$.
\end{proof}

\begin{proof}[Proof of Theorem~\ref{thm:ocseven}]
We prove Eq.~\eqref{eq:mapcap} by induction on $k$; Lemma~\ref{lem:basecap} is the base. For $k\ge4$ and $A:E\to\mathbb C^{H}$ with $\dim E\le7$, congruence Eq.~\eqref{eq:resfinal} by $\Lambda^{4}A$. For each hyperplane $H''$, $\Pi_{H''}=\Pi_{H''}P^{(4)}_{H''}$ and $P^{(4)}_{H''}\Lambda^{4}A=\Lambda^{4}(P_{H''}A)$ by functoriality of $\Lambda^{4}$, with $P_{H''}$ the orthogonal projector $\mathbb C^{H}\to\mathbb C^{H''}$; the $H''$-summand is therefore
\begin{equation}
 \bigl(\Lambda^{4}(P_{H''}A)\bigr)^{\dagger}
 \bigl(q_{H''}\,\mathbb 1-\Pi_{H''}\bigr)
 \bigl(\Lambda^{4}(P_{H''}A)\bigr)\ \succeq\ 0
\end{equation}
by the induction hypothesis at level $k-1$ ($2q_{H}=q_{H''}$, and $P_{H''}A$ has rank at most seven). All summands of the congruenced \eqref{eq:resfinal} are positive semidefinite, proving \eqref{eq:mapcap} at level $k$. Taking $k=n$ and $A$ the inclusion $\mathcal T\hookrightarrow\mathbb C^{d}$ gives Eq.~\eqref{eq:ocseven}. Sharpness: $\PiF e_{F}=q\sum_{F'\parallel F}e_{F'}$ gives $\langle e_{F},\PiF e_{F}\rangle=q$ whenever $e_{F}\in\Lambda^{4}\mathcal T$.
\end{proof}

\section{Proof of the comparator certificate}
\label{app:witness}

This appendix proves Theorem~\ref{thm:witness}. The argument has three steps. The stabilizer purity transfers as a norm; the transfer budget is bounded by two unitary invariants; and the invariants are evaluated on the pure pair.

\emph{(i) Norm transfer.} Let $\rho,\sigma$ be Hermitian on $\mathbb{C}^{d^{2}}$ and $\Delta=\rho-\sigma$. Equation~\eqref{eq:pauliform} exhibits $\mathcal F(V;\cdot)^{1/4}=d^{-1/2}\|(\Tr(RU\cdot U^{\dagger}))_{R}\|_{\ell_{4}}$ as a norm of a linear image of the state. The reverse triangle inequality therefore bounds $|\mathcal F(V;\rho)^{1/4}-\mathcal F(V;\sigma)^{1/4}|$ by $[\frac1{d^{2}}\sum_{R}[\Tr(RU\Delta U^{\dagger})]^{4}]^{1/4}$, pointwise in $(V_{A},V_{B})$. Suppose the latter is at most $B^{1/4}$ for every unitary conjugate of $\Delta$. Then the two stabilizer purities are uniformly close over the local unitaries, and so are their maxima: $|\mathcal F_{*}(\rho)^{1/4}-\mathcal F_{*}(\sigma)^{1/4}|\le B^{1/4}$.

\emph{(ii) Moment bound.} Let $\Delta$ be Hermitian of rank at most two, with $\tau=\Tr\Delta^{2}$ and $\delta$ the product of its two nonzero eigenvalues $a,b$. Set $y_{R}=\Tr(R\Delta)$. Parseval gives $\sum_{R}y_{R}^{2}=d^{2}\tau$, and $|y_{R}|\le\|\Delta\|_{1}=|a|+|b|$ with $(|a|+|b|)^{2}=\tau+2|\delta|$. Hence $\frac1{d^{2}}\sum_{R}y_{R}^{4}\le\max_{R}y_{R}^{2}\cdot\frac1{d^{2}}\sum_{R}y_{R}^{2}\le\tau(\tau+2|\delta|)$. Both $\tau$ and $\delta$ are unitary invariants, so the same bound holds for every conjugate, as step (i) requires.

\emph{(iii) Pure-pair evaluation.} $\Delta_{\nu}=\ketbra{\Psi}{\Psi}-\ketbra{\phi_{\nu}}{\phi_{\nu}}$ is a difference of two rank-one operators, hence of rank at most two. In the basis with $\hat e_{1}=\ket{\Psi}$,
$\Delta_{\nu}=\bigl(\begin{smallmatrix}1-z&-\sqrt{z(c-z)}\\-\sqrt{z(c-z)}&-(c-z)\end{smallmatrix}\bigr)$,
whence $\tau_{\nu}=1+c_{\nu}^{2}-2z_{\nu}$ and $\delta_{\nu}=z_{\nu}-c_{\nu}\le0$; here $z_{\nu}\le c_{\nu}$ is Cauchy--Schwarz. Since $\tau_{\nu}+2|\delta_{\nu}|=(1+c_{\nu})^{2}-4z_{\nu}$, the bound of step (ii) is exactly $B^{\flat}_{\nu}$ of Eq.~\eqref{eq:Bflat}, and step (i) gives Eq.~\eqref{eq:transfer}.

For the certificate, exactness and degree-eight homogeneity give $\mathcal F_{*}(\nu)^{1/4}=c_{\nu}\,\mathcal F_{\mathrm{CB}}(\hat\nu)^{1/4}$. The upper side of Eq.~\eqref{eq:transfer} and the a priori bound $\mathcal F_{*}\le1$ give Eq.~\eqref{eq:Unu}. Taking the infimum over $\mathcal C$ gives Eq.~\eqref{eq:bracket}. \hfill$\qed$

\section{Stabilizer purity of dyadic-staircase spectra}
\label{app:shell}

Here, we prove Theorem~\ref{thm:shellpurity} and Theorem~\ref{thm:ceiling}. Throughout, $m_{k}:=|D_{k}|$ denotes the shell sizes ($m_{0}=1$, $m_{k}=2^{k-1}$ for $k\ge1$), shell indices run over $0\le k\le n$, and $\boldsymbol\nu$ is a dyadic-staircase spectrum with probabilities $\lambda_{i}=\nu_{i}^{2}=w_{k}/m_{k}$ for $i\in D_{k}$; normalization is not assumed, and every statement is degree-eight homogeneous.
We write $\mathrm{sh}(a)$ for the shell of a nonzero label $a\in\Ftwo^{n}$, that is, the unique $k\ge1$ with $a\in D_{k}$, which is one plus the position of the highest nonzero bit of $a$, and we use repeatedly that $\sum_{k}1/m_{k}\le3$ and $\sum_{k>r}1/m_{k}\le1/m_{r}$ for $r\ge1$.

\subsection{Channel decomposition and pivot formulas}

By the diagonal sector values of Eq.~\eqref{eq:diagsectors} and the equation following it, the CB value of any diagonal spectrum splits into three nonnegative Pauli channels,
\begin{equation}
   \mathcal F_{\mathrm{CB}}(\boldsymbol\nu)
   =P_{2}^{2}+3\sum_{\delta\ne0}\mathsf C_{\delta}^{2}+96\sum_{\dim L=2}s_{L}^{2},
   \label{eq:channels}
\end{equation}
where $P_{2}:=\sum_{i}\lambda_{i}^{2}$, $\mathsf C_{\delta}:= \sum_{x}\lambda_{x}\lambda_{x\oplus\delta}$, and $s_{L}:=\sum_{F\parallel L}\prod_{x\in F}\nu_{x}$ is the \emph{plane sum}, which runs over the $d/4$ planes $F$ of direction $L$ and multiplies the four amplitudes of each; in the notation of Appendix~\ref{app:sectors}, $\mathsf C_{\delta}=d\,w_{\delta}$ and $s_{L}=\tfrac d4 w_{L}$.
The correlation $\mathsf C_{\delta}$ sums over the pairs of labels $\{x,x\oplus\delta\}$ that differ by the nonzero label $\delta$, which we call the direction of the pair, and every direction lies in one shell $r=\mathrm{sh}(\delta)\ge1$.
The plane sum $s_{L}$ is attached to a two-dimensional subspace $L\le\Ftwo^{n}$, which has a unique basis $\{g_{1},g_{2}\}$ whose highest nonzero bits lie in distinct shells $a:=\mathrm{sh}(g_{1})>b:=\mathrm{sh}(g_{2})\ge1$ and in which $g_{1}$ vanishes at the highest bit of $g_{2}$; we call $(a,b)$ the \emph{pivot pair} of $L$.
Counting the bits of $g_{1}$ and $g_{2}$ that remain free, the set $\mathsf L_{ab}$ of two-dimensional subspaces with pivot pair $(a,b)$ has exactly $m_{a}m_{b}/2$ elements.
For staircase spectra all three channels depend only on these shell data.

\begin{proposition}[Pivot formulas]
  \label{prop:pivot}
  For a staircase spectrum, the correlation $\mathsf C_{\delta}$ depends on the direction $\delta$ only through its shell $r=\mathrm{sh}(\delta)$,
  \begin{equation}
   \mathsf C_{\delta}=\mathsf C_{r}
   :=\frac{2w_{r}}{m_{r}}\sum_{j<r}w_{j}
   +\sum_{k>r}\frac{w_{k}^{2}}{m_{k}},
   \qquad \mathrm{sh}(\delta)=r,
   \label{eq:linepivot}
\end{equation}
and the plane sum $s_{L}$ depends on $L$ only through its pivot pair, $s_{L}=\mathsf B_{ab}$ for $L\in\mathsf L_{ab}$, with
\begin{align}
   \mathsf B_{ab}
   :={}&\frac14\sum_{k>a}\frac{w_{k}^{2}}{m_{k}}
   +\frac{w_{a}}{2m_{a}}\sum_{b<t<a}w_{t}\notag\\
   &+\frac{w_{a}}{m_{a}}\sqrt{\frac{w_{b}}{m_{b}}}\sum_{j<b}\sqrt{m_{j}w_{j}} .
   \label{eq:planepivot}
\end{align}
\end{proposition}

\begin{proof}
\emph{Lines.} Classify the pairs in $\mathsf C_{\delta}=\sum_{x}\lambda_{x}\lambda_{x\oplus\delta}$ by $\mathrm{sh}(x)$. For $\mathrm{sh}(x)=k>r$ the top bit of $x$ is untouched, so $x\oplus\delta\in D_{k}$: these $m_{k}$ terms contribute $m_{k}(w_{k}/m_{k})^{2}$. For $\mathrm{sh}(x)=j<r$ the bit $r{-}1$ is switched on, so $x\oplus\delta\in D_{r}$: contribution $m_{j}(w_{j}/m_{j})(w_{r}/m_{r})$. For $\mathrm{sh}(x)=r$ the bit $r{-}1$ is switched off, and as $x$ runs over $D_{r}$ the partner $x\oplus\delta$ runs bijectively over $[0,m_{r})$: contribution $(w_{r}/m_{r})\sum_{j<r}m_{j}(w_{j}/m_{j})$, equal to the previous group. Summing proves Eq.~\eqref{eq:linepivot}; no lower bit of $\delta$ entered.

\emph{Planes.} Write $s_{L}=\tfrac14\sum_{x}\prod_{y\in x\oplus L}\nu_{y}$ over the planes $\{x,x{\oplus}g_{1},x{\oplus}g_{2},x{\oplus}g_{1}{\oplus}g_{2}\}$. The bits above position $a{-}1$ are common to all four points of a plane. If they are nonzero, all four points share one shell $s>a$, giving $\tfrac14\sum_{s>a}m_{s}(w_{s}/m_{s})^{2}$: the first term of Eq.~\eqref{eq:planepivot}. Otherwise the plane splits into the two points containing $g_{1}$'s pivot bit---both in shell $a$, since no higher bit is set---and a low pair $\{y,y\oplus g_{2}\}$ with $y\in[0,m_{a})$; each such plane is determined by its low pair, a two-to-one parametrization by $y$. Hence the remaining contribution is
\begin{equation*}
 \frac{w_{a}}{2m_{a}}\sum_{y<m_{a}}\nu_{y}\nu_{y\oplus g_{2}},
\end{equation*}
and the inner sum is a line correlation below shell $a$, classified as before: pairs with $\mathrm{sh}(y)=t\in(b,a)$ stay in shell $t$ and give $\sum_{b<t<a}w_{t}$, while pairs joining a shell $j<b$ to shell $b$ give $2\sqrt{w_{b}/m_{b}}\sum_{j<b}m_{j}\sqrt{w_{j}/m_{j}}$. This is Eq.~\eqref{eq:planepivot}; again no lower bit of $g_{1},g_{2}$ entered, proving constancy on the pivot class.
\end{proof}

\subsection{Proof of Theorem~\ref{thm:shellpurity}}

\emph{Lower bound.} By Eq.~\eqref{eq:Phimatel} the compressed operator of a diagonal spectrum is diagonal in the $\Phi$-basis with nonnegative entries $S(u,v)/d$, so $\mathcal F_{\mathrm{CB}}=\sum_{u,v}S(u,v)^{2}$ with $S(u,v):=\sum_{x}\nu_{x}\nu_{x\oplus u}\nu_{x\oplus v}\nu_{x\oplus u\oplus v}$. 
For $u,v\in[0,m_{s})$ and $x\in D_{s}$, the plane $\{x,x\oplus u,x\oplus v,x\oplus u\oplus v\}$ lies entirely in $D_{s}$, since $u$ and $v$ only change bits below the leading bit of $x$, so $S(u,v)\ge\sum_{s:\,m_{s}>\max(u,v)}w_{s}^{2}/m_{s}$; since $(\sum_{s}a_{s})^{2}\ge\sum_{s}a_{s}^{2}$ for nonnegative $a_{s}$,
\begin{equation*}
 \mathcal F_{\mathrm{CB}}\ \ge\ \sum_{u,v}\ \sum_{s:\,m_{s}>\max(u,v)}\frac{w_{s}^{4}}{m_{s}^{2}}
 =\sum_{s}\frac{w_{s}^{4}}{m_{s}^{2}}\cdot m_{s}^{2}
 =\sum_{s}w_{s}^{4}.
\end{equation*}

\emph{Upper bound.} Write $\sigma:=\sum_{k}w_{k}^{4}$ and bound the three channels of Eq.~\eqref{eq:channels} through Proposition~\ref{prop:pivot}. All estimates below use only the Cauchy--Schwarz inequality, the arithmetic--geometric mean inequality on fourth powers, and geometric summability of the shell sizes.

\emph{(i) Purity channel.} By Cauchy--Schwarz,
$P_{2}^{2}=\bigl(\sum_{k}w_{k}^{2}/m_{k}\bigr)^{2}\le\bigl(\sum_{k}1/m_{k}^{2}\bigr)\sigma=\tfrac73\sigma$, since $\sum_{k}m_{k}^{-2}=1+1+\tfrac14+\tfrac1{16}+\dots=\tfrac73$.

\emph{(ii) Line channel.} With $A_{r}:=\sum_{j<r}w_{j}$, $X_{r}:=2w_{r}A_{r}/m_{r}$, and $T_{r}:=\sum_{k>r}w_{k}^{2}/m_{k}$, the channel is $\|X+T\|^{2}$ in the weighted norm $\|f\|^{2}:=3\sum_{r}m_{r}f_{r}^{2}$, and Minkowski gives $\|X+T\|\le\|X\|+\|T\|$. For $T$, Cauchy--Schwarz with weights $1/m_{k}$ and $\sum_{k>r}1/m_{k}=1/m_{r}$ give $m_{r}T_{r}^{2}\le\sum_{k>r}w_{k}^{4}/m_{k}$, and summing over $r\ge1$ counts each $k$ at most $(k-1)/m_{k}\le\tfrac12$ times: $\|T\|^{2}\le\tfrac32\sigma$. For $X$, $A_{r}^{2}\le r\sum_{j<r}w_{j}^{2}$ and $w_{r}^{2}w_{j}^{2}\le\tfrac12(w_{r}^{4}+w_{j}^{4})$ give
\begin{equation*}
 \sum_{r}\frac{w_{r}^{2}A_{r}^{2}}{m_{r}}
 \le\frac12\sum_{r}\frac{r^{2}}{m_{r}}w_{r}^{4}
 +\frac12\sum_{j}w_{j}^{4}\sum_{r>j}\frac{r}{m_{r}}
 \le\Bigl(\frac98+2\Bigr)\sigma,
\end{equation*}
using $\sup_{r}r^{2}/m_{r}=\tfrac94$ and $\sum_{r\ge1}r/m_{r}=4$, so $\|X\|^{2}\le12\cdot\tfrac{25}8\sigma=\tfrac{75}2\sigma$. Since $\sqrt{75/2}+\sqrt{3/2}=6\sqrt{3/2}$, the line channel is at most $54\sigma$.

\emph{(iii) Plane channel.} With the class multiplicities, the channel is $\|\mathsf B_{1}+\mathsf B_{2}+\mathsf B_{3}\|_{\mathrm{pl}}^{2}$ for the three terms of Eq.~\eqref{eq:planepivot}, in the weighted norm $\|X\|_{\mathrm{pl}}^{2}:=48\sum_{a>b\ge1}m_{a}m_{b}X_{ab}^{2}$; Minkowski bounds it by $\bigl(\|\mathsf B_{1}\|_{\mathrm{pl}}+\|\mathsf B_{2}\|_{\mathrm{pl}}+\|\mathsf B_{3}\|_{\mathrm{pl}}\bigr)^{2}$. For $\mathsf B_{1}=\tfrac14T_{a}$: $\sum_{b<a}m_{b}\le m_{a}$ and, by Cauchy--Schwarz with the geometric weights $m_{a}/m_{k}=2^{a-k}$,
$(m_{a}T_{a})^{2}\le\sum_{k>a}2^{a-k}w_{k}^{4}$, whose $a$-sum is at most $\sigma$: $\|\mathsf B_{1}\|_{\mathrm{pl}}^{2}\le3\sigma$. For $\mathsf B_{2}$: Cauchy--Schwarz gives $\bigl(\sum_{b<t<a}w_{t}\bigr)^{2}\le(i-1)\sum_{b<t<a}w_{t}^{2}$ with $i:=a-b$, so
\begin{equation*}
 \|\mathsf B_{2}\|_{\mathrm{pl}}^{2}
 \le12\sum_{a>b}2^{-i}(i-1)\!\!\sum_{b<t<a}\!\!\frac{w_{a}^{4}+w_{t}^{4}}2
 \le6\bigl(3+3\bigr)\sigma=36\sigma,
\end{equation*}
since $\sum_{i\ge1}(i-1)^{2}2^{-i}=3$ both for the $w_{a}^{4}$ count and, after resumming the double gap, for the $w_{t}^{4}$ count. For $\mathsf B_{3}$: splitting $\sqrt{m_{j}w_{j}}=\bigl[m_{j}2^{(b-j)/2}\bigr]^{1/2}\bigl[w_{j}2^{-(b-j)/2}\bigr]^{1/2}$ in Cauchy--Schwarz and using the exact edge-sensitive sum $\sum_{j<b}m_{j}2^{(b-j)/2}=m_{b}\bigl(1+\sqrt2-2^{-(b-1)/2}\bigr)\le(1+\sqrt2)\,m_{b}$,
\begin{equation*}
 \Bigl(\sum_{j<b}\sqrt{m_{j}w_{j}}\Bigr)^{2}
 \le(1+\sqrt2)\,m_{b}\sum_{j<b}w_{j}\,2^{-(b-j)/2},
\end{equation*}
so $\|\mathsf B_{3}\|_{\mathrm{pl}}^{2}\le48(1+\sqrt2)\sum_{a>b}\sum_{j<b}2^{-(a-b)}2^{-(b-j)/2}\,w_{a}^{2}w_{b}w_{j}$. By the arithmetic--geometric mean inequality, $w_{a}^{2}w_{b}w_{j}\le\tfrac12w_{a}^{4}+\tfrac14w_{b}^{4}+\tfrac14w_{j}^{4}$, and each of the three resulting quartic sums carries the doubly geometric kernel, whose total over the two gap variables is $(\sum_{i\ge1}2^{-i})(\sum_{l\ge1}2^{-l/2})=1+\sqrt2$: hence $\|\mathsf B_{3}\|_{\mathrm{pl}}^{2}\le48(1+\sqrt2)^{2}\sigma$. The plane channel is therefore at most $\bigl[\sqrt3+6+4\sqrt3(1+\sqrt2)\bigr]^{2}\sigma$.

Collecting, $\mathcal F_{\mathrm{CB}}\le\bigl[\tfrac73+54+\bigl(\sqrt3+6+4\sqrt3(1+\sqrt2)\bigr)^{2}\bigr]\sigma=C_{\mathrm{sh}}\,\sigma$. \hfill$\qed$

\subsection{Proof of Theorem~\ref{thm:ceiling}}

Let $\boldsymbol\mu$ be sorted and normalized, with probabilities $\lambda_{i}$ ($0$-indexed). Define the \emph{lower staircase} $\check\lambda_{i}:=\lambda_{2^{k}-1}$ for $i\in D_{k}$: entrywise $\check\lambda\le\lambda$, and its mass obeys the condensation bound
\begin{equation}
 \sum_{k}\check w_{k}\ \ge\ \tfrac12,
 \qquad \check w_{k}:=m_{k}\lambda_{2^{k}-1},
 \label{eq:condense}
\end{equation}
because $\sum_{i\in D_{k+1}}\lambda_{i}\le m_{k+1}\lambda_{2^{k}-1}=2\check w_{k}$ while the two lowest shells are copied verbatim. By the CB monotonicity established in Sec.~\ref{ssec:partners}, the lower bound of Theorem~\ref{thm:shellpurity}, and the power-mean inequality over the $n+1$ shells,
\begin{equation*}
 \mathcal F_{\mathrm{CB}}(\boldsymbol\mu)
 \ \ge\ \sum_{k}\check w_{k}^{4}
 \ \ge\ \frac{\bigl(\sum_{k}\check w_{k}\bigr)^{4}}{(n+1)^{3}}
 \ \ge\ \frac1{16\,(n+1)^{3}} .
\end{equation*}
For the entropy form, truncate first: with $s:=\lceil2S_{1}\rceil+1$ and $R:=2^{s}$, sortedness gives $\lambda_{i}\le1/(i+1)$, hence $\log_{2}(1/\lambda_{i})\ge s$ for $i\ge R-1$ and $\sum_{i\ge R-1}\lambda_{i}\le S_{1}/s<\tfrac12$: the rank-$R$ head, an entrywise minorant, carries mass greater than $\tfrac12$ and occupies the shells $0,\dots,s$. Applying the display above to the head, with its condensation losing another factor two and $s+1\le2S_{1}+3$ shells,
\begin{equation*}
 \mathcal F_{\mathrm{CB}}(\boldsymbol\mu)\ \ge\ \frac{(1/4)^{4}}{(2S_{1}+3)^{3}}
 =\frac1{256\,(2S_{1}+3)^{3}} . \tag*{\qed}
\end{equation*}

\subsection{Binomial spectra of imperfect dimers}
\label{app:dimer}

This subsection proves Theorem~\ref{thm:dimer}. Throughout, $q=1-p\in(0,\frac12)$ is fixed and constants depend only on $p$. The Schmidt probabilities of $\ket{\Psi_{m,p}}$ are the levels
\begin{equation}
 \lambda^{(j)}=p^{m-j}q^{j},
 \qquad\text{multiplicity }\binom mj,
 \qquad j=0,\dots,m,
 \label{eq:binomlevels}
\end{equation}
already sorted, since $q<p$. Let $b_j=\binom mj\,p^{m-j}q^{j}$ be the total mass of level $j$, the $\mathrm{Bin}(m,q)$ weight, and let $R_j=\sum_{\ell\le j}\binom m\ell$ be the rank at the end of level $j$.

\emph{Shell-max bound.} Fix $k\ge1$ and let $j$ be the level containing the label $2^{k-1}$, the top of shell $k$, so that $2^{k-1}<R_j$. Then
\begin{equation}
 \hat w_k
 =2^{k-1}\lambda^{(j)}
 \le R_j\lambda^{(j)}
 =\sum_{\ell\le j}b_\ell\Bigl(\frac qp\Bigr)^{j-\ell}
 \le\frac{p}{p-q}\,\max_\ell b_\ell,
 \label{eq:dimermax}
\end{equation}
and $\hat w_0=b_0$. A standard Stirling estimate gives $\max_\ell b_\ell\le(pqm)^{-1/2}$, so every shell obeys $\hat w_k=O_p(m^{-1/2})$. With the universal bound $\sum_k\hat w_k\le2$ of Sec.~\ref{ssec:value},
\begin{equation}
 \sum_k\hat w_k^{4}
 \le2\bigl(\max_k\hat w_k\bigr)^{3}
 =O_p(m^{-3/2}).
 \label{eq:dimerupper}
\end{equation}

\emph{Shell count of the bulk.} The level of a spectral sample is $\mathrm{Bin}(m,q)$, with variance $pqm$. By Chebyshev, the window $j_1\le j\le j_2$ with $|j-qm|\le2\sqrt{pqm}$ carries mass at least $\frac34$. For consecutive levels, $R_j/R_{j-1}=1+\binom mj/R_{j-1}\le1+\binom mj/\binom m{j-1}=1+\frac{m-j+1}{j}$, which is bounded by $1+3/q$ on the window once $m\ge m_p$. The window therefore spans at most
\begin{equation}
 \log_2\frac{R_{j_2}}{R_{j_1-1}}
 \le\bigl(4\sqrt{pqm}+2\bigr)\log_2\Bigl(1+\frac3q\Bigr)
 =O_p(\sqrt m)
 \label{eq:dimerspan}
\end{equation}
in logarithmic rank, and hence meets $K=O_p(\sqrt m)$ dyadic shells $\mathcal K$. The shell-max masses dominate the true shell masses, so $\sum_{k\in\mathcal K}\hat w_k\ge\frac34$, and the power-mean inequality gives
\begin{equation}
 \sum_k\hat w_k^{4}
 \ \ge\ \sum_{k\in\mathcal K}\hat w_k^{4}
 \ \ge\ \frac{(3/4)^{4}}{K^{3}}
 \ =\ \Omega_p(m^{-3/2}).
 \label{eq:dimerlower}
\end{equation}

\emph{Assembly.} Equations~\eqref{eq:dimerupper} and~\eqref{eq:dimerlower} give $\sum_k\hat w_k^{4}=\Theta_p(m^{-3/2})$, and Eq.~\eqref{eq:shellparticipation} brackets both $\mathcal M_{\mathrm{NL}}$ and $\mathcal M_{\mathrm{CB}}$ between $\frac32\log_2m-O_p(1)$ and $\frac32\log_2m+O_p(1)$. \hfill$\qed$

\section{Numerical evaluation of the nonlocal SRE of the Ising chain}
\label{app:numerics}

In this appendix, we describe how we obtain the single-particle entanglement spectra used to evaluate the quantities shown in Fig.~\ref{fig:ising}.
The ground state of the open chain of Eq.~\eqref{eq:tfim} is a fermionic Gaussian state, so its reduced density matrix is determined by the mode occupations $p_j$~\cite{PeschelEisler2009}, from which the Schmidt probabilities follow.
We describe three approaches: direct diagonalization for $N\le10^{4}$ at every $g$, a closed-form correlation kernel at criticality used up to $N=3\times10^{6}$, and a fitted entanglement-energy ladder that extends the calculations to larger critical chains. The linear-algebra routines are those of SciPy~\cite{Virtanen2020SciPy}.

\subsection{Single-particle entanglement spectrum}
\label{app:numerics-modes}

\textit{Exact spectra.---}
The Jordan--Wigner transformation~\cite{LiebSchultzMattis1961,Pfeuty1970} maps the open chain of Eq.~\eqref{eq:tfim} to a quadratic form in $2N$ Majorana operators $\gamma_{1},\dots,\gamma_{2N}$ with $\{\gamma_{m},\gamma_{n}\}=2\delta_{mn}$,
\begin{equation}
 H=i\sum_{m<n}A_{mn}\gamma_{m}\gamma_{n},
 \qquad
 A_{2i-1,2i}=g,\quad A_{2i,2i+1}=1,
 \label{eq:majoranaH}
\end{equation}
Up to a diagonal unitary transformation, the Hermitian matrix $iA$ is the real symmetric tridiagonal matrix $T$, which we diagonalize with a tridiagonal eigensolver.
The ground state occupies the single-particle levels of negative energy, and its covariance matrix $\Gamma_{mn}=\frac{i}{2}\langle[\gamma_{m},\gamma_{n}]\rangle$ follows from the spectral decomposition of $T$~\cite{PeschelEisler2009}.
The eigenvalues of the block of $i\Gamma$ on the first $2\ell$ Majorana operators, $\ell=N/2$, come in pairs $\pm\nu_{j}$ with $0\le\nu_{j}\le1$, and the occupations of the modes of the half chain are $p_{j}=(1+\nu_{j})/2$~\cite{PeschelEisler2009}, from which $\mathcal M_{\mathrm{FNL}}$ of Eq.~\eqref{eq:mfnl} follows directly.
The computational cost is dominated by forming the $N\times N$ block of $\Gamma$ from the dense eigenvectors of $T$ and by diagonalizing it, both $O(N^{3})$ operations.

\textit{Critical spectra.---}
At $g=1$ we use instead the explicit form of the correlation matrix of the open chain.
In the ordering $(\gamma_{1},\gamma_{3},\dots;\gamma_{2},\gamma_{4},\dots)$ the matrix $A$ is block off-diagonal with the $N\times N$ block $M=g\mathbb 1-S$, where $S$ is the shift matrix, and the correlations between the two species of Majorana operators are $G=VU^{\sf T}$, with $M=U\Sigma V^{\sf T}$ the singular value decomposition of $M$~\cite{LiebSchultzMattis1961,PeschelEisler2009}. The occupations of the half chain are $p_{j}=(1+\nu_{j})/2$ with $\nu_{j}$ the singular values of the block $G_{\ell}=(G_{mn})_{1\le
  m,n\le\ell}$.
For $g=1$ the singular value decomposition of $M$ is known exactly for every $N$~\cite{Pfeuty1970}: $MM^{\sf T}$ is the tridiagonal matrix with entries $(-1,2,-1)$ and corner entry $(MM^{\sf T})_{NN}=1$, so the singular values are $2\sin(\theta_{k}/2)$ with $\theta_{k}=(2k-1)\pi/(2N+1)$, and the sum over $k$ in $G=VU^{\sf T}$ can be performed in closed form,
  \begin{equation}
  \begin{gathered}
   G_{mn}=\frac{2}{2N+1}\Bigl[\tau\bigl(m+n-\tfrac12\bigr)+\tau\bigl(n-m+\tfrac12\bigr)\Bigr],\\
   \tau(x)=\frac{\sin^{2}(N\phi_{x})}{\sin\phi_{x}},
   \qquad
   \phi_{x}=\frac{\pi x}{2N+1}.
  \end{gathered}
   \label{eq:kernel}
\end{equation}
The first term depends on $m+n$ and the second on $n-m$, so $G_{\ell}$ is the sum of a Hankel and a Toeplitz matrix, the structure that an open boundary imposes on the correlation matrices of free-fermion chains~\cite{FagottiCalabrese2011,Deift2011Toeplitz}, and its action on a vector consists of two convolutions, which we evaluate by fast Fourier transforms in $O(N\log N)$ operations without storing any matrix.
We compute the largest eigenvalues $\eta_{j}=1-\nu_{j}^{2}$ of $\mathbb 1-G_{\ell}^{\sf T}G_{\ell}$ with the implicitly restarted Lanczos method of ARPACK and retain the modes with $\eta_{j}\ge10^{-12}$. The others have $p_{j}=1$ to double precision and do not contribute to the considered quantities. The number of retained modes grows only logarithmically with $N$, to less than a few tens for $N=10^{7}$.
We use this approach for $N \le 3\times10^{6}$ in Fig.~\ref{fig:ising}(c).

\textit{The entanglement-energy ladder.---}
The occupations $p_{j}$ define the single-particle entanglement energies $\varepsilon_{j}=\ln[(1+\nu_{j})/(1-\nu_{j})]$, so that $p_{j}=\frac12[1+\tanh(\varepsilon_{j}/2)]$, and the Schmidt
  probabilities are the products over the occupation patterns $\boldsymbol n\in\{0,1\}^{m}$ of the $m$ retained modes,
  \begin{equation}
   \lambda_{\boldsymbol n}=\prod_{j}p_{j}^{1-n_{j}}(1-p_{j})^{n_{j}}
   =\lambda_{\boldsymbol 0}\,e^{-\sum_{j}n_{j}\varepsilon_{j}},
   \qquad
   \lambda_{\boldsymbol 0}=\prod_{j}p_{j},
   \label{eq:productspectrum}
  \end{equation}
  so that the $\varepsilon_{j}$ are the single-particle levels of the entanglement Hamiltonian $-\ln\rho_{A}$, which is quadratic in the fermions of the half chain~\cite{PeschelEisler2009}.
  For the critical chain these levels are approximately equidistant, with a spacing proportional to $1/\ln N$~\cite{PeschelKaulkeLegeza1999,ChungPeschel2001,PeschelEisler2009}.
  We define $b\equiv\pi^{2}/\varepsilon_{1}$, for which the exact chains give $b-\ln N=2.83$ at $N=10^{3}$, increasing to $2.88$ at $N=1.2\times10^{7}$.
  The equidistant ladder $\varepsilon_{j}=\pi^{2}(2j-1)/b$ is not accurate enough for our purpose: at $N=10^{6}$ the exact levels deviate from it by $1\%$ at $j=2$ and by $15\%$ at $j=20$,
  which shifts $\mathcal M_{\mathrm{FNL}}$ by $0.04$.

We therefore write the levels as the known ladder times a correction that we expand in the small parameter $1/b$, keeping two terms,
  \begin{equation}
   \varepsilon_{j}
   =\frac{\pi^{2}(2j-1)}{b}
   \Bigl[R(x_{j})+\frac{S(x_{j})}{b}\Bigr],
   \qquad
   x_{j}=\frac{2j-1}{b},
   \label{eq:ladder}
  \end{equation}
with the assumption that $R$ and $S$ depend on $j$ and $N$ only through $x_{j}$.
We determine $R$ and $S$ from the exact ladders at $N=10^{3},3\times10^{3},\dots,10^{6}$ by interpolating the scaled energies $\varepsilon_{j}b/[\pi^{2}(2j-1)]$ onto a grid of $64$ points in $0.055\le x\le2.4$ and fitting, at every grid point, the seven values by $R(x)+S(x)/b$, with $R(0)=1$ and $S(0)=0$ fixed by the definition of $b$.
We test this assumption finding that it reproduces the exact $\mathcal M_{\mathrm{FNL}}$ and $\mathcal M_{\mathrm{CB}}$ from the corner matrix approach at $N=3\times10^{6}$ and $1.2\times10^{7}$ to $2\times10^{-4}$.
Modes with $\varepsilon_{j}>40$ have $p_{j}=1$ to double precision and are dropped.
Beyond these sizes, Eq.~\eqref{eq:ladder} is an extrapolation, which we use for $b$ from $18.5$ to $68$ with the size label $N\simeq e^{b-2.88}$, up to $N\simeq10^{28}$.
Hence, the points beyond $N=1.2\times10^{7}$ in Fig.~\ref{fig:ising}(c) are exact evaluations on this modeled spectrum, whose accuracy relies only on the above test.


%

\end{document}